\documentclass{article}

\usepackage{amsmath}
\usepackage{amsfonts}
\usepackage{amssymb}
\usepackage{graphicx}
\usepackage{tikz}
\usepackage{esvect}
\usepackage[ruled]{algorithm2e}
\usepackage{caption}
\usepackage{subcaption}
\usepackage[left= 3 cm,right=3 cm,top=3 cm,bottom=3 cm]{geometry}
\usepackage{amsthm}
\usepackage{authblk}
\usepackage{mathtools}
\usepackage[colorlinks = true, citecolor = {blue}]{hyperref}
\usepackage{array}

\usetikzlibrary{arrows}

\DeclareMathOperator*{\argmax}{argmax}

\newtheorem{theorem}{Theorem}
\newtheorem{definition}{Definition}
\newtheorem{lemma}{Lemma}

\newtheorem{corollary}{Corollary}

\newcommand{\set}[1]{\left\{ #1 \right\}}
\newcommand{\card}[1]{\left| #1 \right|}
\newcommand{\ith}[1]{#1^{\mbox{\scriptsize{th}}}}

\newcommand{\ifend}{\textbf{endif}}

\newcommand{\diam}{\mbox{\textsf{diam}}}

\newcommand{\ecc}{\mbox{\textsf{ecc}}}
\newcommand{\opp}{\mbox{\textsf{op}}}
\newcommand{\rc}{\mbox{RC}}
\newcommand{\paradj}{\parallel_{\mbox{\scriptsize{P}}}}

\newcommand{\paradji}{\parallel_{\scriptsize{\Theta}}^{-1}}
\newcommand{\perpadj}{\perp_{\scriptsize{\Theta}}}
\newcommand{\paradjm}{\parallel_{\mbox{\scriptsize{Pm}}}}

\newcommand{\hul}{H_u^{\mbox{\scriptsize{L}}}}
\newcommand{\hual}{H_u^{\mbox{\scriptsize{AL}}}}
\newcommand{\varphis}{\varphi_{\subseteq}}
\newcommand{\psis}{\psi_{\supseteq}}
\newcommand{\varphim}{\varphi_{\mbox{\scriptsize{m}}}}

\newcommand{\mcalc}{\mathcal{C}}
\newcommand{\mcalcm}{\mathcal{C}_{\mbox{\scriptsize{max}}}}
\newcommand{\trp}{\mbox{Trp}}
\newcommand{\huv}{H_{u \rightarrow v}}
\newcommand{\hvu}{H_{v \rightarrow u}}

\title{Subquadratic-time algorithm for the diameter and all eccentricities on median graphs}
\author[1,4]{Pierre Berg\'e}
\author[2,3]{Guillaume Ducoffe}
\author[4]{Michel Habib}
\affil[1]{LIMOS, CNRS, Universit\'e Clermont Auvergne, France, \texttt{pierre.berge@uca.fr}}
\affil[2]{National Institute for Research and Development in Informatics, Romania}
\affil[3]{University of Bucharest, Romania}
\affil[4]{IRIF, CNRS, Universit\'e de Paris, France}
\date{}

\begin{document}

\maketitle

\abstract{
On sparse graphs, Roditty and Williams [2013] proved that no $O(n^{2-\varepsilon})$-time algorithm achieves an approximation factor smaller than $\frac{3}{2}$ for the diameter problem unless SETH fails.
In this article, we solve an open question formulated in the literature: can we use the structural properties of median graphs to break this global quadratic barrier?

We propose the first combinatiorial algorithm computing exactly all eccentricities of a median graph in truly subquadratic time. Median graphs constitute the family of graphs which is the most studied in metric graph theory because their structure represents many other discrete and geometric concepts, such as CAT(0) cube complexes. Our result generalizes a recent one, stating that there is a linear-time algorithm for all eccentricities in median graphs with bounded dimension $d$, {\em i.e.} the dimension of the largest induced hypercube. This prerequisite on $d$ is not necessarily anymore to determine all eccentricities in subquadratic time. The execution time of our algorithm is $O(n^{1.6408}\log^{O(1)} n)$.

We provide also some satellite outcomes related to this general result. In particular, restricted to simplex graphs, this algorithm enumerates all eccentricities with a quasilinear running time. Moreover, an algorithm is proposed to compute exactly all reach centralities in time $O(2^{3d}n\log^{O(1)}n)$.
}

\section{Introduction} \label{sec:intro}

Median graphs can be certainly identified as the most important family of graphs in metric graph theory. Indeed, they are related to numerous areas: universal algebra~\cite{Av61,BiKi47}, CAT(0) cube complexes~\cite{BaCh08,Ch00}, abstract models of concurrency~\cite{BaCo93,SaNiWi93}, and genetics~\cite{BaQuSaMa02,BaFoSyRi95}. Let $d(a,b)$ be the length ({\em i.e.} number of edges) of the shortest $(a,b)$-path for $a,b \in V$ and $I(a,b)$ be the set made up of all vertices $u$ metrically between $a$ and $b$, {\em i.e.} $d(a,b) = d(a,u) + d(u,b)$. Median graphs are the graphs such that for any triplet of distinct vertices $x,y,z \in V$, the intersection $I(x,y) \cap I(y,z) \cap I(z,x)$ is a singleton, containing the \textit{median} of this triplet, denoted by $m(x,y,z)$.

The purpose of this article is to break the quadratic barrier for the computation time of certain metric parameters on median graphs. In particular, we focus on one of the most fundamental problems in algorithmic graph theory related to distances: the \textit{diameter}. Given an undirected graph $G=(V,E)$, the diameter is the maximum distance $d(u,v)$ for all $u,v \in V$. Two vertices at maximum distance form a \textit{diametral pair}. An even more general problem consists in determining all eccentricities of the graph. The eccentricity $\ecc(v)$ of a vertex $v$ is the maximum length of a shortest path starting from $v$: $\ecc(v) = \max_{w \in V} d(v,w)$. The diameter is thus the maximum eccentricity.

\subsection{State of the art}

Executing a \textit{Breadth First Search} (BFS) from each vertex of an input graph $G$ suffices to obtain its eccentricities in $O(n\card{E})$, with $n = \card{V}$. As median graphs are relatively sparse, $\card{E} \le n\log n$, these multiple BFSs compute all eccentricities in time $O(n^2\log n)$ for this class of graphs.
Very efficient algorithms determining the diameter already exist on other classes of graphs, for example~\cite{AbWiWa16,Ca17,DuHaVi20}. Many works have also been devoted to approximation algorithms for this parameter. Chechik {\em et al.}~\cite{ChLaRoScTaWi14} showed that the diameter can be approximated within a factor $\frac{3}{2}$ in time $\tilde{O}(m^{\frac{3}{2}})$ on general graphs. On sparse graphs, it was shown in~\cite{RoWi13} that no $O(n^{2-\varepsilon})$-time algorithm can achieve an approximation factor smaller than $\frac{3}{2}$ for the diameter unless the Strong Exponential Time Hypothesis (SETH) fails.

Median graphs are bipartite and can be isometrically embedded into hypercubes. They are the 1-skeletons of CAT(0) cube complexes~\cite{Ch00} and the domains of event structures~\cite{BaCo93}. They admit structural properties, such as the Mulder's convex expansion~\cite{Mu78,Mu80}. They are strongly related to hypercubes retracts~\cite{Ba84}, Cartesian products and gated amalgams~\cite{BaCh08}, but also Helly hypergraphs~\cite{MuSc79}. They do not contain induced $K_{2,3}$, otherwise a triplet of vertices would admit at least two medians. The \textit{dimension} $d$ of a median graph $G$ is the dimension of its largest induced hypercube. The value of this parameter is at most $\lfloor \log n \rfloor$ and meets this upper bound when $G$ is an hypercube. Moreover, parameter $d$ takes part in the sparsity of median graphs: $\card{E} \le dn$.

An important concept related to median graphs is the equivalence relation $\Theta$. This is the reflexive and transitive closure of relation $\Theta_0$, where two edges are in $\Theta_0$ if they are opposite in a common 4-cycle. A $\Theta$-\textit{class} is an equivalence class of $\Theta$. Each $\Theta$-class of a median graph forms a matching cutset, splitting the graph into two convex connected components, called \textit{halfspaces}. The number $q \le n$ of $\Theta$-classes corresponds to the dimension of the hypercube in which the median graph $G$ isometrically embeds.  Value $q$ satisfies the Euler-type formula $2n-m-q\le 2$~\cite{KlMuSk98}. A recent LexBFS-based algorithm~\cite{BeChChVa20} identifies the $\Theta$-classes in linear time $O(\card{E})=O(dn)$.

Two subquadratic-time algorithms have been proposed for the recognition of median graphs. Using convex characterizations of halfspaces, Hagauer {\em et al.}~\cite{HaImKl99} showed that median graphs can be recognized in $O(n^{\frac{3}{2}}\sqrt{n})$. In~\cite{ImKlMu99}, a bijection between median and triangle-free graphs make the recognition algorithms for triangle-free graphs work on median ones~\cite{AlYuZw97}. Hence, median graphs can be recognized in $O((n\log^2 n)^{1.41})$ using this reduction.

There exist efficient algorithms for some metric parameters on median graphs. For example, the median set and the Wiener index can be determined in $O(\card{E})$~\cite{BeChChVa20}. 
Subfamilies of median graphs have also been studied. There is an algorithm computing the diameter and the radius in linear time for squaregraphs~\cite{ChDrVa02}. A more recent contribution introduces a quasilinear time algorithm - running in $O(n\log^{O(1)} n)$ - for the diameter on cube-free median graphs~\cite{Du20}, using distance and routing labeling schemes proposed in~\cite{ChLaRa19}. Eventually, a linear-time algorithm~\cite{BeHa21} for the diameter on constant-dimension median graphs was proposed, {\em i.e.} for median graphs satisfying $d=O(1)$.

The existence of a truly subquadratic-time algorithm for the diameter on all median graphs is open and was recently formulated in~\cite{BeChChVa20,Du20}. An even more ambitious question can be asked. Can this subquadratic barrier be overpassed for the problem of finding all eccentricities of a median graph ? As the total size of the output is linear and this problem generalizes the diameter one, this question is legitimate. More generally, the question holds for all metric parameters (except the median set and the Wiener index for which a linear-time algorithm was recently designed). In this article, we propose the first subquadratic-time algorithm computing all eccentricities on median graphs. 

\subsection{Contributions}

Our first contribution in this paper is the design of a quasilinear, {\em i.e.} $O((\log n)^{O(1)}n)$, time algorithm computing the diameter of simplex graphs. A simplex graph $K(G) = (V_K,E_K)$ of a graph $G$ is obtained by considering the induced complete graphs (cliques) of $G$ as vertices $V_K$. Then, two of these cliques are connected by an edge if they differ by only one element: one is $C$, the other is $C \cup \set{v}$. These edges form the set $E_K$. All simplex graphs are median~\cite{BaCh08,BaLeMo86}. Moreover, we observe that simplex graphs admit an interesting property: they admit a central vertex - representing the empty clique - and every $\Theta$-class has an edge incident to that vertex (Lemma~\ref{le:center_simplex}, Section~\ref{subsec:crossing}).

We describe the algorithm in a few words. We observed that the eccentricities of each vertex of a simplex graph could be written as functions the size of certain sets of pairwise orthogonal $\Theta$-classes (POFs). Based on that property, we order the POFs in function of their size and execute partition refinements. This reveals us a tree structure of the POFs from which the eccentricity of each vertex can be extracted. 

First, this algorithm extends the set of median graphs for which a quasilinear time procedure computing the diameter exists. Indeed, simplex graphs form a sub-class of median graphs containing instances with unbounded dimension $d$.

\begin{itemize}
    \item There is a combinatorial algorithm determining the diameter and all eccentricities of simplex graphs in $O((d^3+\log n)n)$: Corollary~\ref{co:linear_simplex}, Section~\ref{subsec:partitioning}.
\end{itemize}

Second, we remark that this method can be integrated to the algorithm already proposed in~\cite{BeHa21} to compute all eccentricities of median graphs in time $O(2^{O(d\log d)}n)$. This allows us to decrease this running time. Thanks to this modification, the new algorithm proposed computes all eccentricities of a median graph in $\tilde{O}(2^{2d}n)$, where notation $\tilde{O}$ neglects poly-logarithmic factors. Even if the algorithm stays linear for constant-dimension median graphs, observe that the dependence on $d$ decreases, from a slightly super-exponential function to a simple exponential one.

\begin{itemize}
\item There is a combinatorial algorithm determining all eccentricities of median graphs in $\tilde{O}(2^{2d}n)$: Theorem~\ref{th:simple_ecc}, Section~\ref{subsec:constant_dim}.
\end{itemize}

The second and main contribution in this paper is the design of a subquadratic-time dynamic programming procedure which computes all eccentricities of any median graph. Here, the linear simple-exponential-FPT algorithm for all eccentricities presented above plays a crucial role: it is the base case. This framework consists in partitioning recursively the input graph $G$ into the halfspaces of its largest $\Theta$-class. With our construction, the leaves of this recursive tree are median graphs with dimension at most $\frac{1}{3}\log n$ and we can apply the former linear-time FPT algorithm.

\begin{itemize}
    \item There is a combinatorial algorithm determining all eccentricities of median graphs in $\tilde{O}(n^{\frac{5}{3}})$: Theorem~\ref{thm:guigui-5}, Section~\ref{subsec:reduction}.
\end{itemize}

These three results so far stand as the core of our article (first line of Table~\ref{tab:results}). They put in evidence fast algorithms for the computation of eccentricities on median graphs. In particular, the last one solves a challenging open question: the existence of subquadratic-time exact algorithms for this problem.

\begin{table}[h]
\centering
\setlength{\extrarowheight}{0.2cm}
\begin{tabular}{|c|c|c|}
\hline
 & Function of $d,n$ & Function of $n$\\
 \hline
 Main results & $\tilde{O}(4^dn)$ & $n^{\frac{5}{3}}$\\
 \hline
 Improvements & $\tilde{O}(3.5394^dn)$ & $n^{1.6408}$\\
 \hline
\end{tabular}
\caption{Best running times obtained for the computation of eccentricities on median graphs}
\label{tab:results}
\end{table}

We terminate the article (Section~\ref{sec:discussion}) with some improvements of the framework we designed. We focus first on the computation of all reach centralities~\cite{Gu04} in a median graph. The reach centrality of a vertex $u$ is the maximum value $\min \set{d(s,u),d(u,t)}$ over all pairs $s,t$ satisfying $u \in I(s,t)$. A linear simple-exponential-FPT algorithm is proposed, as for eccentricities.

\begin{itemize}
    \item There is a combinatorial algorithm determining all reach centralities of median graphs in $\tilde{O}(2^{3d}n)$: Theorem~\ref{th:simple_rc}, Section~\ref{subsec:reach_centrality}.
\end{itemize}

Then, we define a new discrete structure on median graphs: the \textit{maximal outgoing POFs} (MOPs), generalizing the POFs used throughout the paper and defined in~\cite{BeHa21}. We propose an alternative procedure to compute all eccentricities, based on the enumeration of MOPs. Furthermore, the MOPs admit an interesting property: for median graphs with ``large'' $d$, their number is subquadratic. This provides us with a better subquadratic-time algorithm for the eccentricities using the following win-win approach: either the dimension $d$ of the input graph $G$ is ``small'' and the linear FPT algorithm is executed fast, or the dimension is ``large'', then the MOPs can be enumerated fast and our new procedure ensures a smaller runtime.

\begin{itemize}
    \item There is a combinatorial algorithm determining all eccentricities of median graphs in $\tilde{O}(n^{\beta})$, where $\beta = 1.6456$: Theorem~\ref{th:subquadramop}, Section~\ref{subsec:mop}.
\end{itemize}

Eventually, we present a new relationship between POFs which allows us to improve the running time of the linear-time FPT algorithm. 

\begin{itemize}
    \item There is a combinatorial algorithm determining all eccentricities of median graphs in $\tilde{O}(3.5394^dn)$: Corollary~\ref{co:compute_ecc_fast}, Section~\ref{subsec:faster_enum}.
\end{itemize}

Combining this new tool with the techniques associated with MOPs, we obtain a third subquadratic-time algorithm in this paper for the computation of eccentricities on median graphs.

\begin{itemize}
    \item There is a combinatorial algorithm determining all eccentricities of median graphs in $\tilde{O}(n^{\gamma})$, where $\gamma = 1.6408$: Theorem~\ref{th:subquadradj}, Section~\ref{subsec:faster_enum}.
\end{itemize}

All these outcomes put in evidence a relationship between the design of linear-time FPT algorithms and the design of subquadratic-time algorithms determining metric parameters on median graphs. We believe that the ideas proposed to establish all these results represent interesting tools to break the subquadratic barrier on other open questions.

Table~\ref{tab:results} summarizes the results of our paper. We would like to distinguish the contributions of Sections~\ref{sec:simplex} and~\ref{sec:subquadratic}, which represent significant advances for the computation of metric parameters on median graphs, with the contributions of Section~\ref{sec:discussion} which consist in improvements of the latter.

\subsection{Organization}

In Section~\ref{sec:median}, we remind the definition of median graphs. The well-known properties and concepts related to them are listed, among them $\Theta$-classes, signature, and POFs. Section~\ref{sec:simplex} is utterly dedicated to simplex graphs: we establish some characterizations of these graphs and we present our quasilinear-time algorithm determining their eccentricities. In Section~\ref{sec:subquadratic}, we show how to obtain a linear simple-exponential-FPT algorithm for all eccentricities of a median graph, parameterized by the dimension $d$. Thanks to it, we propose a dynamic programming procedure to transform the computation of eccentricities of any median graph into a sequence of constant-dimension cases. In Section~\ref{sec:discussion}, we extend the results obtained so far. We present a linear simple-exponential-FPT algorithm computing all reach centralities of a median graph, parameterized by $d$. Moreover, we define the notion of MOPs, provide an upper bound of their cardinality and show the impact of this bound on the time complexity on the algorithms proposed earlier. We also introduce a new relationship between POFs. Eventually, we conclude in Section~\ref{sec:conclusion} and give some directions of research which could follow the contributions of this article.

\section{Median graphs} \label{sec:median}

In this section, we recall some notions related to distances in graphs, and more particularly median graphs. Two important tools are presented: the $\Theta$-classes, which are equivalences classes over the edge set, and the \textit{Pairwise Orthogonal Families} (POFs) characterizing $\Theta$-classes belonging to a common hypercube. 

\subsection{$\Theta$-classes} 

All graphs $G = (V,E)$ considered in this paper are undirected, unweighted, simple, finite and connected. We denote by $N(u)$ the \textit{open neighborhood} of $u \in V$, {\em i.e.} the set of vertices adjacent to $u$ in $G$. We extend it naturally: for any set $A \subseteq V$, the neighborhood $N(A)$ of $A$ is the set of vertices outside $A$ adjacent to some $u \in A$.

Given two vertices $u,v \in V$, let $d(u,v)$ be the \textit{distance} between $u$ and $v$, {\em i.e.} the length of the shortest $(u,v)$-path. The \textit{eccentricity} $\ecc(u)$ of a vertex $u \in V$ is the length of the longest shortest path starting from $u$. Put formally, $\ecc(u)$ is the maximum value $d(u,v)$ for all $v \in V$: $\ecc(u) = \max_{v \in V} d(u,v)$. The diameter of graph $G$ is the maximum distance between two of its vertices: $\diam(G) = \max_{u \in V} \ecc(u)$. 

We denote by $I(u,v)$ the \textit{interval} of pair $u,v$. It contains exactly the vertices which are metrically between $u$ and $v$:
$I(u,v) = \set{x \in V: d(u,x) + d(x,v) = d(u,v)}$. The vertices of $I(u,v)$ are lying on at least one shortest $(u,v)$-path.

We say that a set $H\subseteq V$ (or the induced subgraph $G\left[H\right]$) is \textit{convex} if $I(u,v) \subseteq H$ for any pair $u,v \in H$. Moreover, we say that $H$ is \textit{gated} if any vertex $v \notin H$ admits a \textit{gate} $g_H(v) \in H$, {\em i.e.} a vertex that belongs to all intervals $I(v,x)$, $x\in H$. For any $x \in H$, we have $d(v,g_H(v)) + d(g_H(v),x) = d(v,x)$. Gated sets are convex by definition.

Given an integer $k \ge 1$, the hypercube of dimension $k$, $Q_k$, is a graph representing all the subsets of $\set{1,\ldots,k}$ as the vertex set. An edge connects two subsets if one is included into the other and they differ by only one element. Hypercube $Q_2$ is a \textit{square} and $Q_3$ is a $3$-cube.

\begin{definition}[Median graph]
A graph is \textit{median} if, for any triplet $x,y,z$ of distinct vertices, the set $I(x,y) \cap I(y,z) \cap I(z,x)$ contains exactly one vertex $m(x,y,z)$ called the median of $x,y,z$.
\label{def:median}
\end{definition}

Observe that certain well-known families of graphs are median: trees, grids, squaregraphs~\cite{BaChEp10}, and hypercubes $Q_k$.
Median graphs are bipartite and do not contain an induced $K_{2,3}$~\cite{BaCh08,HaImKl11,Mu78}. They can be obtained by Mulder's convex expansion~\cite{Mu78,Mu80} starting from a single vertex.

\begin{figure}[h]

\begin{subfigure}[b]{0.33\columnwidth}
\centering
\scalebox{0.6}{\begin{tikzpicture}


\node[draw, circle, minimum height=0.2cm, minimum width=0.2cm, fill=black] (P0a) at (2,5) {};
\node[draw, circle, minimum height=0.2cm, minimum width=0.2cm, fill=black] (P0b) at (4,5) {};
\node[draw, circle, minimum height=0.2cm, minimum width=0.2cm, fill=black] (P1) at (3,4) {};
\node[draw, circle, minimum height=0.2cm, minimum width=0.2cm, fill=black] (P2) at (2,2.5) {};
\node[draw, circle, minimum height=0.2cm, minimum width=0.2cm, fill=black] (P3) at (4,2.5) {};
\node[draw, circle, minimum height=0.2cm, minimum width=0.2cm, fill=black] (P4) at (1.2,1) {};
\node[draw, circle, minimum height=0.2cm, minimum width=0.2cm, fill=black] (P5) at (2.8,1) {};


\draw[line width = 1.4pt] (P0a) -- (P1);
\draw[line width = 1.4pt] (P0b) -- (P1);
\draw[line width = 1.4pt] (P1) -- (P2);
\draw[line width = 1.4pt] (P1) -- (P3);
\draw[line width = 1.4pt] (P2) -- (P4);
\draw[line width = 1.4pt] (P2) -- (P5);

\end{tikzpicture}}
\caption{Tree, $d=1$}
\end{subfigure}
\begin{subfigure}[b]{0.33\columnwidth}
\centering
\scalebox{0.6}{\begin{tikzpicture}


\node[draw, circle, minimum height=0.2cm, minimum width=0.2cm, fill=black] (P11) at (1,1) {};
\node[draw, circle, minimum height=0.2cm, minimum width=0.2cm, fill=black] (P12) at (1,2.5) {};

\node[draw, circle, minimum height=0.2cm, minimum width=0.2cm, fill=black] (P21) at (3,1) {};
\node[draw, circle, minimum height=0.2cm, minimum width=0.2cm, fill=black] (P22) at (3,2.5) {};
\node[draw, circle, minimum height=0.2cm, minimum width=0.2cm, fill=black] (P23) at (3,4) {};

\node[draw, circle, minimum height=0.2cm, minimum width=0.2cm, fill=black] (P31) at (5,1) {};
\node[draw, circle, minimum height=0.2cm, minimum width=0.2cm, fill=black] (P32) at (5,2.5) {};
\node[draw, circle, minimum height=0.2cm, minimum width=0.2cm, fill=black] (P33) at (5,4) {};

\node[draw, circle, minimum height=0.2cm, minimum width=0.2cm, fill=black] (P41) at (1.3,3.8) {};
\node[draw, circle, minimum height=0.2cm, minimum width=0.2cm, fill=black] (P42) at (1.3,5.3) {};
\node[draw, circle, minimum height=0.2cm, minimum width=0.2cm, fill=black] (P43) at (-0.7,3.8) {};

\node[draw, circle, minimum height=0.2cm, minimum width=0.2cm, fill=black] (P44) at (-0.7,2.0) {};

\node[draw, circle, minimum height=0.2cm, minimum width=0.2cm, fill=black] (P51) at (4.0,5.3) {};
\node[draw, circle, minimum height=0.2cm, minimum width=0.2cm, fill=black] (P52) at (6.0,5.3) {};


\draw[line width = 1.4pt] (P11) -- (P12);
\draw[line width = 1.4pt] (P11) -- (P21);
\draw[line width = 1.4pt] (P12) -- (P22);
\draw[line width = 1.4pt] (P21) -- (P22);

\draw[line width = 1.4pt] (P21) -- (P31);
\draw[line width = 1.4pt] (P22) -- (P32);
\draw[line width = 1.4pt] (P31) -- (P32);

\draw[line width = 1.4pt] (P22) -- (P23);
\draw[line width = 1.4pt] (P23) -- (P33);
\draw[line width = 1.4pt] (P32) -- (P33);

\draw[line width = 1.4pt] (P22) -- (P41);
\draw[line width = 1.4pt] (P12) -- (P43);
\draw[line width = 1.4pt] (P23) -- (P42);
\draw[line width = 1.4pt] (P41) -- (P43);
\draw[line width = 1.4pt] (P41) -- (P42);
\draw[line width = 1.4pt] (P12) -- (P44);

\draw[line width = 1.4pt] (P23) -- (P51);
\draw[line width = 1.4pt] (P33) -- (P52);
\draw[line width = 1.4pt] (P51) -- (P52);

\end{tikzpicture}}
\caption{Squaregraph, $d=2$}
\end{subfigure}
\begin{subfigure}[b]{0.33\columnwidth}
\centering
\scalebox{0.6}{\begin{tikzpicture}


\node[draw, circle, minimum height=0.2cm, minimum width=0.2cm, fill=black] (P11) at (4,5) {};
\node[draw, circle, minimum height=0.2cm, minimum width=0.2cm, fill=black] (P12) at (4,6.5) {};

\node[draw, circle, minimum height=0.2cm, minimum width=0.2cm, fill=black] (P21) at (6,5) {};
\node[draw, circle, minimum height=0.2cm, minimum width=0.2cm, fill=black] (P22) at (6,6.5) {};

\node[draw, circle, minimum height=0.2cm, minimum width=0.2cm, fill=black] (P31) at (5.0,5.4) {};
\node[draw, circle, minimum height=0.2cm, minimum width=0.2cm, fill=black] (P32) at (5.0,6.9) {};
\node[draw, circle, minimum height=0.2cm, minimum width=0.2cm, fill=black] (P33) at (7.0,5.4) {};
\node[draw, circle, minimum height=0.2cm, minimum width=0.2cm, fill=black] (P34) at (7.0,6.9) {};


\node[draw, circle, minimum height=0.2cm, minimum width=0.2cm, fill=black] (P41) at (7,1) {};
\node[draw, circle, minimum height=0.2cm, minimum width=0.2cm, fill=black] (P42) at (7,2.5) {};

\node[draw, circle, minimum height=0.2cm, minimum width=0.2cm, fill=black] (P51) at (9,1) {};
\node[draw, circle, minimum height=0.2cm, minimum width=0.2cm, fill=black] (P52) at (9,2.5) {};

\node[draw, circle, minimum height=0.2cm, minimum width=0.2cm, fill=black] (P61) at (8.0,1.4) {};
\node[draw, circle, minimum height=0.2cm, minimum width=0.2cm, fill=black] (P62) at (8.0,2.9) {};
\node[draw, circle, minimum height=0.2cm, minimum width=0.2cm, fill=black] (P63) at (10.0,1.4) {};
\node[draw, circle, minimum height=0.2cm, minimum width=0.2cm, fill=black] (P64) at (10.0,2.9) {};


\draw[line width = 1.4pt] (P11) -- (P12);

\draw[line width = 1.4pt] (P11) -- (P21);
\draw[line width = 1.4pt] (P12) -- (P22);
\draw[line width = 1.4pt] (P21) -- (P22);

\draw[line width = 1.4pt] (P11) -- (P31);
\draw[line width = 1.4pt] (P12) -- (P32);
\draw[line width = 1.4pt] (P21) -- (P33);
\draw[line width = 1.4pt] (P22) -- (P34);
\draw[line width = 1.4pt] (P31) -- (P32);
\draw[line width = 1.4pt] (P31) -- (P33);
\draw[line width = 1.4pt] (P32) -- (P34);
\draw[line width = 1.4pt] (P33) -- (P34);

\draw[line width = 1.4pt] (P41) -- (P42);

\draw[line width = 1.4pt](P41) -- (P51);
\draw[line width = 1.4pt] (P42) -- (P52);
\draw[line width = 1.4pt] (P51) -- (P52);

\draw[line width = 1.4pt] (P41) -- (P61);
\draw[line width = 1.4pt] (P42) -- (P62);
\draw[line width = 1.4pt] (P51) -- (P63);
\draw[line width = 1.4pt] (P52) -- (P64);
\draw[line width = 1.4pt] (P61) -- (P62);
\draw[line width = 1.4pt] (P61) -- (P63);
\draw[line width = 1.4pt] (P62) -- (P64);
\draw[line width = 1.4pt] (P63) -- (P64);

\draw[line width = 1.4pt] (P11) -- (P41);
\draw[line width = 1.4pt] (P12) -- (P42);
\draw[line width = 1.4pt] (P21) -- (P51);
\draw[line width = 1.4pt] (P22) -- (P52);
\draw[line width = 1.4pt] (P31) -- (P61);
\draw[line width = 1.4pt] (P32) -- (P62);
\draw[line width = 1.4pt] (P33) -- (P63);
\draw[line width = 1.4pt] (P34) -- (P64);

\end{tikzpicture}}
\caption{4-cube, $d=4$}
\end{subfigure}

\caption{Examples of median graphs}
\label{fig:median_examples}
\end{figure}
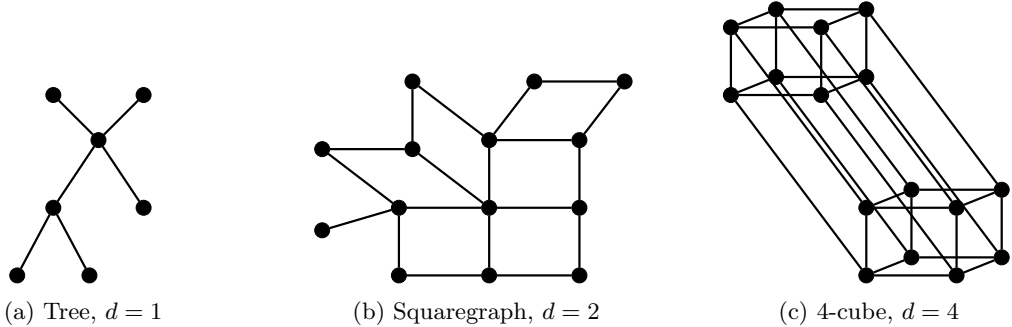

Now, we define a parameter which has a strong influence on the study of median graphs.
The dimension $d = \mbox{dim}(G)$ of a median graph $G$ is the dimension of the largest hypercube contained in $G$ as an induced subgraph.
In other words, $G$ admits $Q_d$ as an induced subgraph, but not $Q_{d+1}$. Median graphs with $d=1$ are exactly the trees. Median graphs with $d\le 2$ are called \textit{cube-free} median graphs.

Figure~\ref{fig:median_examples} presents three examples of median graphs. (a) is a tree: $d=1$. (b) is a cube-free median graph: it has dimension $d=2$. To be more precise, it is a squaregraph~\cite{BaChEp10}, which is a sub-family of cube-free median graphs. The last one (c) is a 4-cube: it has dimension $d=4$. 

We provide a list of properties satisfied by median graphs. In particular, we define the notion of $\Theta$-classes which is a key ingredient of several existing algorithms~\cite{BeChChVa20,HaImKl99,ImKlMu99}.

In general graphs, all gated subgraphs are convex. The reverse is true in median graphs.
\begin{lemma}[Convex$\Leftrightarrow$Gated~\cite{BaCh08,BeChChVa20}]
Any convex subgraph of a median graph is gated.
\end{lemma}

To improve readibility, edges $(u,v) \in E$ are sometimes denoted by $uv$. We remind the notion of $\Theta$-class, which is well explained in~\cite{BeChChVa20}, and enumerate some properties related to it. We say that the edges $uv$ and $xy$ are in relation $\Theta_0$ if they form a square $uvyx$, where $uv$ and $xy$ are opposite edges. Then, $\Theta$ refers to the reflexive and transitive closure of relation $\Theta_0$. Let $q$ be the number of equivalence classes obtained with this relation. The classes of the equivalence relation $\Theta$ are denoted by $E_1,\ldots,E_q$. Concretely, two edges $uv$ and $u'v'$ belong to the same $\Theta$-class if there is a sequence $uv = u_0v_0, u_1v_1, \ldots, u_rv_r= u'v'$ such that $u_iv_i$ and $u_{i+1}v_{i+1}$ are opposite edges of a square. We denote by $\mathcal{E}$ the set of $\Theta$-classes: $\mathcal{E} = \set{E_1,\ldots,E_q}$. To avoid confusions, let us highlight that parameter $q$ is different from the dimension $d$: for example, on trees, $d=1$ whereas $q = n-1$. Moreover, the dimension $d$ is at most $\lfloor \log n \rfloor$ in general.

\begin{lemma}[$\Theta$-classes in linear time~\cite{BeChChVa20}]
There exists an algorithm which computes the $\Theta$-classes $E_1,\ldots,E_q$ of a median graph in linear time $O(\card{E}) = O(dn)$.
\label{le:linear_classes}
\end{lemma}

In median graphs, each class $E_i$, $1\le i\le q$, is a perfect matching cutset and its two sides $H_i'$ and $H_i''$ verify nice properties, that are presented below.

\begin{lemma}[Halfspaces of $E_i$~\cite{BeChChVa20,HaImKl99,Mu80}]
For any $1\le i\le q$, the graph $G$ deprived of edges of $E_i$, {\em i.e.} $G\backslash E_i = (V,E\backslash E_i)$, has two connected components $H_i'$ and $H_i''$, called \textit{halfspaces}. Edges of $E_i$ form a matching: they have no endpoint in common. Halfspaces satisfy the following properties.
\begin{itemize}
\item Both $H_i'$ and $H_i''$ are convex/gated.
\item If $uv$ is an edge of $E_i$ with $u \in H_i'$ and $v \in H_i''$, then $H_i' = W(u,v) = \set{x \in V: d(x,u) < d(x,v)}$ and $H_i'' = W(v,u) = \set{x \in V: d(x,v) < d(x,u)}$.
\end{itemize}
\label{le:halfspaces}
\end{lemma}

\begin{figure}[h]
\centering
\scalebox{0.85}{\begin{tikzpicture}


\draw [dashed, color = white!40!blue, fill = white!92!blue] (-1.0,0.6) -- (4.4,0.6) -- (4.4,5.6) -- (-1.0,5.6) -- (-1.0,0.6) node[below left] {$H_i'$};
\draw [dashed, color = white!40!red, fill = white!92!red] (6.3,0.6) -- (6.3,5.6) -- (4.6,5.6) -- (4.6,0.6) -- (6.3,0.6) node[below right] {$H_i'' = \partial H_i''$};


\node[draw, circle, minimum height=0.2cm, minimum width=0.2cm, fill=black] (P11) at (1,1) {};
\node[draw, circle, minimum height=0.2cm, minimum width=0.2cm, fill=black] (P12) at (1,2.5) {};

\node[draw, circle, minimum height=0.2cm, minimum width=0.2cm, fill=blue] (P21) at (3,1) {};
\node[draw, circle, minimum height=0.2cm, minimum width=0.2cm, fill=blue] (P22) at (3,2.5) {};
\node[draw, circle, minimum height=0.2cm, minimum width=0.2cm, fill=blue] (P23) at (3,4) {};

\node[draw, circle, minimum height=0.2cm, minimum width=0.2cm, fill=red] (P31) at (5,1) {};
\node[draw, circle, minimum height=0.2cm, minimum width=0.2cm, fill=red] (P32) at (5,2.5) {};
\node[draw, circle, minimum height=0.2cm, minimum width=0.2cm, fill=red] (P33) at (5,4) {};

\node[draw, circle, minimum height=0.2cm, minimum width=0.2cm, fill=black] (P41) at (1.3,3.8) {};
\node[draw, circle, minimum height=0.2cm, minimum width=0.2cm, fill=black] (P42) at (1.3,5.3) {};
\node[draw, circle, minimum height=0.2cm, minimum width=0.2cm, fill=black] (P43) at (-0.7,3.8) {};

\node[draw, circle, minimum height=0.2cm, minimum width=0.2cm, fill=black] (P44) at (-0.7,2.0) {};

\node[draw, circle, minimum height=0.2cm, minimum width=0.2cm, fill=blue] (P51) at (4.0,5.3) {};
\node[draw, circle, minimum height=0.2cm, minimum width=0.2cm, fill=red] (P52) at (6.0,5.3) {};


\draw[line width = 1.4pt] (P11) -- (P12);
\draw[line width = 1.4pt] (P11) -- (P21);
\draw[line width = 1.4pt] (P12) -- (P22);
\draw[line width = 1.4pt] (P21) -- (P22);

\draw[line width = 1.6pt, color = green] (P21) -- (P31);
\draw[line width = 1.6pt, color = green] (P22) -- (P32);
\draw[line width = 1.4pt] (P31) -- (P32);

\draw[line width = 1.4pt] (P22) -- (P23);
\draw[line width = 1.6pt, color = green] (P23) -- (P33);
\draw[line width = 1.4pt] (P32) -- (P33);

\draw[line width = 1.4pt] (P22) -- (P41);
\draw[line width = 1.4pt] (P12) -- (P43);
\draw[line width = 1.4pt] (P23) -- (P42);
\draw[line width = 1.4pt] (P41) -- (P43);
\draw[line width = 1.4pt] (P41) -- (P42);
\draw[line width = 1.4pt] (P12) -- (P44);

\draw[line width = 1.4pt] (P23) -- (P51);
\draw[line width = 1.4pt] (P33) -- (P52);
\draw[line width = 1.6pt, color=green] (P51) -- (P52);


\node[scale=1.2, color = blue] at (2.9,4.8) {$\partial H_i'$};
\node[scale=1.4, color = green] at (4.0,0.2) {$E_i$};

\end{tikzpicture}}
\caption{A class $E_i$ with sets $H_i', H_i'', \partial H_i', \partial H_i''$}
\label{fig:halfspaces}
\end{figure}
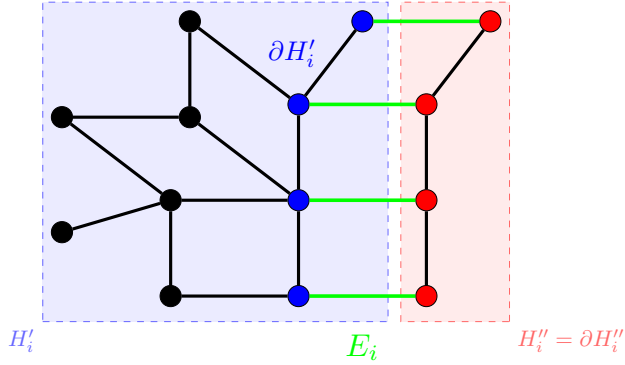

We denote by $\partial H_i'$ the subset of $H_i'$ containing the vertices which are adjacent to a vertex in $H_i''$: $\partial H_i' = N(H_i'')$ Put differently, the set $\partial H_i'$ is made up of vertices of $H_i'$ which are endpoints of edges in $E_i$. Symmetrically, set $\partial H_i''$ contains the vertices of $H_i''$ which are adjacent to $H_i'$. We say these sets are the \textit{boundaries} of halfspaces $H_i'$ and $H_i''$ respectively. Figure~\ref{fig:halfspaces} illustrates the notions of $\Theta$-class, halfspace and boundary on a small example. In this particular case, an halfspace is equal to its boundary: $\partial H_i'' = H_i''$. The vertices of $\partial H_i'$ are colored in blue.

\begin{lemma}[Boundaries~\cite{BeChChVa20,HaImKl99,Mu80}]
Both $\partial H_i'$ and $\partial H_i''$ are convex/gated. Moreover, the edges of $E_i$ define an isomorphism between $\partial H_i'$ and $\partial H_i''$.
\label{le:boundaries}
\end{lemma}

As a consequence, suppose $uv$ and $u'v'$ belong to $E_i$: if $uu'$ is an edge and belongs to class $E_j$, then $vv'$ is an edge too and it belongs to $E_j$. We terminate this list of lemmas with a last property dealing with the orientation of edges  from a canonical basepoint $v_0 \in V$. The \textit{$v_0$-orientation} of the edges of $G$ according to $v_0$ is such that, for any edge $uv$, the orientation is $\vv{uv}$ if $d(v_0,u) < d(v_0,v)$. Indeed, we cannot have $d(v_0,u) = d(v_0,v)$ as $G$ is bipartite. The $v_0$-orientation is acyclic.

\begin{lemma}[Orientation~\cite{BeChChVa20}]
All edges can be oriented according to any canonical basepoint $v_0$.
\end{lemma}

From now on, we suppose that an arbitrary basepoint $v_0 \in V$ has been selected and we refer automatically to the $v_0$-orientation when we mention incoming or outgoing edges.

\subsection{Shortest paths and signature} \label{subsec:signature}

We fix an arbitrary canonical basepoint $v_0$ and for each class $E_i$, we say that the halfspace containing $v_0$ is $H_i'$.
Given two vertices $u,v \in V$, we define the set which contains the $\Theta$-classes separating $u$ from $v$.

\begin{definition}[Signature $\sigma_{u,v}$]
We say that the {\em signature} of the pair of vertices $u,v$, denoted by $\sigma_{u,v}$, is the set of classes $E_i$ such that $u$ and $v$ are separated in $G\backslash E_i$. In other words, $u$ and $v$ are in different halfspaces of $E_i$.
\label{def:signature}
\end{definition}

The signature of two vertices provide us with the composition of any shortest $(u,v)$-path. Indeed, all shortest $(u,v)$-paths contain exactly one edge for each class in $\sigma_{u,v}$.

\begin{lemma}[\cite{BeHa21}]
For any shortest $(u,v)$-path $P$, the edges in $P$ belong to classes in $\sigma_{u,v}$ and, for any $E_i \in \sigma_{u,v}$, there is exactly one edge of $E_i$ in path $P$. Conversely, a path containing at most one edge of each $\Theta$-class is a shortest path between its departure and its arrival.
\label{le:signature}
\end{lemma}

This result is a direct consequence of the convexity of halfspaces. By contradiction, a shortest path that would pass through two edges of some $\Theta$-class $E_i$ would escape temporarily from an halfspace, say w.l.o.g $H_i'$,  which is convex (Lemma~\ref{le:halfspaces}).

Definition~\ref{def:signature} can be generalized: given a set of edges $B\subseteq E$, its signature is the set of $\Theta$-classes represented in that set: $\set{E_i : uv \in E_i \cap B}$. The signature of a path is the set of classes which have at least one edge in this path. In this way, the signature $\sigma_{u,v}$ is also the signature of any shortest $(u,v)$-path. The signature of a hypercube is the set of $\Theta$-classes represented in its edges: the cardinality of the signature is thus equal to the dimension of the hypercube.

\subsection{Orthogonal $\Theta$-classes and hypercubes}

We present now another important notion on median graphs: \textit{orthogonality}. In~\cite{Ko09}, Kovse studied a relationship between \textit{splits} which refer to the halfspaces of $\Theta$-classes. It says that two splits $\set{H_i',H_i''}$ and $\set{H_j',H_j''}$ are \textit{incompatible} if the four sets $H_i' \cap H_j'$, $H_i'' \cap H_j'$, $H_i' \cap H_j''$, and $H_i'' \cap H_j''$ are nonempty. Another definition was proven equivalent to this one.

\begin{definition}[Orthogonal $\Theta$-classes]
We say that classes $E_i$ and $E_j$ are {\em orthogonal} ($E_i \perp E_j$) if there is a square $uvyx$ in $G$, where $uv,xy \in E_i$ and $ux,vy \in E_j$.
\end{definition}

Indeed, classes $E_i$ and $E_j$ are orthogonal if and only if the splits produced by their halfspaces are incompatible.

\begin{lemma}[Orthogonal$\Leftrightarrow$Incompatible~\cite{BeHa21}] Given two $\Theta$-classes $E_i$ and $E_j$ of a median graph $G$, the following statements are equivalent:
\begin{itemize}
    \item Classes $E_i$ and $E_j$ are orthogonal,
    \item Splits $\set{H_i',H_i''}$ and $\set{H_j',H_j''}$ are incompatible,
    \item The four sets $\partial H_i' \cap \partial H_j'$, $\partial H_i'' \cap \partial H_j'$, $\partial H_i' \cap \partial H_j''$, and $\partial H_i'' \cap \partial H_j''$ are nonempty.
\end{itemize}
\label{le:perp_incomp}
\end{lemma}

The concept of \textit{orthogonality} is sometimes described with different words in the literature depending on the context: incompatible, \textit{concurrent} or \textit{crossing}. We say that $E_i$ and $E_j$ are \textit{parallel} if they are not orthogonal, that is $H_i \subseteq H_j$ for some $H_i \in \set{H_i',H_i''}$ and $H_j \in \set{H_j',H_j''}$. 

We pursue with a property on orthogonal $\Theta$-classes: if two edges of two orthogonal classes $E_i$ and $E_j$ are incident, they belong to a common square.

\begin{lemma}[Squares~\cite{BaCo93,BeHa21}]
Let $xu \in E_i$ and $uy \in E_j$. If $E_i$ and $E_j$ are orthogonal, then there is a vertex $v$ such that $uyvx$ is a square.
\label{le:squares}
\end{lemma}

\textbf{Pairwise orthogonal families}. We focus on the set of $\Theta$-classes which are pairwise orthogonal.

\begin{definition}[Pairwise Orthogonal Family]
We say that a set of classes $X \subseteq \mathcal{E}$ is a {\em Pairwise Orthogonal Family (POF for short)} if for any pair $E_j,E_h \in X$, we have $E_j \perp E_h$.
\end{definition}

The empty set is considered as a POF. We denote by $\mathcal{L}$ the set of POFs of the median graph $G$. The notion of POF is strongly related to the induced hypercubes in median graphs. First, observe that all $\Theta$-classes of a median graph form a POF if and only if the graph is a hypercube of dimension $\log n$~\cite{Ko09,MoMuRo98}. Secondly, the next lemma precises the relationship  between POFs and hypercubes.

\begin{lemma}[POFs adjacent to a vertex~\cite{BeHa21}]
Let $X$ be a POF, $v \in V$, and assume that for each $E_i \in X$, there is an edge of $E_i$ adjacent to $v$. There exists a hypercube $Q$ containing vertex $v$ and all edges of $X$ adjacent to $v$. Moreover, the $\Theta$-classes of the edges of $Q$ are the classes of $X$.
\label{le:pof_adjacent}
\end{lemma}

There is a natural bijection between the vertices of a median graph and the POFs. The next lemma exhibits this relationship.

\begin{lemma}[POFs and hypercubes~\cite{BaChDrKo06,BaQuSaMa02,Ko09}]
Consider an arbitrary canonical basepoint $v_0 \in V$ and the $v_0$-orientation for the median graph $G$. Given a vertex $v \in V$, let $N^-(v)$ be the set of edges going into $v$ according to the $v_0$-orientation. Let $\mathcal{E}^-(v)$ be the classes of the edges in $N^-(v)$. The following propositions are true:
\begin{itemize}
\item For any vertex $v\in V$, $\mathcal{E}^-(v)$ is a POF. Moreover, vertex $v$ and the edges of $N^-(v)$ belong to an induced hypercube formed by the classes $\mathcal{E}^-(v)$. Hence, $\card{\mathcal{E}^-(v)} = \card{N^-(v)} \le d$.
\item For any POF $X$, there is an unique vertex $v_X$ such that $\mathcal{E}^-(v_X) = X$. Vertex $v_X$ is the closest-to-$v_0$ vertex $v$ such that $X \subseteq \mathcal{E}^-(v)$.
\item The number of POFs in $G$ is equal to the number $n$ of vertices: $n = \card{\mathcal{L}}$.
\end{itemize}
\label{le:pof_hypercube}
\end{lemma}

\begin{figure}[h]
\centering
\scalebox{0.95}{\begin{tikzpicture}


\node[draw, circle, minimum height=0.2cm, minimum width=0.2cm, fill=black] (P11) at (1,1) {};
\node[draw, circle, minimum height=0.2cm, minimum width=0.2cm, fill=black] (P12) at (1,2.5) {};

\node[draw, circle, minimum height=0.2cm, minimum width=0.2cm, fill=black] (P21) at (3,1) {};
\node[draw, circle, minimum height=0.2cm, minimum width=0.2cm, fill=black] (P22) at (3,2.5) {};
\node[draw, circle, minimum height=0.2cm, minimum width=0.2cm, fill=black] (P23) at (3,4) {};

\node[draw, circle, minimum height=0.2cm, minimum width=0.2cm, fill=black] (P31) at (5,1) {};
\node[draw, circle, minimum height=0.2cm, minimum width=0.2cm, fill=black] (P32) at (5,2.5) {};
\node[draw, circle, minimum height=0.2cm, minimum width=0.2cm, fill=black] (P33) at (5,4) {};


\draw[line width = 1.4pt, color = green] (P11) -- (P12);
\draw[line width = 1.4pt, color = red] (P11) -- (P21);
\draw[line width = 1.4pt, color = red] (P12) -- (P22);
\draw[line width = 1.4pt, color = green] (P21) -- (P22);

\draw[line width = 1.4pt, color = blue] (P21) -- (P31);
\draw[line width = 1.4pt, color = blue] (P22) -- (P32);
\draw[line width = 1.4pt, color = green] (P31) -- (P32);

\draw[line width = 1.4pt, color = orange] (P22) -- (P23);
\draw[line width = 1.4pt, color = blue] (P23) -- (P33);
\draw[line width = 1.4pt, color = orange] (P32) -- (P33);


\node[scale=1.2, color = red] at (2.0,2.9) {$E_1$};
\node[scale=1.2, color = blue] at (4.0,4.4) {$E_2$};
\node[scale=1.2, color = green] at (5.5,1.75) {$E_3$};
\node[scale=1.2, color = orange] at (5.5,3.25) {$E_4$};

\node[scale = 1.2] at (0.6,0.7) {$v_0$};
\node[scale = 1.2] at (0.6,2.2) {$v_1$};
\node[scale = 1.2] at (2.6,0.7) {$v_2$};
\node[scale = 1.2] at (2.6,2.2) {$v_3$};
\node[scale = 1.2] at (2.6,3.7) {$v_4$};
\node[scale = 1.2] at (4.6,0.7) {$v_5$};
\node[scale = 1.2] at (4.6,2.2) {$v_6$};
\node[scale = 1.2] at (4.6,3.7) {$v_7$};

\node[scale = 1.2] (table) at (10.5,2.0) {
$\begin{array}{|c||c|c|c|c|}
\hline
\mbox{Vertex} & v_0 & v_1 & v_2 & v_3\\
\hline
 \mbox{POF} & \emptyset & \set{E_3} & \set{E_1} & \set{E_1,E_3}\\
\hline
\hline
\mbox{Vertex} & v_4 & v_5 & v_6 & v_7\\
\hline 
\mbox{POF} & \set{E_4} & \set{E_2} & \set{E_2,E_3} & \set{E_2,E_4}\\
\hline
\end{array}$};

\end{tikzpicture}}
\caption{Illustration of the bijection between $V$ and the set of POFs.}
\label{fig:vertices_pof}
\end{figure}
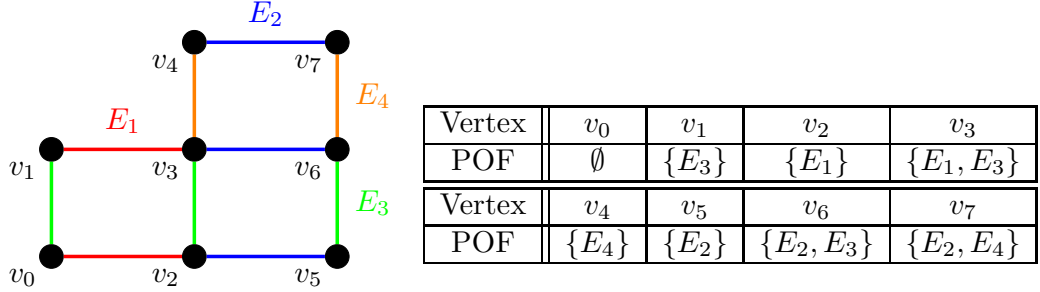

An example is given in Figure~\ref{fig:vertices_pof} with a small median graph of dimension $d=2$. $v_0$ is the canonical basepoint  and edges are colored according to  their $\Theta$-class. For example, $v_1v_3 \in E_1$. We associate with any POF $X$ of $G$ the vertex $v_X$ satisfying $\mathcal{E}^-(v_X) = X$ with the $v_0$-orientation. Obviously, the empty POF is associated with $v_0$ which has no incoming edges.

A straightforward consequence of this bijection is that parameter $q$, the number of $\Theta$-classes, is less than the number of vertices $n$. But it can be used less trivially to enumerate the POFs of a median graph in linear time~\cite{BaQuSaMa02,Ko09}. Given a basepoint $v_0$, we say that the \textit{basis} (resp. \textit{anti-basis}) of an induced hypercube $Q$ is the single vertex $v$ such that all edges of the hypercube adjacent to $v$ are outgoing from (resp. incoming into) $v$. Said differently, the basis of $Q$ is its closest-to-$v_0$ vertex and its anti-basis is its farthest-to-$v_0$ vertex. What Lemma~\ref{le:pof_hypercube} states is also that we can associate with any POF $X$ a hypercube $Q_X$ which contains exactly the classes $X$ and admits $v_X$ as its anti-basis. This observation implies that the number of POFs is less than the number of hypercubes in $G$. Moreover, the hypercube $Q_X$ is the closest-to-$v_0$ hypercube formed with the classes in $X$. Figure~\ref{subfig:ingoing_edges} shows a vertex $v$ with its incoming and outgoing edges with the $v_0$-orientation. The dashed edges represent the hypercube with anti-basis $v$ and POF $\mathcal{E}^-(v)$.

\begin{figure}[h]
\begin{subfigure}[b]{0.49\columnwidth}
\centering
\scalebox{0.8}{\begin{tikzpicture}


\node[draw, circle, minimum height=0.2cm, minimum width=0.2cm, fill=black] (P0) at (6,6) {};

\node[draw, circle, minimum height=0.2cm, minimum width=0.2cm, fill=black] (P1) at (4,6) {};
\node[draw, circle, minimum height=0.2cm, minimum width=0.2cm, fill=black] (P2) at (4.5,5.2) {};
\node[draw, circle, minimum height=0.2cm, minimum width=0.2cm, fill=black] (P3) at (6,4) {};

\node[draw, circle, minimum height=0.2cm, minimum width=0.2cm, fill=black] (P1') at (6.5,7.5) {};
\node[draw, circle, minimum height=0.2cm, minimum width=0.2cm, fill=black] (P2') at (7.5,6.5) {};

\node[draw, circle, minimum height=0.2cm, minimum width=0.2cm] (P4) at (2.5,5.2) {};
\node[draw, circle, minimum height=0.2cm, minimum width=0.2cm] (P5) at (4.5,3.2) {};
\node[draw, circle, minimum height=0.2cm, minimum width=0.2cm] (P6) at (4,4) {};
\node[draw, circle, minimum height=0.2cm, minimum width=0.2cm] (P7) at (2.5,3.2) {};


\draw[->, line width = 1.4pt] (P1) -- (P0);
\draw[->, line width = 1.4pt] (P2) -- (P0);
\draw[->, line width = 1.4pt] (P3) -- (P0);

\draw[->, line width = 1.4pt] (P0) -- (P1');
\draw[->, line width = 1.4pt] (P0) -- (P2');

\draw[dashed] (P4) -- (P1);
\draw[dashed] (P4) -- (P2);
\draw[dashed] (P5) -- (P2);
\draw[dashed] (P5) -- (P3);
\draw[dashed] (P6) -- (P1);
\draw[dashed] (P6) -- (P3);
\draw[dashed] (P4) -- (P7);
\draw[dashed] (P5) -- (P7);
\draw[dashed] (P6) -- (P7);


\node[scale=1.2] at (5.0,6.4) {$E_i$};
\node[scale=1.2] at (5.3,5.1) {$E_j$};
\node[scale=1.2] at (6.4,5.0) {$E_h$};

\node[scale = 1.2] at (6.3,6.3) {$v$};

\end{tikzpicture}}
\caption{The hypercube ``induced'' by the edges incoming into a vertex (its antibasis).}
\label{subfig:ingoing_edges}
\end{subfigure}
\begin{subfigure}[b]{0.49\columnwidth}
\centering
\scalebox{0.8}{\begin{tikzpicture}


\node[draw, circle, minimum height=0.2cm, minimum width=0.2cm, fill=blue] (P5) at (5.0,4.0) {};
\node[draw, circle, minimum height=0.2cm, minimum width=0.2cm, fill=black] (P5b) at (5.0,2.5) {};
\node[draw, circle, minimum height=0.2cm, minimum width=0.2cm, fill=black] (P5c) at (4.2,2.9) {};
\node[draw, circle, minimum height=0.2cm, minimum width=0.2cm, fill=black] (P5d) at (4.2,1.4) {};


\node[draw, circle, minimum height=0.2cm, minimum width=0.2cm, fill=blue] (P7) at (9.0,3.3) {};
\node[draw, circle, minimum height=0.2cm, minimum width=0.2cm, fill=black] (P7b) at (9.0,1.8) {};
\node[draw, circle, minimum height=0.2cm, minimum width=0.2cm, fill=black] (P7c) at (8.2,2.2) {};
\node[draw, circle, minimum height=0.2cm, minimum width=0.2cm, fill=black] (P7d) at (8.2,0.7) {};

\node[draw, circle, minimum height=0.2cm, minimum width=0.2cm, fill=blue] (P8) at (11.0,3.0) {};
\node[draw, circle, minimum height=0.2cm, minimum width=0.2cm, fill=black] (P8b) at (11.0,1.5) {};
\node[draw, circle, minimum height=0.2cm, minimum width=0.2cm, fill=black] (P8c) at (10.2,1.9) {};
\node[draw, circle, minimum height=0.2cm, minimum width=0.2cm, fill=black] (P8d) at (10.2,0.4) {};


\draw[->,line width = 1.4pt, color = red] (P5b) -- (P5);
\draw[->,line width = 1.4pt, color = green] (P5c) -- (P5);
\draw[->,line width = 1.4pt, color = red] (P5d) -- (P5c);
\draw[->,line width = 1.4pt, color = green] (P5d) -- (P5b);

\draw[->,line width = 1.4pt, dashed, color = blue] (P7) -- (P5);

\draw[->,line width = 1.4pt, color = red] (P7b) -- (P7);
\draw[->,line width = 1.4pt, color = green] (P7c) -- (P7);
\draw[->,line width = 1.4pt, color = red] (P7d) -- (P7c);
\draw[->,line width = 1.4pt, color = green] (P7d) -- (P7b);

\draw[->,line width = 1.4pt] (P8) -- (P7);
\draw[->,line width = 1.4pt, color = red] (P8b) -- (P8);
\draw[->,line width = 1.4pt, color = green] (P8c) -- (P8);
\draw[->,line width = 1.4pt, color = red] (P8d) -- (P8c);
\draw[->,line width = 1.4pt, color = green] (P8d) -- (P8b);

\draw[dotted] (P8b) -- (P7b);
\draw[dotted] (P8c) -- (P7c);
\draw[dotted] (P8d) -- (P7d);


\node[scale=1.2, color = red] at (9.4,2.5) {$E_i$};
\node[scale=1.2, color = green] at (8.1,2.8) {$E_h$};
\node[scale=1.2, color = black] at (10.0,3.5) {$E_j$};

\node[scale = 1.2, color = blue] at (7.8,4.2) {$\partial H_i'' \cap \partial H_j''$};

\node[scale=1.2] at (4.9,4.5) {$v'$};
\node[scale=1.2] at (11.1,3.5) {$v$};

\node[scale=1.2] at (4.7,3.1) {$Q'$};
\node[scale=1.2] at (10.7,2.1) {$Q$};

\end{tikzpicture}}
\caption{A POF signing at least two hypercubes $Q$ and $Q'$ is not maximal.}
\label{subfig:maximal_pof}
\end{subfigure}
\caption{Properties of POFs}
\label{fig:properties_pofs}
\end{figure}
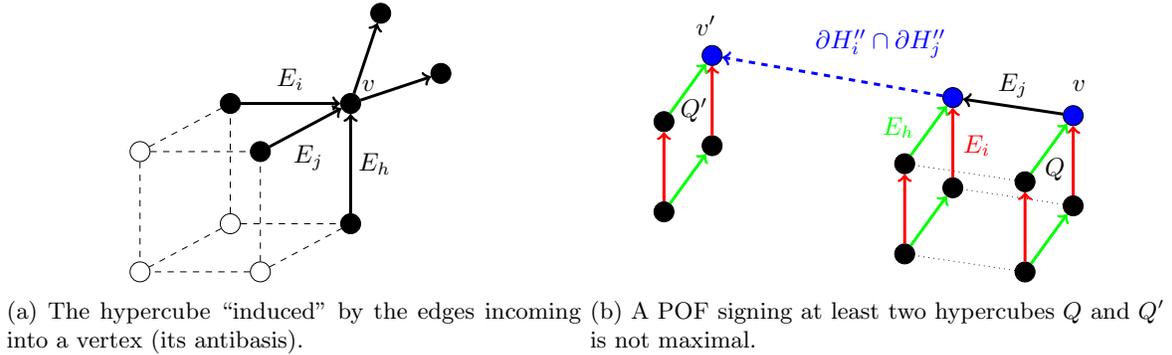 

\textbf{Number of hypercubes}. We remind a formula establishing a relationship between the number of POFs and the number of hypercubes in the literature. Let $\alpha(G)$ (resp. $\beta(G)$) be the number of hypercubes (resp. POFs) in $G$. Let $\beta_i(G)$ be the number of POFs of cardinality $i \le d$ in $G$. According to~\cite{BaQuSaMa02,Ko09}, we have:
\begin{equation}
\alpha(G) = \sum_{i=0}^d 2^i\beta_i(G)
\label{eq:number_hypercubes}
\end{equation}

Equation~\eqref{eq:number_hypercubes} produces a natural upper bound for the number of hypercubes.

\begin{lemma}[Number of hypercubes]
$\alpha(G)\le 2^dn$.
\label{le:number_hypercubes}
\end{lemma}

Value $\alpha(G)$ consider all hypercubes, in particular those of dimension 0, {\em i.e.} vertices. From now on, the word ``hypercube'' refers to the hypercubes of dimension at least one.

Each hypercube in the median graph $G$ can be defined with only its anti-basis $v$ and the edges $\widehat{N}$ of the hypercube that are adjacent and going into $v$ according to the $v_0$-orientation. These edges are a subset of $N^-(v)$: $\widehat{N} \subseteq N^-(v)$. Conversely, given a vertex $v$, each subset of $N^-(v)$ produces a hypercube which admits $v$ as an anti-basis (this hypercube is a sub-hypercube of the one obtained with $v$ and $N^-(v)$, Lemma~\ref{le:pof_hypercube}). Another possible bijection is to consider a hypercube as a pair composed of its anti-basis $v$ and the $\Theta$-classes $\widehat{\mathcal{E}}$ of the edges in $\widehat{N}$ (its signature).

As a consequence, a simple graph search as BFS enables us to enumerate the hypercubes in $G$ in time $O(d2^dn)$.

\begin{lemma}[Enumeration of hypercubes~\cite{BeHa21}]
We can enumerate all triplets $(v,u,\widehat{\mathcal{E}})$, where $v$ is the anti-basis of a hypercube $Q$, $u$ its basis, and $\widehat{\mathcal{E}}$ the signature of $Q$ in time $O(d2^dn)$. Moreover, the list obtained fulfils the following partial order: if $d(v_0,v) < d(v_0,v')$, then any triplet $(v,u,\widehat{\mathcal{E}})$ containing $v$ appears before any triplet $(v',u',\widehat{\mathcal{E}}')$ containing $v'$.
\label{le:enum_hypercubes}
\end{lemma}

The enumeration of hypercubes is thus executed in linear time for median graphs with constant dimension. In summary, given any median graph, one can compute the set of $\Theta$-classes and their orthogonality relationship (for each $E_i$, the set of $\Theta$-classes orthogonal to $E_i$) in linear time, and the set of hypercubes with its basis, anti-basis and signature in $\tilde{O}(2^dn)$.

\textbf{Maximal POFs}. We terminate this preliminary section with a few words on maximal POFs, {\em i.e.} POFs $X$ such that there is no other POF $Y \supsetneq X$. There is a natural bijection between maximal POFs and maximal hypercubes, put in evidence by the following result.

\begin{theorem}[Maximal POFs and hypercubes]
For any maximal hypercube, its signature is a maximal POF.
Conversely, for any maximal POF $X$, there exists a unique hypercube of signature $X$. Its anti-basis is the vertex $v$ such that $\mathcal{E}^-(v) = X$. 
\label{th:maximal_pofs}
\end{theorem}
\begin{proof}
Let $Q$ be a maximal hypercube and $X_Q$ its signature. We begin with the proof that $X_Q$ is a maximal POF. Assume that $Y \supsetneq X_Q$. As $Y$ is a POF, there is a vertex $v_Y \in V$ satisfying $\mathcal{E}^-(v_Y) = Y$. Similarly, we denote by $v_X$ the vertex such that $\mathcal{E}^-(v_X) = X_Q$. Both $v_X$ and $v_Y$ belong to the boundary of any $\Theta$-class of $X_Q$ (the one which is the farthest from $v_0$). In brief, 
\begin{equation}v_X,v_Y \in \bigcap_{E_i \in X_Q} \partial H_i''
\label{eq:bigcap_gated}
\end{equation}
As every $\partial H_i''$ is gated, then the intersection written in Eq.~\eqref{eq:bigcap_gated} is convex/gated too. Thus, the shortest $(v_X,v_Y)$-path is entirely contained in set $\bigcap_{X_Q} \partial H_i''$. Let $(v_X,z)$ be the first edge of this path and $E_j$ the $\Theta$-class of this edge. Each $\Theta$-class form an isomorphism between its two boundaries (Lemma~\ref{le:boundaries}): as $z \in \bigcap_{X_Q} \partial H_i''$, there is a hypercube isomorphic to $Q$ in the boundary of $E_j$ containing $z$. Therefore, there is a hypercube of dimension $\card{X_Q}+1$ containing all vertices of $Q$. This yields a contradiction as $Q$ is supposed to be maximal.

Conversely, let $Q$ be a hypercube and assume that its signature $X_Q$ is a maximal POF. We suppose that there is a second hypercube $Q' \neq Q$ such that $X_Q = X_{Q'}$. Then, the set $\bigcap_{X_Q} \partial H_i''$ contains at least two elements: the anti-bases of hypercubes $Q$ and $Q'$. Using the same argument as above, we can put in evidence an edge with two endpoints in $\bigcap_{X_Q} \partial H_i''$. The $\Theta$-class of this edge is thus orthogonal of any class $E_i$ of $X_Q$ which defines an isomorphism between $\partial H_i'$ and $\partial H_i''$. Consequently, we obtain a POF superset of $X_Q$, a contradiction.

For any POF $X$, there is at least one hypercube with signature $X$ such that its anti-basis $v$ verifies $\mathcal{E}^-(v) = X$ according to Lemma~\ref{le:pof_hypercube}. In summary, any maximal POF $X$ can be associated with an unique hypercube of signature $X$.
\end{proof}

Figure~\ref{subfig:maximal_pof} illustrates this proof with two squares $Q$ and $Q'$ with the same signature $\set{E_i,E_h}$. One can observe the appearance of a hypercube of larger dimension containing $Q$, giving evidence of the non-maximality of $X_Q$.

The number of maximal hypercubes in a median graph is thus equal to the number of maximal POFs, which is itself at most linear in the number of vertices.

\section{Simplex graphs} \label{sec:simplex}

In this section, we present a combinatorial algorithm computing the diameter and all eccentricities of simplex graphs in quasilinear time, {\em i.e.} in $O((d^3+\log n)n)$. We begin with some properties of these graphs and then describe the algorithm.

\subsection{Simplex and crossing graphs: cliques, diameter and opposites} \label{subsec:crossing}

Given any undirected graph $G$, the vertices of the simplex graph $K(G)$ associated to $G$ represent the induced cliques (not necessarily maximal) of $G$. Two of these cliques are connected by an edge if they differ by exactly one element.

\begin{definition}[Simplex graphs~\cite{BaVe89}]
The simplex graph $K(G)=(V_K,E_K)$ of $G=(V,E)$ is made up of the vertex set $V_K = \set{C\subseteq V: C ~\emph{induced complete graph of}~ G}$ and the edge set $E_K = $\\ $\set{(C,C') : C,C' \in V_K, C \subsetneq C', \card{C'}-\card{C} = 1}$.
\label{def:simplex}
\end{definition}

Observe that the empty clique of $G$ also corresponds to a vertex $v_{\emptyset}$ in $K(G)$.  Therefore $v_{\emptyset}$ has degree $n = \card{V}$ in $K(G)$. Figure~\ref{fig:example_simplex} shows, as an example, the simplex graph of a $C_5$. Labels indicate the correspondence between vertices and cliques. We note that simplex graphs of cycles are cogwheels, {\em i.e.} wheels with a subdivision on each external edge.

The simplex graph of a $n$-complete graph is a hypercube of dimension $n$. More generally, simplex graphs are median~\cite{BaCh08,BaLeMo86}. We will see later than certain median graphs are not simplex graphs. As a subfamily of median graphs, the design of a subquadratic-time algorithm for the eccentricities on simplex graphs is of interest. Moreover, the dimension $d$ of simplex graphs can be unbounded because hypercubes belong to this class of graphs (their dimension is logarithmic in the size of the vertex set).

\begin{figure}[h]
\begin{subfigure}[b]{0.49\columnwidth}
\centering
\scalebox{0.7}{\begin{tikzpicture}


\node[draw, circle, minimum height=0.2cm, minimum width=0.2cm, fill=black] (P1) at (1,1) {};
\node[draw, circle, minimum height=0.2cm, minimum width=0.2cm, fill=black] (P2) at (1,3) {};
\node[draw, circle, minimum height=0.2cm, minimum width=0.2cm, fill=black] (P3) at (4,1) {};
\node[draw, circle, minimum height=0.2cm, minimum width=0.2cm, fill=black] (P4) at (4,3) {};
\node[draw, circle, minimum height=0.2cm, minimum width=0.2cm, fill=black] (P5) at (2.5,4.5) {};


\draw[line width = 1.4pt] (P1) -- (P2);
\draw[line width = 1.4pt] (P2) -- (P5);
\draw[line width = 1.4pt] (P5) -- (P4);
\draw[line width = 1.4pt] (P4) -- (P3);
\draw[line width = 1.4pt] (P3) -- (P1);

\node[scale=1.4] (E1) at (0.6,0.6) {$1$};
\node[scale=1.4] (E2) at (0.6,3.3) {$2$};
\node[scale=1.4] (E3) at (2.5,4.9) {$3$};
\node[scale=1.4] (E4) at (4.4,3.3) {$4$};
\node[scale=1.4] (E5) at (4.4,0.6) {$5$};

\end{tikzpicture}}
\end{subfigure}
\begin{subfigure}[b]{0.49\columnwidth}
\centering
\scalebox{0.7}{\begin{tikzpicture}


\node[draw, circle, minimum height=0.2cm, minimum width=0.2cm, fill=black] (P11) at (1,1) {};
\node[draw, circle, minimum height=0.2cm, minimum width=0.2cm, fill=black] (P12) at (1,2.5) {};

\node[draw, circle, minimum height=0.2cm, minimum width=0.2cm, fill=black] (P21) at (3,1) {};
\node[draw, circle, minimum height=0.2cm, minimum width=0.2cm, fill=black] (P22) at (3,2.5) {};
\node[draw, circle, minimum height=0.2cm, minimum width=0.2cm, fill=black] (P23) at (3,4) {};

\node[draw, circle, minimum height=0.2cm, minimum width=0.2cm, fill=black] (P31) at (5,1) {};
\node[draw, circle, minimum height=0.2cm, minimum width=0.2cm, fill=black] (P32) at (5,2.5) {};
\node[draw, circle, minimum height=0.2cm, minimum width=0.2cm, fill=black] (P33) at (5,4) {};

\node[draw, circle, minimum height=0.2cm, minimum width=0.2cm, fill=black] (P41) at (1.3,3.8) {};
\node[draw, circle, minimum height=0.2cm, minimum width=0.2cm, fill=black] (P42) at (1.3,5.3) {};
\node[draw, circle, minimum height=0.2cm, minimum width=0.2cm, fill=black] (P43) at (-0.7,3.8) {};


\draw[line width = 1.4pt] (P11) -- (P12);
\draw[line width = 1.4pt] (P11) -- (P21);
\draw[line width = 1.4pt] (P12) -- (P22);
\draw[line width = 1.4pt] (P21) -- (P22);

\draw[line width = 1.4pt] (P21) -- (P31);
\draw[line width = 1.4pt] (P22) -- (P32);
\draw[line width = 1.4pt] (P31) -- (P32);

\draw[line width = 1.4pt] (P22) -- (P23);
\draw[line width = 1.4pt] (P23) -- (P33);
\draw[line width = 1.4pt] (P32) -- (P33);

\draw[line width = 1.4pt] (P22) -- (P41);
\draw[line width = 1.4pt] (P12) -- (P43);
\draw[line width = 1.4pt] (P23) -- (P42);
\draw[line width = 1.4pt] (P41) -- (P43);
\draw[line width = 1.4pt] (P41) -- (P42);


\node (empty) at (3.3,2.8) {$\emptyset$};
\node (e1) at (0.6,2.1) {$\set{1}$};
\node (e2) at (0.9,4.2) {$\set{2}$};
\node (e3) at (3.3,4.4) {$\set{3}$};
\node (e4) at (5.5,2.5) {$\set{4}$};
\node (e5) at (2.6,0.6) {$\set{5}$};
\node (e12) at (-1.2,4.2) {$\set{1,2}$};
\node (e23) at (0.9,5.7) {$\set{2,3}$};
\node (e34) at (5.5,4.4) {$\set{3,4}$};
\node (e45) at (5.5,0.6) {$\set{4,5}$};
\node (e51) at (0.6,0.6) {$\set{1,5}$};

\end{tikzpicture}}
\end{subfigure}
\caption{A cycle $C_5$ and its simplex graph $K(C_5)$}
\label{fig:example_simplex}
\end{figure}
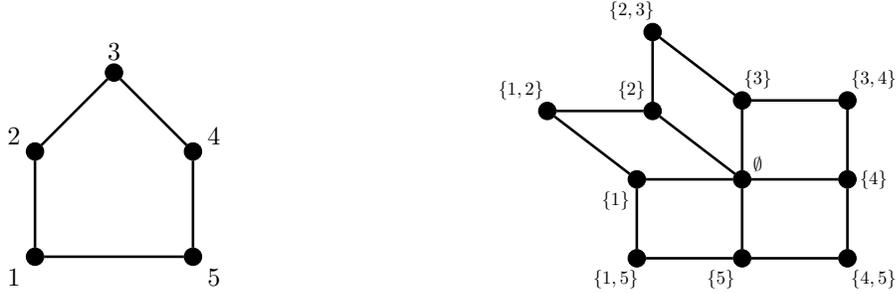

Simplex graphs can be characterized as particular median graphs.

\begin{theorem}
Let $G$ be a median graph. The following statements are equivalent:
\begin{description}
    \item[(1)] $G$ is a simplex graph.
    \item[(2)] There is a vertex $v_0 \in V(G)$ such that each $\Theta$-class of $G$ is adjacent to $v_0$, {\em i.e.} $\forall 1\le i\le q, \exists v_i\in V(G), v_0v_i \in E_i$.
    \item[(3)] There is a vertex $v_0 \in V(G)$ contained in any maximal hypercube of $G$.
\end{description}
\label{th:simplex}
\end{theorem}
\begin{proof}
(2)$\Leftrightarrow$(3). Assume each $\Theta$-class is adjacent to $v_0$. For each POF $X$, there is a hypercube with signature $X$ and containing $v_0$ (Lemma~\ref{le:pof_adjacent}). For each maximal hypercube, its signature (the set of $\Theta$-classes it contains) is a maximal POF and, moreover, no other maximal hypercube has the same signature (Theorem~\ref{th:maximal_pofs}). For this reason, each maximal hypercube necessarily contains $v_0$. 
Conversely, if any maximal hypercube contains $v_0$, as each $\Theta$-class belongs to at least one maximal POF, then each $\Theta$-class is necessarily adjacent to $v_0$.

(2)$\Rightarrow$(1). We consider a median graph $G$ such that all $\Theta$-classes are adjacent to $v_0$. Our objective is to prove that there exists $G'$ such that $G = K(G')$. Let $G'$ be the graph where its vertices represent the $\Theta$-classes of $G$ and two of them are connected by an edge if the $\Theta$-classes are orthogonal. In this way, every clique of $G'$ (even the empty one) corresponds to a POF of $G$. For any POF $X$ of $G$, its $\Theta$-classes are adjacent to $v_0$, so there exists a hypercube containing $v_0$ with signature $X$, according to Lemma~\ref{le:pof_adjacent}. Moreover, given a POF $X$, this hypercube is unique. Its anti-basis (opposite of $v_0$) thus represents the clique $X$ in $G'$. Conversely, according to the $v_0$-orientation, each vertex admits its own set of incoming $\Theta$-classes which forms a POF (Lemma~\ref{le:pof_hypercube}). Therefore, we can associate to each vertex of $G$ a clique of $G'$. 

Then, two vertices $u,v$ of $G$ are adjacent if and only if $\mathcal{E}^-(v) = \mathcal{E}^-(u) \cup \set{E_i}$, for some class $E_i$. On one hand, suppose $uv \in E_i$. Assume that there exists $E_j \in \mathcal{E}^-(u) \backslash \mathcal{E}^-(v)$. Then, $E_j$ is parallel to $E_i$, otherwise $E_j \in \mathcal{E}^-(v)$ (Lemma~\ref{le:squares}). This is a contradiction: as $H_i'' \subsetneq H_j''$, $E_i$ cannot be adjacent to $v_0$. Moreover, if $E_h \in \mathcal{E}^-(v) \backslash \mathcal{E}^-(u)$, $h \neq i$, then $E_h$ is necessarily incoming into $u$ because of Lemma~\ref{le:squares}, as $E_h$ and $E_i$ are orthogonal. So, $\mathcal{E}^-(v) \backslash \mathcal{E}^-(u) = \set{E_i}$. On the other hand, if $\mathcal{E}^-(v) = \mathcal{E}^-(u) \cup \set{E_i}$, let $v'$ be the vertex such that $v'v \in E_i$. Again, Lemma~\ref{le:squares} implies that all $\Theta$-classes incoming into $v$ are also incoming into $v'$, so $\mathcal{E}^-(u) \subseteq \mathcal{E}^-(v')$. Furthermore, if there is some $E_j \notin \mathcal{E}^-(u)$ incoming into $v'$, then it should be parallel to $E_i$, a contradiction. Hence, $v' = u$.

(1)$\Rightarrow$(2). Assume that $G=K(G')$: we denote by $v_{\emptyset}$ the vertex representing the empty clique. By contradiction, we suppose that there exists a $\Theta$-class which is not adjacent to $v_{\emptyset}$: we denote it by $E_1$. We consider the $v_{\emptyset}$-orientation of the graph. As $\set{E_1}$ is a POF, there is necessarily one vertex $v_1$ with only one incoming edge belonging to $E_1$. As $v_1 \neq v_{\emptyset}$, vertex $v_1$ represents a clique of $G'$ of size at least 1. If $v_1$ represents a clique of size exactly 1, then it is adjacent to $v_{\emptyset}$ because only one element differ between the cliques represented by both $v_{\emptyset}$ and $v_1$, which is a contradiction. If $v_1$ represents a clique of size at least 2, then it has at least two incoming edges. another contradiction. In summary, if $G=K(G')$, each $\Theta$-class of $G$ is adjacent to $v_{\emptyset}$.
\end{proof}

In this section only, on simplex graphs, the canonical basepoint $v_0$ is not selected arbitrarily. We fix $v_0$ as a vertex adjacent to all $\Theta$-classes, as put in evidence by Theorem~\ref{th:simplex}. We call $v_0$ the \textit{central vertex} of the simplex graph.

What Theorem~\ref{th:simplex} says is that any simplex graph can be seen as a set of maximal POFs (or hypercubes) that are ``medianly'' assembled. Indeed, one cannot define a simplex graph given any collection of sets - excluding subsets - representing maximal POFs. Consider as an example the collection $\set{\set{E_1,E_2},\set{E_2,E_3},\set{E_3,E_1}}$: it would produce a simplex graph with 3 squares with basis $v_0$ and sharing pairwise an edge. This graph is $Q_3^-$ (the 3-cube minus a vertex) and is not median. The collection implies that $E_1,E_2,E_3$ are pairwise orthogonal, so $\set{E_1,E_2,E_3}$ should be the maximal POF here.

The most obvious example of median graph which is not simplex is certainly the path $P_4$. Indeed, it has three $\Theta$-classes which are all pairwise parallel. For any vertex of $P_4$, there exists a $\Theta$-class which is not adjacent to it.

The proof (2)$\Rightarrow$(1) reveals the reverse application of $K$. 

\begin{definition}[Crossing graphs~\cite{BaVe89,KlMu02}]
Let $G$ be a median graph. Its \textit{crossing graph} $G^{\#}$ is the graph obtained by considering $\Theta$-classes as its vertices and such that two $\Theta$-classes are adjacent if they are orthogonal.
\label{def:crossing}
\end{definition}

Restricted to simplex graphs, this transformation is the reverse of $K$: indeed, as stated in~\cite{KlMu02}, $G = K(G)^{\#}$. The clique number of $G^{\#}$ is exactly the dimension of median graph $G$. For example, the crossing graph of a cube-free median graph contains no triangle. Each simplex graph admits a central vertex ($v_0$ in Theorem~\ref{th:simplex}) which represents the empty clique of $G^{\#}$.

Now, we focus on the problem of determining a diametral pair of a simplex graph $G$ and more generally all eccentricities. Observe that the distance between the central vertex $v_0$ and any vertex $u$ of $G$ can be deduced directly from the edges incoming into $u$. We state that $\sigma_{v_0,u} = \mathcal{E}^-(u)$. This is a consequence of Theorem~\ref{th:simplex}: all $\Theta$-classes of $\mathcal{E}^-(u)$ are adjacent to $v_0$, so $v_0$ is the basis of the hypercube with signature $\mathcal{E}^-(u)$ and anti-basis $u$. A shortest $(v_0,u)$-path is thus made up of edges of this hypercube. The distance $d(v_0,u)$ is equal to its dimension: $d(v_0,u) = \card{\mathcal{E}^-(u)}$.

A key result is the fact that the central vertex $v_0$ of the simplex graph belongs to the interval $I(u,v)$ of any pair $u,v$ satisfying $d(u,v) = \ecc(u)$.

\begin{lemma}
Let $u,v \in V(G)$ such that $d(u,v) = \ecc(u)$. Then, $v_0 \in I(u,v)$.
\label{le:center_simplex}
\end{lemma}
\begin{proof}
Let $H$ be the graph such that $G = K(H)$. The distance $d(u,v)$ is equal to the number of $\Theta$-classes of $G$ separating them, so $d(u,v) = \card{\mathcal{E}^-(u) \bigtriangleup \mathcal{E}^-(v)}$. Let $w$ be the vertex of $G$ such that $\mathcal{E}^-(w) = \mathcal{E}^-(v) \backslash \mathcal{E}^-(u)$. If $\mathcal{E}^-(u) \cap \mathcal{E}^-(v) \neq \emptyset$, then $d(u,w) = \card{\mathcal{E}^-(u) \bigtriangleup \mathcal{E}^-(w)} = \card{\mathcal{E}^-(u)} + \card{\mathcal{E}^-(w)} > \card{\mathcal{E}^-(u) \backslash \mathcal{E}^-(v)} + \card{\mathcal{E}^-(v) \backslash \mathcal{E}^-(u)} = d(u,v)$. So, $d(u,w) > \ecc(u)$, a contradiction.
%
%
\end{proof}

\begin{figure}[h]
\centering
\scalebox{0.8}{\begin{tikzpicture}


\node[draw, circle, minimum height=0.2cm, minimum width=0.2cm, fill=blue] (P2) at (2.5,3.5) {};

\node[draw, circle, minimum height=0.2cm, minimum width=0.2cm, fill=black] (P3) at (4,3.1) {};

\node[draw, circle, minimum height=0.2cm, minimum width=0.2cm, fill=black] (P4) at (5.5,2.8) {};

\node[draw, circle, minimum height=0.2cm, minimum width=0.2cm, fill=black] (P5) at (7,2.5) {};
\node[draw, circle, minimum height=0.2cm, minimum width=0.2cm, fill=black] (v0) at (7,1.0) {};

\node[draw, circle, minimum height=0.2cm, minimum width=0.2cm, fill=black] (P6) at (8.5,2.8) {};
\node[draw, circle, minimum height=0.2cm, minimum width=0.2cm, fill=black] (P6b) at (8.5,1.3) {};

\node[draw, circle, minimum height=0.2cm, minimum width=0.2cm, fill=black] (P7) at (10.0,3.1) {};
\node[draw, circle, minimum height=0.2cm, minimum width=0.2cm, fill=black] (P7b) at (10.0,1.6) {};

\node[draw, circle, minimum height=0.2cm, minimum width=0.2cm, fill=blue] (P8) at (11.5,3.5) {};
\node[draw, circle, minimum height=0.2cm, minimum width=0.2cm, fill=black] (P8b) at (11.5,2.0) {};


\draw[->,line width = 1.4pt, color = red] (v0) -- (P5);

\draw[<-,line width = 1.4pt] (P2) -- (P3);
\draw[<-,line width = 1.4pt] (P3) -- (P4);
\draw[<-,line width = 1.4pt] (P4) -- (P5);
\draw[->,line width = 1.4pt] (P5) -- (P6);
\draw[->,line width = 1.4pt, color = red] (P6b) -- (P6);
\draw[->,line width = 1.4pt] (P6) -- (P7);
\draw[->,line width = 1.4pt, color = red] (P7b) -- (P7);
\draw[->,line width = 1.4pt] (P7) -- (P8);
\draw[->,line width = 1.4pt, color = red] (P8b) -- (P8);


\node[scale=1.2, color = red] at (6.5,1.7) {$E_i$};

\node[scale = 1.2] at (7.4,0.8) {$v_0$};
\node[scale = 1.2] at (7.3,2.9) {$m$};
\node[scale = 1.2, color = blue] at (2.3,3.9) {$u$};
\node[scale = 1.2, color = blue] at (11.7,3.9) {$v$};
\node[scale = 1.2] at (11.7,1.6) {$w$};

\end{tikzpicture}}
\caption{Illustration of the contradiction in the proof of Lemma~\ref{le:center_simplex}.}
\label{fig:center_simplex}
\end{figure}
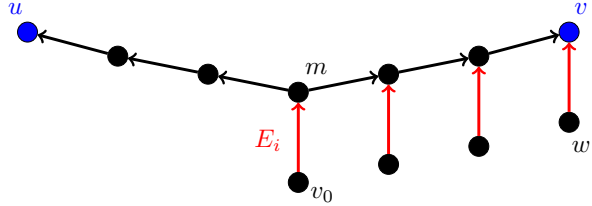

Figure~\ref{fig:center_simplex} shows a simple example with $m=m(u,v,v_0)\neq v_0$ and $d(v_0,m) = 1$. The edges are oriented according to the $v_0$-orientation. The $\Theta$-class $E_i$ in $\sigma_{v_0,m}$ is present alongside the interval $I(m,v)$. The contradiction of the previous proof comes from the fact that the shortest $(u,v)$-path could be extended with the vertex $w$ which is the neighbor of $v$ in $\partial H_i'$.

Two vertices $u,v$ forming a diametral pair cannot share a common incoming $\Theta$-class $E_i$, in other words $\mathcal{E}^-(u) \cap \mathcal{E}^-(v) = \emptyset$, otherwise $m = m(u,v,v_0) \in I(u,v) \subseteq H_i''$ and $v_0 \in H_i'$. Moreover, the distance $d(u,v)$ is exactly $\card{\mathcal{E}^-(u)} + \card{\mathcal{E}^-(v)}$ because $\card{\mathcal{E}^-(u)} = d(v_0,u)$ and $\card{\mathcal{E}^-(v)} = d(v_0,v)$. So, determining the diameter of a simplex graph $G$ is equivalent to maximizing the sum $\card{X} + \card{Y}$, where $X$ and $Y$ are two POFs of $G$ that are disjoint. Computing the diameter is equivalent to find the largest pair of disjoint cliques in the crossing graph $G^{\#}$. Similarly, the eccentricity of a vertex $u$ is exactly the size $\card{\mathcal{E}^-(u)} + \card{\mathcal{E}^-(v)}$ of the largest pair of disjoint POFs $(\mathcal{E}^-(u),\mathcal{E}^-(v))$.
Now, we can define the notion of \textit{opposite}.

\begin{definition}
Let $G$ be a simplex graph and $X$ a POF of $G$. We denote by $\opp(X)$ the \textit{opposite} of $X$, {\em i.e.} the POF~ $Y$ disjoint from $X$ with the maximum cardinality.
\[
\opp(X) = \argmax_{Y \cap X = \emptyset} \card{Y}.
\]
\label{def:opposite}
\end{definition}

With this definition, the eccentricity of a vertex $u$, if we fix $X_u = \mathcal{E}^-(u)$, is written $\ecc(u) = \card{X_u} + \card{\opp(X_u)}$. Hence, the diameter of the simplex graph $G$ can be written as the size of the largest pair POF-opposite: $\diam(G) = \max_{X\in \mathcal{L}}( \card{X} + \card{\opp(X)})$.

We propose now the definition of two problems on simplex graphs. The first one, called \textsc{Opposites} (OPP) consists in finding all pairs POF-opposite. Its output has thus a linear size. Given the solution of OPP on graph $G$, one can deduce both the diameter and all eccentricities in $O(n)$ time with the formulae presented above.

\begin{definition}[OPP]~

\textbf{Input}: Simplex graph $G$, central vertex $v_0$.

\textbf{Output}: For each POF $X$, its opposite $\opp(X)$.
\label{def:dpp}
\end{definition}

We define an even larger version of the problem where a positive integer weight is associated with each POF. We call it \textsc{Weighted Opposites} (WOPP).

\begin{definition}[WOPP]~

\textbf{Input}: Simplex graph $G$, central vertex $v_0$, weight function $\omega : \mathcal{L} \rightarrow \mathbb{N}^+$.

\textbf{Output}: For each POF $X$, its weighted opposite $Y$ maximizing $\omega(Y)$ such that $X \cap Y = \emptyset$.
\label{def:wdpp}
\end{definition}

Obviously, OPP is a special case of WOPP when $\omega$ is the cardinality function. In Section~\ref{subsec:partitioning}, we show that WOPP can be solved in quasilinear time $O((d^3+\log n)n)$. As a consequence, all eccentricities of a simplex graph $G$ can also be determined with such time complexity. Moreover, we will see in Section~\ref{sec:subquadratic} that solving WOPP in quasilinear time implies that all eccentricities of any median graph can be computed with a simple exponential time $2^{O(d)}n$, improving the slightly super-exponential time proposed in~\cite{BeHa21}.

\subsection{Quasilinear algorithm for all eccentricities in simplex graphs} \label{subsec:partitioning}

We propose an algorithm solving WOPP in quasilinear time $\tilde{O}(n)$. 

\begin{theorem}
There is a combinatorial algorithm solving WOPP in time $O((d^3+\log n)n)$. 
\label{th:solving_wopp}
\end{theorem}

Thus, we can compute the diameter and all eccentricities of simplex graphs in quasilinear time, even when the dimension is not bounded.

\begin{corollary}
There is a combinatorial algorithm determining all eccentricities of a simplex graph in time $O((d^3+\log n)n)$.
\label{co:linear_simplex}
\end{corollary}

The entire subsection is the proof of Theorem~\ref{th:solving_wopp}. We consider a simplex graph $G$ with a central vertex $v_0$ and a weight function $\omega : \mathcal{L} \rightarrow \mathbb{N}^+$. The algorithm is presented in Sections~\ref{subsubsec:tree_opp} and~\ref{subsubsec:cp} and its analysis in Section~\ref{subsubsec:analysis_wopp}.

\subsubsection{Tree structure of the opposites} \label{subsubsec:tree_opp}

The first step of our algorithm consists in building a binary tree $T$. Tree $T$ is a representation of a partition refinement procedure over $\mathcal{L}$. We remind the reader that partition refinement is a powerful algorithmic technique leading to the design of linear-time algorithms for many well-known problems~\cite{HaPaVi99,PaTa87}. It consists in successive partitionings of a collection of sets.

In our context, the collection which is splitted is the set of POFs $\mathcal{L}$ of the simplex graph $G$. The vertices $a \in V(T)$ of tree $T$, called \textit{nodes}, represent the sets obtained from the successive partitionings. They are indexed with POFs $L_a \in \mathcal{L}$. The edges of $T$ are indexed with a pair class-boolean. For each $a \in V(T)$ with two children, the two edges connecting it to its children are both indexed with the same $\Theta$-class but not the same boolean. This $\Theta$-class is denoted by $E\left[a\right]$. To improve the readibility, index $(E\left[a\right],\mbox{true})$ will be denoted by $+E\left[a\right]$ while $(E\left[a\right],\mbox{false})$ becomes $-E\left[a\right]$.

We denote by $\mathcal{L}\left[E_i\right]$ the \textit{adjacency list} of $\Theta$-class $E_i$, {\em i.e.} the list of POFs in $\mathcal{L}$ which contain $E_i$. As at most $d$ classes belong to any POF, the total size of all adjacency lists is upper-bounded by $dn = d\card{\mathcal{L}}$.

We sort the POFs in $\mathcal{L}$ in function of their weights $\omega(L)$ in the decreasing order. This takes $O(n\log n)$. We denote by $\tau_{\omega}$ this ordering. Let $L_0$ be the maximum-weighted POF in $\mathcal{L}$, in other words the first POF in ordering $\tau_{\omega}$.

The description of the construction of $T$ by partition refinement begins. We assign an arbitrary ordering to the $\Theta$-classes which are in $L_0$, for example based on their index. Let $E_{i_1}, E_{i_2},\ldots, E_{i_r}$ be the classes of $L_0$ ordered, $\card{L_0} = r$. First, we split $\mathcal{L}$ in two sets: one with POFs containing $E_{i_1}$, the other with POFs which do not contain $E_{i_1}$. In brief, we obtain $\mathcal{L}\left[E_{i_1}\right]$ and its complementary. Second, we split each of these two sets regarding $\Theta$-class $E_{i_2}$: on one side the POFs containing $E_{i_2}$, on the other side POFs without $E_{i_2}$. We pursue in this way with all classes of $L_0$. At the end, there are at most $2^{\card{L_0}}$ sets in the partition.

Until now, we obtained the top $r$ depths of tree $T$. Before pursuing the construction, we define some notations. Let $a_0$ be the root of $T$. Node $a_0$ represents the entire collection $\mathcal{L}$. Let $\Omega$ be the function indicating the sets represented by each node $a \in V(T)$. We have $\Omega(a_0) = \mathcal{L}$. For any $a \in T$, its index $L_a$ is defined as the maximum-weighted POF of its \textit{universe} $\Omega(a)$. So, $L_{a_0} = L_0$. The root $a_0$ has two children, one representing $\mathcal{L}\left[E_{i_1}\right]$ and the other one its complementary as they are the result of a partition refinement from $E_{i_1}$. The edges connecting $a_0$ and its children are indexed by $+E_{i_1}$ and $-E_{i_1}$ respectively. 

We denote by $R(a)$ the set of edge indices of the simple path from $a$ to the root $a_0$. For example, for node $a$ with an universe $\Omega(a)$ being the set of POFs which contain  $E_{i_2}$ but not $E_{i_1}$, we have $R(a)=\set{E_{i_1},E_{i_2}}$. We write $R(a) = R^+(a) \cup R^-(a)$, where $R^+(a)$ contains indices with boolean true, while $R^-(a)$ contains indices with boolean false. In this example, $R^+(a) = \set{E_{i_2}}$ and $R^-(a) = \set{E_{i_1}}$. Set $R^+(a)$ is a POF: it contains $\Theta$-classes which are pairwise orthogonal, otherwise $R^+(a)$ would not be the subset of a POF belonging to $\mathcal{L}$.

We pursue the construction of $T$. For each leaf $a$ of the current tree, we execute the following process: we only consider the $\Theta$-classes of $L_a$ which have not been treated earlier, {\em i.e.} which are not in $R^+(a)$. We order them arbitrarily and we split the universe $\Omega(a)$ successively. We pursue with the new leaves obtained, etc. In this way, tree $T$ can be seen as a stacking of small binary trees (depth at most $d$) that we call \textit{blocks}. For example, the root $a_0$ of $T$ belongs to the top block which is produced from the partition refinement over $\Theta$-classes of $L_0$. In Figure~\ref{subfig:first_layer}, we represent some nodes $a$ of this block and their index $L_a$ with $\card{L_0} = 3$: $i_1 = i$, $i_2 = j$, and $i_3 = \ell$. The leaves of this block both belong to the top block and are the roots of another block below this one.

\begin{figure}[h]
\centering
\begin{subfigure}[b]{0.40\columnwidth}
\centering
\scalebox{0.7}{\begin{tikzpicture}


\draw [color = black, fill = white] (6.0,10.0) -- (6.0,11.0) -- (9.0,11.0) -- (9.0,10.0) --  (6.0,10.0);
\draw [color = black, fill = white] (3.5,8.0) -- (3.5,9.0) -- (6.5,9.0) -- (6.5,8.0) -- (3.5,8.0);
\draw [color = black, fill = white] (5.5,6.0) -- (5.5,7.0) -- (8.5,7.0) -- (8.5,6.0) -- (5.5,6.0);
\draw [color = black, fill = white] (3.0,4.0) -- (3.0,5.0) -- (6.0,5.0) -- (6.0,4.0) -- (3.0,4.0);


\node (P1) at (7.5,10.5) {$\set{E_i,E_j,E_{\ell}}$};
\node (P2) at (5.0,8.5) {$\set{E_j,E_{\ell},E_q}$};
\node (P3) at (7.0,6.5) {$\set{E_j,E_{\ell},E_q}$};
\node (P4) at (4.5,4.5) {$\set{E_j,E_r}$};


\node[scale=1.1, color = red] at (9.3,11.1) {$a_0$};
\node[scale=1.1, color = red] at (6.8,9.1) {$a_1$};
\node[scale=1.1, color = red] at (8.8,7.1) {$a_2$};
\node[scale=1.1, color = red] at (6.3,5.1) {$a_3$};

\node[scale=1.1] at (5.5,9.8) {$-E_i$};
\node[scale=1.1] at (9.5,9.8) {$+E_i$};

\node[scale=1.1] at (3.0,7.8) {$-E_j$};
\node[scale=1.1] at (7.0,7.8) {$+E_j$};

\node[scale=1.1] at (9.0,5.8) {$+E_{\ell}$};
\node[scale=1.1] at (5.0,5.8) {$-E_{\ell}$};


\draw[rounded corners=5pt,line width = 1.4pt] (6.0,10.3) -| (5.0,9.0);
\draw[->,>=latex,rounded corners=5pt,line width = 1.4pt] (3.5,8.3) -| (2.5,7.0);

\draw[->,>=latex,rounded corners=5pt,line width = 1.4pt] (9.0,10.3) -| (10.0,9.0);

\draw[rounded corners=5pt,line width = 1.4pt] (6.5,8.3) -| (7.5,7.0);

\draw[rounded corners=5pt,line width = 1.4pt] (5.5,6.3) -| (4.5,5.0);
\draw[->,>=latex,rounded corners=5pt,line width = 1.4pt] (8.5,6.3) -| (9.5,5.0);

\end{tikzpicture}}
\caption{Some nodes of $T$ obtained with the partition refinement over $L_0$}
\label{subfig:first_layer}
\end{subfigure}
\begin{subfigure}[b]{0.59\columnwidth}
\centering
\scalebox{0.7}{\begin{tikzpicture}

\draw [color = black] (6.2,7.5) -- (1.0,4.5) -- (11.4,4.5) -- (6.2,7.5);

\draw [color = black] (2.0,5.0) -- (1.5,2.5) -- (2.5,2.5) -- (2.0,5.0);
\draw [color = black] (3.2,5.0) -- (2.7,2.5) -- (3.7,2.5) -- (3.2,5.0);
\draw [color = black] (4.4,5.0) -- (3.9,2.5) -- (4.9,2.5) -- (4.4,5.0);
\draw [color = black] (5.6,5.0) -- (5.1,2.5) -- (6.1,2.5) -- (5.6,5.0);
\draw [color = black] (6.8,5.0) -- (6.3,2.5) -- (7.3,2.5) -- (6.8,5.0);
\draw [color = black] (8.0,5.0) -- (7.5,2.5) -- (8.5,2.5) -- (8.0,5.0);
\draw [color = black] (9.2,5.0) -- (8.7,2.5) -- (9.7,2.5) -- (9.2,5.0);
\draw [color = black] (10.4,5.0) -- (9.9,2.5) -- (10.9,2.5) -- (10.4,5.0);

\draw [color = black] (7.0,3.0) -- (6.6,1.5) -- (7.4,1.5) -- (7.0,3.0);

\draw [color = red, line width = 1.4pt, dashed] (6.2,7.9) -- (0.5,4.5) -- (0.5,2.3) -- (11.9,2.3) -- (11.9,4.5) -- (6.2,8.0);


\node[draw, circle, fill=black, scale=0.6] (a0) at (6.2,7.0) {};
\node[draw, circle, fill=black, scale=0.6] (a1) at (4.9,6.2) {};
\node[draw, circle, fill=black, scale=0.6] (a2) at (5.0,5.4) {};
\node[draw, circle, fill=black, scale=0.6] (a3) at (4.4,4.6) {};
\node[draw, circle, fill=black, scale=0.6] (a) at (7.0,2.1) {};

\draw (a0) -- (a1);
\draw[dashed] (a0) -- (7.5,6.2);
\draw (a1) -- (a2);
\draw[dashed] (a1) -- (4.2,5.4);
\draw (a2) -- (a3);
\draw[dashed] (a2) -- (5.6,4.6);


\node[color = red] at (8.8,7.0) {$T_m^{(2)}$};

\node at (6.2,6.7) {$a_0$};
\node at (5.3,6.1) {$a_1$};
\node at (5.4,5.3) {$a_2$};
\node at (4.1,4.8) {$a_3$};
\node at (7.1,1.8) {$a$};

\end{tikzpicture}}
\caption{Some blocks of $T$: all blocks of layer 0 and 1, one block of layer 2. Node $a$ has layer 2.}
\label{subfig:layers}
\end{subfigure}

\caption{An example of tree $T$: its first block and structure}
\label{fig:tree_2}
\end{figure}
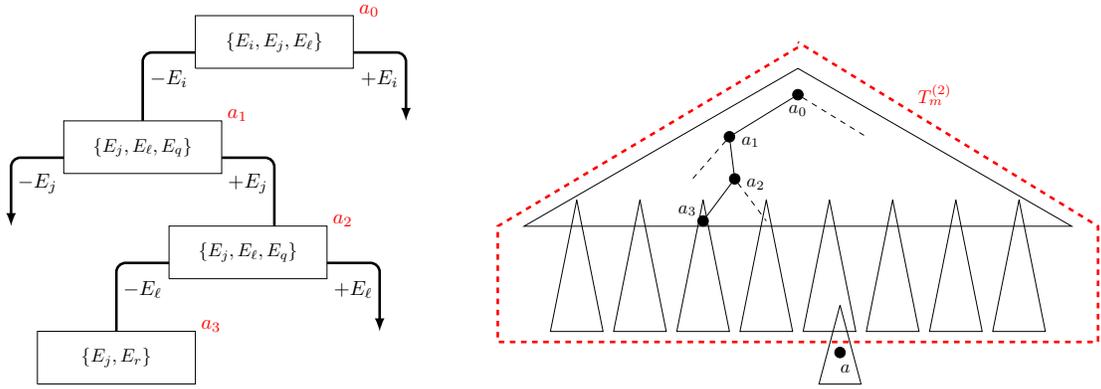

Algorithm~\ref{algo:refinement} provides us with the pseudocode of a partition refinement procedure. We call it Ordered Internal Partition Refinement (OIPR). We execute OIPR to obtain tree $T$ with the following input: the ground set contains the $\Theta$-classes of $G$ ($W = \mathcal{E}$), the collection is made up of the POFs outgoing from $m$ ($\mathcal{S} = \mathcal{L}$) and the ordering of the sets come from the POFs weights ($\tau = \tau_{\omega}$). Each while loop (line~\ref{line:while_loop}) corresponds to the construction of a block of $T$. Such step starts by picking up the first element, say $L_a$, of a non-singleton set of $\mathcal{P}$, say $\Omega(a)$. Then, we consider the $\Theta$-classes of $L_a$ which have not been locally visited, {\em i.e.} which do not belong to $R^+(a)$ (line~\ref{line:local_visit}). For each of these classes $E_j$, we refine the non-singleton $\Omega(a)$: we split it successively with the POFs containing $E_j$ and the POFs not containing $E_j$ (line~\ref{line:refine}).

\begin{algorithm}[h]
\SetKwFor{For}{for}{do}{\nl endfor}
\SetKwFor{Forall}{for all}{do}{\nl endfor}
\SetKwIF{If}{ElseIf}{Else}{if}{then}{else if}{else}{}
\DontPrintSemicolon
\SetNlSty{}{}{:}
\SetAlgoNlRelativeSize{0}
\SetNlSkip{1em}
\nl\KwIn{ground set $W=\set{w_1,\ldots,w_q}$, collection $\mathcal{S} = \set{S_1,\ldots,S_N}$, $S_i \subseteq W$ for any $1\le i\le N$, and an ordering $\tau$ of $\mathcal{S}$.}
\nl\KwOut{An ordered partition $\mathcal{P}$ of $\mathcal{S}$ made up of singletons.}
\nl Initialize $\mathcal{P} \leftarrow \set{\mathcal{S}}$, partition with a single set;\;
\nl \While{there exists a part of $\mathcal{P}$ which is not a singleton \label{line:while_loop}}{
	\nl $Q \leftarrow$ first non-singleton of the ordered partition $\mathcal{P}$;\;
	\nl $A \leftarrow$ first element of $Q$ according to $\tau$;\;

\nl \For{every $w_j \in A$ non locally visited\label{line:local_visit}}{
	\nl Substitute $Q$ in $\mathcal{P}$ by $\textsf{Refine}(Q,\set{w_j})$\label{line:refine};\;
	}	
}
\caption{Ordered Internal Partition Refinement (OIPR)}
\label{algo:refinement}
\end{algorithm}

The time needed to run Algorithm~\ref{algo:refinement}, using a doubly linked list data structure, depends on the number of appearances of each element $w_j$ of the ground set into the collection $\mathcal{S}$.

\begin{lemma}[Execution time of OIPR~\cite{HaPaVi99}]
Let $M(w_j)$ be the number of sets $S_i \in \mathcal{S}$ such that $w_j \in S_i$. Then, OIPR runs in $O(\sum_{j=1}^q M(w_j))$.
\label{le:OIPR}
\end{lemma}

In our context, values $M(w_j)$ are the sizes of adjacency lists $\mathcal{L}\left[E_j\right]$. We explained above why the total size of these adjacency lists could not exceed $dn$.

\begin{corollary} OIPR applied with $W = \mathcal{E}$, $\mathcal{S} = \mathcal{L}$ and $\tau = \tau_{\omega}$ runs in $O(n(d+\log n))$.
\label{co:oipr}
\end{corollary}
\begin{proof}
The time needed to sort all POFs according to their weights, {\em i.e.} in order $\tau_{\omega}$, is $O(n\log n)$. Moreover, OIPR runs in $O(dn)$  according to Lemma~\ref{le:OIPR}.
\end{proof}

Each set obtained from a refinement is represented by a node of $T$. Its children are obtained from a refinement of $\Omega(a)$ with some $\Theta$-class $E_j=E\left[a\right]$: one represents the elements of $\Omega(a)$ containing $E_j$ (the edge from $a$ to this child is indexed with $+E_j$), the other represents the complementary (the edge from $a$ to this child is indexed with $-E_j$).

Now, we give some notations and properties related to the tree $T$. At least one partitioning is executed at each depth of tree $T$ so, as $\mathcal{L}$ is finite, $T$ is too. Its depth is at most $n = \card{\mathcal{L}}$. 
We say the \textit{layer} of a node $a$ is the number of blocks we pass through when we traverse the simple path between the root and $a$ minus 1. For example, the root has layer 0 and the leaves of the top block have layer 1. We denote by $T^{(j)}$, $j\ge 1$, the subtree of $T$ induced on the set made up of (i) nodes of layer at most $j-1$ and (ii) nodes of layer $j$ which are the roots of a block. In this way, nodes of layer $j$ in $T^{(j)}$ are leaves. The depth of $T^{(j)}$ is at most $jd$, as the depth of each block is at most $d$. 

Certain nodes $a \in V(T)$ may admit only one child. This situation occurs when $E\left[a\right]$ is not orthogonal to at least one class of $R^+(a)$. Indeed, each POF of $\Omega(a)$ contains all $\Theta$-classes of $R^+(a)$: if $R^+(a) \cup \set{E\left[a\right]}$ is not a POF, then there is no POF superset of $R^+(a) \cup \set{E\left[a\right]}$. 

For any node $a$, we store its universe $\Omega(a)$. With the doubly linked list data structure for each partitioning, we can preserve the original ordering of the POFs. Consequently, in all sets $\Omega(a)$, the POFs are sorted in function of their weight. For any $a \in V(T)$, the maximum-weighted POF of $\Omega(a)$, {\em i.e.} index $L_a$, is thus obtained with constant running time.

Even if the execution time of OIPR is quasilinear in $n$, the extra time needed to store the tree $T$ and particularly all sets $\Omega(a)$ may not be linear in $n$. In the remainder, we will see that considering tree $T^{(d)}$ suffices to solving our problem. In this way, storing $T^{(d)}$ becomes quasilinear in $n$, such as the execution of the partitioning.

\subsubsection{Constraint pairs} \label{subsubsec:cp}

The second step of our algorithm uses a data structure called \textit{constraint pair}, whose definition is based on tree $T$. We provide a dynamic programming (DP) algorithm which computes one value per constraint pair. At the end of the execution, we can deduce the opposite of each POF in $O(1)$.

\begin{definition}[Constraint pair]
A constraint pair $(a,X)$ is made up of a node $a \in V(T)$ and a POF $X$ such that {\em (i)} $X \cap R(a) = \emptyset$ and {\em (ii)} $X \cup R^+(a) \in \mathcal{L}$.
\label{def:cp}
\end{definition}

The existence of a constraint pair $(a,X)$ implies that no class $X$ is present in the edge indices from $a_0$ to $a$ in $T$. Moreover, each $\Theta$-class of $X$ is orthogonal to all $\Theta$-classes of $R^+(a)$. We denote by $\mathcal{C}$ the set of constraint pairs. Let $\mathcal{C}^{(j)}$ be the set of constraint pairs $(a,X)$ such that $a\in T^{(j)}$: we have $\mathcal{C}^{(j)} \subseteq \mathcal{C}^{(j+1)} \subseteq \mathcal{C}$.

For any constraint pair $(a,X)$, we denote by $h(a,X)$ a POF $X^*$ disjoint from $X$ in $\Omega(a)$ with the maximum weight. For any POF $L \in \mathcal{L}$, the POF $h(a_0,L)$ is an opposite of $L$ because $\Omega(a_0) = \mathcal{L}$, so we can write, according to Definition~\ref{def:opposite}, $h(a_0,L) = \opp(L)$. Observe that any pair $(a_0,L)$ is a constraint pair as $R(a_0) = \emptyset$. We present a method to compute all opposites $h(a_0,L)$.

The description of the DP algorithm starts. For any constraint pair $(a,X)$, we denote by $a^+$ (resp. $a^-$) the child of $a$ which is the endpoint of the edge indexed with $+E\left[a\right]$ (resp. $-E\left[a\right]$) in $T$ (Figure~\ref{fig:dag}). 
We define the set $\mathcal{C}(a,X)$ which describes the recursive calls needed to compute $h(a,X)$. Formally, the objective is to make sure that value $h(a,X)$ is exactly a function of all $h(a',X')$, where $(a',X') \in \mathcal{C}(a,X)$.
The construction of set $\mathcal{C}(a,X)$ is described below. We distinguish four cases: A, B, C, and D.

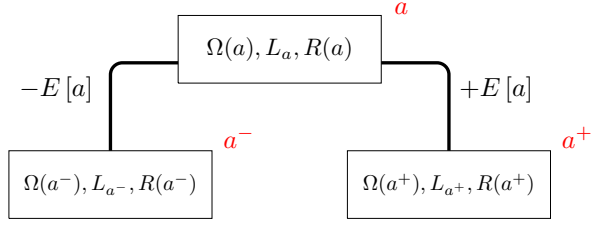
\begin{figure}[h]
\centering
\scalebox{0.9}{\begin{tikzpicture}


\draw [color = black, fill = white] (6.0,10.0) -- (6.0,11.0) -- (9.0,11.0) -- (9.0,10.0) --  (6.0,10.0);
\draw [color = black, fill = white] (3.5,8.0) -- (3.5,9.0) -- (6.5,9.0) -- (6.5,8.0) -- (3.5,8.0);
\draw [color = black, fill = white] (8.5,8.0) -- (8.5,9.0) -- (11.5,9.0) -- (11.5,8.0) -- (8.5,8.0);


\node (P1) at (7.5,10.5) {$\Omega(a),L_a,R(a)$};
\node[scale=0.9] (P2) at (5.0,8.5) {$\Omega(a^-),L_{a^-},R(a^-)$};
\node[scale=0.9] (P3) at (10.0,8.5) {$\Omega(a^+),L_{a^+},R(a^+)$};


\node[scale=1.1, color = red] at (9.3,11.1) {$a$};
\node[scale=1.1, color = red] at (6.9,9.2) {$a^-$};
\node[scale=1.1, color = red] at (11.9,9.2) {$a^+$};

\node[scale=1.1] at (4.2,9.9) {$-E\left[a\right]$};
\node[scale=1.1] at (10.7,9.9) {$+E\left[a\right]$};


\draw[rounded corners=5pt,line width = 1.4pt] (6.0,10.3) -| (5.0,9.0);

\draw[rounded corners=5pt,line width = 1.4pt] (9.0,10.3) -| (10.0,9.0);

\end{tikzpicture}}
\caption{Nodes $a$, $a^+$ and $a^-$ in tree $T$.}
\label{fig:dag}
\end{figure}

\begin{itemize}
\item \textbf{Case A}. No class of $X$ is in $L_a$, {\em i.e.} $L_a \cap X = \emptyset$. Otherwise, see next cases.

As $L_a$ is the maximum-weighted POF in $\Omega(a)$ and $L_a \cap X = \emptyset$, we have $h(a,X) = L_a$. No recursive call is needed: $\mathcal{C}(a,X) = \emptyset$. 

A special case of Case A is when $\Omega(a)$ is a singleton: $\Omega(a) = R^+(a)$. We know from Definition~\ref{def:cp} that $R^+(a) \cap X = \emptyset$.
\item \textbf{Case B}. Class $E\left[a\right]$ belongs to $X$: $E\left[a\right] \in X$. Otherwise, see next cases.

As $R(a^-) = R(a) \cup \set{E\left[a\right]}$, $(a^-,X\backslash (E\left[a\right]))$ is a constraint pair (Definition~\ref{def:cp}): $R(a^-) \cap (X\backslash (E\left[a\right])) =\emptyset$ and $R^+(a^-)=R^+(a)$ is orthogonal to all $\Theta$-classes of set $X\backslash (E\left[a\right]) \subsetneq X$. 

We fix $\mathcal{C}(a,X)$ as a singleton containing $(a^-,X\backslash (E\left[a\right]))$: $\mathcal{C}(a,X) = \set{(a^-,X\backslash (E\left[a\right]))}$.
\item \textbf{Case C}. Set $E\left[a\right] \cup X$ is a POF: $E\left[a\right] \cup X \in \mathcal{L}$.

As $E\left[a\right] \notin X$, we have $X \cap R(a^+) = X \cap R(a^-) = \emptyset$. First, $R^+(a^-) = R^+(a)$, so $X \cup R^+(a^-)$ is a POF and $(a^-,X) \in \mathcal{C}$. Second, $R^+(a^+) = R^+(a) \cup E\left[a\right]$: as $E\left[a\right]$ is orthogonal to all $\Theta$-classes of $X$, $R^+(a^+)\cup X$ is a POF and $(a^+,X) \in \mathcal{C}$. 

These two constraint pairs are the elements of $\mathcal{C}(a,X)$: $\mathcal{C}(a,X) = \set{(a^+,X),(a^-,X)}$.
\item \textbf{Case D}. Set $E\left[a\right] \cup X$ is not a POF: $E\left[a\right] \cup X \notin \mathcal{L}$.

Let $X_{|E[a]}\subsetneq X$ be the $\Theta$-classes of $X$ which are orthogonal to $E\left[a\right]$. Pair $(a^-,X)$ is a constraint pair, using the same arguments than in Case C. We verify whether another pair $(a^+,X_{|E[a]})\in \mathcal{C}$. As $X_{|E[a]} \subsetneq X$ and $E\left[a\right] \notin X$, we have $R(a^+) \cap X_{|E[a]} = \emptyset$. Furthermore, $X_{|E[a]} \cup R^+(a)$ is a POF because $(a,X)$ is a constraint pair and $X_{|E[a]} \cup E\left[a\right]$ is a POF by definition, so $X_{|E[a]} \cup R^+(a^+)$ is a POF.

We fix $\mathcal{C}(a,X) = \set{(a^-,X),(a^+,X_{|E[a]})}$.
\end{itemize}

Observe that when $(a',X') \in \mathcal{C}(a,X)$, then $a$ is a parent of $a'$ in $T$. Moreover, $X' \subseteq X$. The size of all sets of recursive calls $\mathcal{C}(a,X)$ is at most two.
The following theorem justifies that POF $h(a,X)$ can be determined as a function of all $h(a',X')$ satisfying $(a',X') \in \mathcal{C}(a,X)$.

\begin{theorem}
Let $(a,X) \in \mathcal{C}$. If $\mathcal{C}(a,X) = \emptyset$ (Case A), $h(a,X)$ is equal to $L_a$. Otherwise:
\begin{equation}
h(a,X) = \argmax_{h(a',X') \mbox{~s.t.~} (a',X') \in \mathcal{C}(a,X)} \omega\left(h(a',X')\right)
\label{eq:dp}
\end{equation}
\label{th:dp_pairs}
\end{theorem}
\begin{proof}
The justification for Case A was evoked above: as $L_a \cap X = \emptyset$ and $L_a$ is the maximum-weighted POF of $\Omega(a)$, then $h(a,X) = L_a$.

In Case B, $E\left[a\right] \in X$. Recall that $\mathcal{C}(a,X) = \set{(a^-,X\backslash E\left[a\right])}$ in this case. POF $h(a,X)$, which is disjoint from $X$, cannot contain class $E\left[a\right]$. As a consequence, $h(a,X)$ belongs to $\Omega(a^-)$ which is made up of the POFs of $\Omega(a)$ without $E\left[a\right]$. Moreover, $h(a,X)$ does not contain any class of $X\backslash E\left[a\right]$. We have: $h(a,X) = h(a^-,X\backslash E\left[a\right])$.

In Case C, we assume that $E\left[a\right] \cup X$ is a POF. As $E\left[a\right] \notin X$, $h(a,X)$ can be either in $\Omega(a^-)$ or in $\Omega(a^+)$. POFs $h(a^-,X)$ and $h(a^+,X)$ are respectively the maximum-weighted POFs of $\Omega(a^-)$ and $\Omega(a^+)$ without any class of $X$. Equation~\eqref{eq:dp} holds as $\mathcal{C}(a,X) = \set{(a^-,X),(a^+,X)}$.

In Case D, we assume that $E\left[a\right] \cup X$ is not a POF. As in Case C, $E\left[a\right] \notin X$, so $(a^-,X)$ is a constraint pair. If $h(a,X)$ belongs to $\Omega(a^-)$, then it is $h(a^-,X)$. Moreover, $(a^+,X_{|E[a]})\in \mathcal{C}$ and we prove that if $h(a,X)$ belongs to $\Omega(a^+)$, it is $h(a^+,X_{|E[a]})$. Let $E_h$ be a $\Theta$-class of $X\backslash X_{|E[a]}$. Classes $E_h$ and $E\left[a\right]$ are parallel, otherwise $E_h$ would belong to $X_{|E[a]}$. All POFs of set $\Omega(a^+)$ contain $E\left[a\right]$, so none of them can contain $E_h$: it would be contradictory with the non-orthogonality of these $\Theta$-classes. In summary, no POF in $\Omega(a^+)$ contains a $\Theta$-class of $X\backslash X_{|E[a]}$. Therefore, determining the maximum-weighted POF of $\Omega(a^+)$ disjoint from $X$ is equivalent to finding the maximum-weighted POF  of $\Omega(a^+)$ disjoint from $X_{|E[a]}$. In brief, if $h(a,X)$ is in $\Omega(a^+)$, it is $h(a^+,X_{|E[a]})$. Equation~\eqref{eq:dp} holds as $\mathcal{C}(a,X) = \set{(a^-,X),(a^+,X_{|E[a]})}$.
\end{proof}

The DP algorithm consists in recursively applying Equation~\eqref{eq:dp} from all $h(a_0,L)$, for any $L \in \mathcal{L}$ and store the POFs $h(a,X)$ which are computed throughout the execution.
Case A is the base case of the recursion. 
Let $H$ be the directed acyclic graph (DAG) representing the recursive calls of our DP. Its vertex set $V(H)$ contains all $(a,X) \in \mathcal{C}$ such that $h(a,X)$ is called for the computation of certain POFs $h(a_0,L)$. Its edge set is made up of arcs from $(a,X) \in V(H)$ to elements in $\mathcal{C}(a,X)$. Certain constraint pairs of $\mathcal{C}$ may not belong to $V(H)$, {\em i.e.} they are not needed to compute the opposites. Figure~\ref{fig:constraint_pairs} illustrates the DAG $H$ with a vertex $(a,X)$ and its nearby successors. In this example, pair $(a,X)$ is needed for the computation of $\opp(L) = h(a_0,L)$ and $\opp(\widehat{L}) = h(a_0,\widehat{L})$: dotted lines mean there is a path between two pairs.

\begin{figure}[h]
    \centering
    \scalebox{0.9}{\begin{tikzpicture}


\draw [color = red, dashed] (1.0,4.8) -- (1.0,6.2) -- (9
.0,6.2) -- (9.0,4.8) -- (1.0,4.8);
\draw [color = blue, dashed] (1.0,3.3) -- (1.0,4.7) -- (4
.0,4.7) -- (4.0,3.3) -- (1.0,3.3);


\draw [color = black, fill = white] (6.0,9.5) -- (6.0,10.5) -- (8.0,10.5) -- (8.0,9.5) --  (6.0,9.5);
\draw [color = black, fill = white] (2.0,9.5) -- (2.0,10.5) -- (4.0,10.5) -- (4.0,9.5) --  (2.0,9.5);
\draw [color = black, fill = white] (4.0,6.5) -- (4.0,7.5) -- (6.0,7.5) -- (6.0,6.5) -- (4.0,6.5);
\draw [color = black, fill = white] (1.5,5.0) -- (1.5,6.0) -- (3.5,6.0) -- (3.5,5.0) -- (1.5,5.0);
\draw [color = black, fill = white] (6.5,5.0) -- (6.5,6.0) -- (8.5,6.0) -- (8.5,5.0) -- (6.5,5.0);
\draw [color = black, fill = white] (1.5,3.5) -- (1.5,4.5) -- (3.5,4.5) -- (3.5,3.5) -- (1.5,3.5);


\node (P1) at (3.0,10.0) {$(a_0,L)$};
\node (P1') at (7.0,10.0) {$(a_0,\widehat{L})$};
\node (P2) at (5.0,7.0) {$(a,X)$};
\node (P3) at (2.5,5.5) {$(a',X')$};
\node (P4) at (7.5,5.5) {$(a'',X'')$};
\node (P5) at (2.5,4.0) {$(\Tilde{a},\Tilde{X})$};


\node[scale=1.1, color = red] at (10.0,5.5) {$\mathcal{C}(a,X)$};
\node[scale=1.1, color = blue] at (5.0,4.0) {$\mathcal{C}(a',X')$};


\draw[->,>=latex,line width = 1.4pt] (4.5,6.5)--(3.0,6.0);
\draw[->,>=latex,line width = 1.4pt] (5.5,6.5)--(7.0,6.0);
\draw[->,>=latex,line width = 1.4pt] (2.5,5.0)--(2.5,4.5);

\draw[->,>=latex,dotted,line width = 1.4pt] (7.0,9.5) -- (7.0,8.0) -- (5.5,7.5);
\draw[->,>=latex,dotted,line width = 1.4pt] (3.5,9.5) -- (4.0,9.0) -- (4.0,8.0) -- (4.5,7.5);

\end{tikzpicture}}
    \caption{Some vertices of the DAG $H$}
    \label{fig:constraint_pairs}
\end{figure}
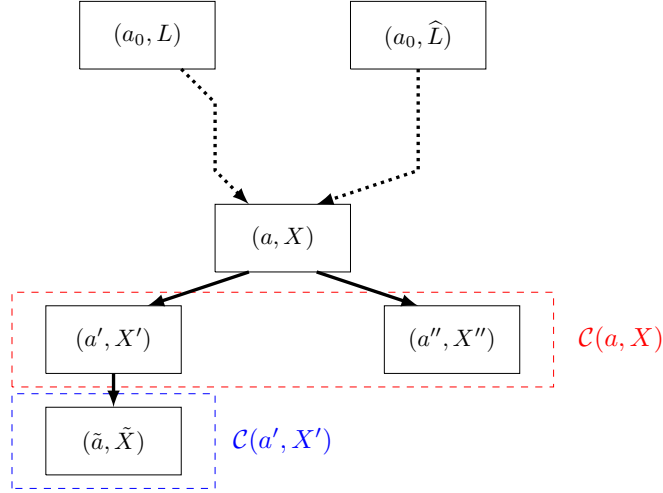

We observe that the constraint pairs of $V(H)$ are made up of nodes of $T^{(d)}$.

\begin{lemma}
Constraint pairs $(a,X)$ in $V(H)$ satisfy the following inequality: $\card{X}\le d-\textsf{layer}(a)$.
\label{le:layer_inequality}
\end{lemma}
\begin{proof}
We proceed by induction. The roots of the DAG, {\em i.e.} the constraint pairs $(a_0,X)$, verify this inequality, as $\textsf{layer}(a_0) = 0$. 

Assume now that the ancestors of $(a,X) \in V(H)$ satisfy the inequality. We distinguish two cases. First, suppose there is a predecessor of $(a,X)$, {\em i.e.} a constraint pair $(a',X')$ and an arc going from it to $(a,X)$, such that $\textsf{layer}(a) = \textsf{layer}(a')+1$. We select a pair $(\widehat{a}',\widehat{X}')$ in $V(H)$ which is an ancestor of $(a',X')$ - there is a directed path from  $(\widehat{a}',\widehat{X}')$ to $(a',X')$ in $H$ - and such that $\widehat{a}'$ is the ancestor of $a'$ in $T$ which is at the top of the block $\mathcal{B}_{\widehat{a}'}$ of layer $\textsf{layer}(a')$ containing $a'$. Thus, node $a$ is a leaf of block $\mathcal{B}_{\widehat{a}'}$. The existence of a directed path in $H$ from $(\widehat{a}',\widehat{X}')$ to $(a,X)$ means that the recursive computation of $h(\widehat{a}',\widehat{X}')$ uses the term $h(a,X)$. The POF $L_{\widehat{a}'}$ is not disjoint from $\widehat{X}'$ otherwise pair $(\widehat{a}',\widehat{X}')$ would be a leaf in $H$ (Case A). Let $E_h \in L_{\widehat{a}'} \cap \widehat{X}'$. As $a$ is a leaf of the block $\mathcal{B}_{\widehat{a}'}$, it verifies $E_h \in R(a)$ by construction of $T$. As $(\widehat{a}',\widehat{X}')$ is an ancestor of $(a,X)$, then $X \subseteq \widehat{X}'$. But $(a,X)$ is a constraint pair and $E_h \in R(a)$, so $X \subseteq \widehat{X}'\backslash \set{E_h}$. In brief, $\card{X} \le \card{\widehat{X}'} - 1$. Using the induction hypothesis, $\card{X} \le d-\textsf{layer}(\widehat{a}')-1 = d-\textsf{layer}(a)$.

Second, suppose that all predecessors $(a',X')$ of $(a,X)$ satisfy $\textsf{layer}(a') = \textsf{layer}(a)$. As $X \subseteq X'$, $\card{X}\le \card{X'} \le d-\textsf{layer}(a') = d-\textsf{layer}(a)$. 
\end{proof}

As a consequence, for any pair $(a,X) \in V(H)$, the depth of node $a$ in $T$ can be upper-bounded by $d^2$ because each block has at most depth $d$ and the layer of $a$ is at most $d$. Formally, $V(H) \subseteq \mathcal{C}^{(d)}$. This shows that computing tree $T^{(d)}$ - and getting rid of the larger depths of the tree - is sufficient for the execution of the DP. This observation allows us to state in the next subsubsection that this algorithm runs in quasilinear time.

\subsubsection{Analysis} \label{subsubsec:analysis_wopp}

In this section, we prove that the DP algorithm described earlier can run in quasilinear time $O(d^3n)$. The global algorithm is divided into two parts.
\begin{enumerate}
\item Construction of the tree $T^{(d)}$ with maximum layer $d$, using partition refinement, storage of $\Omega(a)$, $L_a$, and $R(a)$ for each node $a$ of $T^{(d)}$.
\item Computation of all opposites $h(a_0,L)$ with DP.
\end{enumerate}
The runtime needed to build $T^{(d)}$ is $O(d^2n)$. The execution time of all partitionings is $O((d+\log n)n)$, according to Corollary~\ref{co:oipr}.
The memory space used to store sets $\Omega(a),L_a,R(a)$ for each node $a\in V(T^{(d)})$ can be upper-bounded by $O(d^2n)$ because the sets $\Omega(a)$ together at some fixed depth form a partition of $\mathcal{L}$ and the depth of $T^{(d)}$ is at most $d^2$.

The analysis for the second part of the algorithm is not trivial. Our key argument consists in providing an upper bound for the number of constraint pairs in $\mathcal{C}^{(d)}$. It provides us with a maximum number of POFs which have to be stored during the DP.

\begin{theorem}
There are at most $d^2n$ constraint pairs in $\mathcal{C}^{(d)}$.
\label{th:cardinality_c}
\end{theorem}
\begin{proof}
We define a function $f : \mathcal{L} \rightarrow 2^{\mathcal{C}}$ such that, for any $L \in \mathcal{L}$, $\card{f(L)} \le d^2$. We show that for any constraint pair $(a,X) \in \mathcal{C}^{(d)}$, there is a POF $Y$, $X \subseteq Y$, such that $(a,X) \in f(Y)$. As the total number of constraint pairs generated by $f$ is at most $d^2n$, we have $\card{\mathcal{C}^{(d)}} \le d^2n$.

- \textbf{Step 1: definition of function $f$}. Our method consists in defining an \textit{itinerary function} $I : \mathcal{C}^{(d)} \rightarrow \mathcal{C}^{(d)}$, where $(a,X) = I(a',X')$ implies that $X \subseteq X'$ and $a$ is a child of $a'$ in $T^{(d)}$. Starting from some $(a_0,Y)$, where $Y \in \mathcal{L}$, successive appliances of function $I$ provide us with a descent in $T^{(d)}$, {\em i.e.} a simple path from $a_0$ to a leaf of $T^{(d)}$, if we refer to the nodes of the constraint pairs obtained. Function $f$ will be defined as follows: $f(Y) = \set{(a_0,Y),I(a_0,Y),I(I(a_0,Y)),\ldots}$. As the depth of $T^{(d)}$ is at most $d^2$, we confirm that $\card{f(Y)}\le d^2$.

Here is the definition of function $I$ and we verify the properties announced above. Let $(a',X') \in \mathcal{C}^{(d)}$: we denote by $a^-$ (resp. $a^+$) his child which is the endpoint of edge indexed by $-E\left[a'\right]$ (resp. $+E\left[a'\right]$). We distinguish three cases: as they are similar to some Cases enumerated the previous subsubsection, we denote them by B*, C* and D* respectively.
\begin{itemize}
\item \textbf{Case B*}: $\Theta$-class $E\left[a'\right]$ belongs to $X$, {\em i.e.} $E\left[a'\right] \in X$. We fix $I(a',X') = (a^-,X\backslash \set{E\left[a'\right]})$.
\item \textbf{Case C*}: $E\left[a'\right] \notin X$ and $E\left[a'\right] \cup X$ is a POF. We fix $I(a',X') = (a^+,X)$.
\item \textbf{Case D*}: $E\left[a'\right] \notin X$ and $E\left[a'\right] \cup X$ is not a POF. We fix $I(a',X') = (a^-,X)$.
\end{itemize}
Any constraint pair is concerned by one of these three cases. One can see that if $I(a',X') = (a,X)$ then $a$ is a child of $a'$ and $X \subseteq X'$.

- \textbf{Step 2: any constraint pair is generated by $f$}. We show that any $(a,X) \in \mathcal{C}^{(d)}$ can be written as $I^k(a_0,Y)$ for some POF $Y \supseteq X$ and $0\le k < d^2$. We proceed by induction on the depth of $a$.
The base case is trivial: If $\textsf{depth}(a) = 0$, then $a=a_0$. For any $X \in \mathcal{L}$, $(a_0,X) \in f(X)$.

Let $(a,X) \in \mathcal{C}^{(d)}$, $a \neq a_0$, we assume that each constraint pair containing an ancestor of $a$ belongs to some $f(Y)$. We distinguish three scenarii depending on the nature of both the parent $a'$ of $a$ and $\Theta$-class $E\left[a'\right]$. For each scenario, we show that there is a constraint pair $(a',X')$ such that $I(a',X') = (a,X)$ and $X\subseteq X'$. As $(a',X') = I^k(a_0,Y)$ and $X' \subseteq Y$ by induction, such statement implies that $(a,X) = I^{k+1}(a_0,Y)$ and $X \subseteq Y$, in other words $(a,X) \in f(Y)$.
\begin{itemize}
\item \textbf{Scenario 1}: edge $(a',a)$ in $T$ is indexed by $-E\left[a'\right]$ and $X \cup E\left[a'\right]$ is a POF. 

Let $X' = X \cup \set{E\left[a'\right]}$. We begin with the proof that $(a',X')$ is a constraint pair. First, set $R(a)$ is the union of singleton $\set{E\left[a'\right]}$ with $R(a')$, so $E\left[a'\right] \notin R(a')$. As $(a,X) \in \mathcal{C}$, $X \cap R(a) = \emptyset$, so $X \cup \set{E\left[a'\right]}$ has no intersection with $R(a')$. In brief, $X' \cap R(a')=\emptyset$. Second,  as $(a,X) \in \mathcal{C}$, we affirm that any $\Theta$-class of $X$ is orthogonal to any class of $R^+(a) = R^+(a')$. Showing that $E\left[a'\right]$ is orthogonal to any class of $R^+(a')$ will imply that $X' \cup R^+(a')$ is a POF. If $a'$ is the root of some block, then $E\left[a'\right]$ is a $\Theta$-class of $L_{a'}$. As $L_{a'}$ contains all $\Theta$-classes of $R^+(a')$ by definition, this shows the orthogonality of $E\left[a'\right]$ with $R^+(a')$. Now, assume that $a'$ is not the root of some block: $\mathcal{B}_{a'}$ denotes the block it belongs to and the root of $\mathcal{B}_{a'}$ is written $a^*\neq a'$ (Figure~\ref{fig:scenar1}). Node $a^*$ is an ancestor of $a'$. Set $R^+(a')$ can be splitted as follows: on one hand, the $\Theta$-classes of $R^+(a')$ indexed between $a_0$ and $a^*$, {\em i.e.} $R^+(a^*)$, and on the other hand the $\Theta$-classes of $R^+(a')$ whose indices appear in the block $\mathcal{B}_{a'}$, {\em i.e.} $R^+(a') \cap L_{a^*}$. Class $E\left[a'\right]$ belongs to $L_{a^*}$, otherwise $a'$ would be the leaf of this block and, therefore, the root of another block. Both $R^+(a^*)$ and $R^+(a') \cap L_{a^*}$ are subsets of $L_{a^*}$, so $R^+(a') \subsetneq L_{a^*}$. As a conclusion, $R^+(a') \cup \set{E\left[a'\right]} \subseteq L_{a^*}$ is a POF, therefore $X' \cup R^+(a')$ is a POF and $(a',X')$ is thus a constraint pair.

We show that $I(a',X') = (a,X)$. We refer to Case B*: the $\Theta$-class $E\left[a'\right]$ belongs to $X'$, so function $I$ returns $(a^-,X'\backslash \set{E\left[a'\right]}) = (a,X)$.

\begin{figure}[t]
\centering
\scalebox{0.9}{\begin{tikzpicture}

\draw [color = black] (6.2,7.5) -- (1.0,4.5) -- (11.4,4.5) -- (6.2,7.5);


\node[draw, circle, fill=black, scale=0.6] (a0) at (7.2,9.0) {};

\node[draw, circle, fill=black, scale=0.6] (as) at (6.2,7.0) {};
\node[draw, circle, fill=black, scale=0.6] (a1) at (4.9,6.2) {};
\node[draw, circle, fill=black, scale=0.6] (a2) at (5.0,5.4) {};

\draw[dashed] (a0) -- (as);
\draw[dashed] (as) -- (a1);
\draw[dashed] (as) -- (7.5,6.2);
\draw[line width =1.2pt] (a1) -- (a2);
\draw[dashed] (a1) -- (4.2,5.4);
\draw[dashed] (a2) -- (4.6,4.7);
\draw[dashed] (a2) -- (5.4,4.7);
\draw[dashed] (4.6,4.7) -- (4.2,4.0);
\draw[dashed] (4.6,4.7) -- (5.0,4.0);

\draw[->, color = blue, line width = 1.5pt] (6.9,8.8) -- (6.2,7.4);
\draw[->, color = blue, line width = 1.5pt] (5.7,7.0) -- (4.9,6.5);


\node at (8.2,6.8) {$\mathcal{B}_{a'}$};
\node at (7.5,8.7) {$a_0$};

\node[color = blue] at (6.2,8.1) {$I$};
\node[color = blue] at (6.6,9.3) {$(a_0,Y)$};
\node[color = blue] at (6.6,6.0) {$(a,X) = I^k(a_0,Y)$};
\node[color = blue] at (7.0,5.2) {$(a',X') = I^{k+1}(a_0,Y)$};

\node at (6.2,6.7) {$a^*$};
\node at (4.5,6.2) {$a$};
\node at (4.6,5.4) {$a'$};

\end{tikzpicture}}
\caption{Nodes $a$, $a'$, and $a^*$ in $T$ when $a'$ is not the root of a block in Scenario 1. Blue arrows indicate the successive appliances of the itinerary function on $(a_0,Y)$.}
\label{fig:scenar1}
\end{figure}
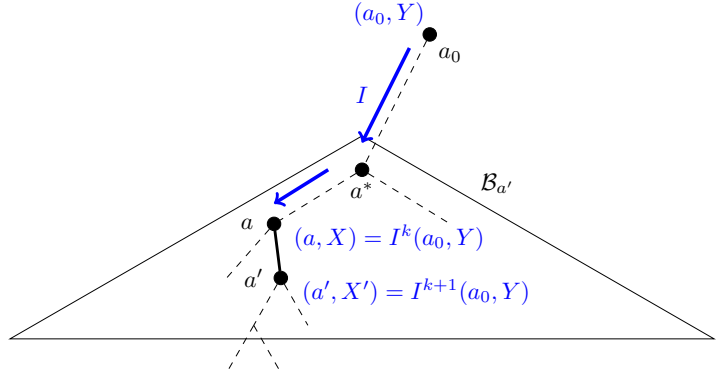

\item \textbf{Scenario 2}: edge $(a',a)$ in $T$ is indexed by $-E\left[a'\right]$ and $X \cup E\left[a'\right]$ is not a POF.

Let $X' = X$. As $(a,X) \in \mathcal{C}$ and $R(a') \subsetneq R(a)$, we have $X \cap R(a') = \emptyset$. Moreover, $R^+(a') = R^+(a)$, so $X \cap R^+(a')$ is a POF. In brief, $(a',X') = (a',X) \in \mathcal{C}$.

We prove that $I(a',X') = (a,X)$. In this scenario, $X \cup \set{E\left[a'\right]}$ is not a POF. Furthermore, $X \cap R(a) = \emptyset$ because $(a,X) \in \mathcal{C}$, so $E\left[a'\right] \notin X$. We refer to Case D* and $I(a',X') = (a^-,X') = (a,X)$.

\item \textbf{Scenario 3}: edge $(a',a)$ in $T$ is indexed by $+E\left[a'\right]$.

Let $X' = X$. As $(a,X) \in \mathcal{C}$ and $R(a') \subsetneq R(a)$, we have $X \cap R(a') = \emptyset$ as in Scenario 2. Moreover, $R^+(a') \subsetneq R^+(a)$: only $E\left[a'\right]$ is in $R^+(a)$ and not in $R^+(a')$. As $X \cup R^+(a)$ is a POF, its subset $X \cup R^+(a')$ is too. In brief, $(a',X') = (a',X) \in \mathcal{C}$.

We show that $I(a',X') = (a,X)$. Class $E\left[a'\right]$ is not in $X'=X$ because $X \cap R(a) = \emptyset$. Set $X \cup \set{E\left[a'\right]}$ is a subset of the POF $X \cup R^+(a)$ because $E\left[a'\right] \in R^+(a)$, so $X \cup \set{E\left[a'\right]}$ is a POF. We refer to Case C*: $I(a',X') = (a^+,X') = (a,X)$.
\end{itemize}

In summary, there is a constraint pair $(a',X')$ such that $a'$ is the parent of $a$, $X \subseteq X'$, and $I(a',X') = (a,X)$. The induction hypothesis terminates the proof of our claim: as $(a',X')$ is in  some itinerary $f(Y)$ with $X' \subseteq Y$, and $I(a',X') =(a,X)$, then $(a,X)$ also belongs to $f(Y)$. As the size of each $f(Y)$, $Y \in \mathcal{L}$ is upper-bounded by $d^2$, the total number of constraint pairs is at most $d^2n$.
\end{proof}
The size of the state space of the DP procedure is at most $O(d^3n)$ because $\card{h(a,X)} \le d$ for any $(a,X) \in V(H)$. For any pair $(a,X) \in V(H)$, at most two recursive calls are launched to compute $h(a,X)$. In brief, Theorem~\ref{th:cardinality_c} allows us to affirm that the execution time of the DP procedure is $O(d^3n)$ and is thus the total running time of our algorithm.

\section{Subquadratic-time algorithm for all eccentricities on median graphs} \label{sec:subquadratic}

This section is dedicated to the design of algorithms computing all eccentricities for the whole class of  median graphs (not only simplex graphs). We begin in Section~\ref{subsec:constant_dim} with the proposal of a linear-time FPT algorithm, parameterized by the dimension $d$, running in $2^{O(d)}n$. It is mainly based on a paper of the literature~\cite{BeHa21} which provides a slightly super-exponential time algorithm - running in $2^{O(d\log d)}n$ - for the same problem. Replacing one step of this procedure by the partitioning conceived in Section~\ref{subsec:partitioning} allows us to decrease the exponential dependence on $d$. Thanks to this outcome, in Section~\ref{subsec:reduction}, we are able to design a first subquadratic-time algorithm for all median graphs running in $\tilde{O}(n^{\frac{5}{3}})$.

\subsection{Linear FPT algorithm for constant-dimension median graphs} \label{subsec:constant_dim}

We remind in this subsection the different steps needed to obtain a linear-time algorithm computing all eccentricities of a median graph with constant dimension, $d=O(1)$. We show how the algorithm of Section~\ref{subsec:partitioning} can be integrated to it in order to improve the dependence on $d$. Let us begin with a reminder of the former result.

\begin{lemma}[\cite{BeHa21}]
There is a combinatorial algorithm computing all eccentricities in a median graph $G$ with running time $\tilde{O}(2^{d(\log d + 1)}n)$.
\label{le:slightly_super_exp}
\end{lemma}

Some parts of this subsection are redundant with~\cite{BeHa21}, however we keep this subsection self-contained. The new outcomes presented are Theorems~\ref{th:compute_opp} and~\ref{th:simple_ecc}. The results that are reminded will also be useful for Section~\ref{sec:discussion}.

The algorithm evoked by Lemma~\ref{le:slightly_super_exp} consists in the computation of three kinds of labels: \textit{ladder} labels $\varphi$, \textit{opposite} labels $\opp$ and \textit{anti-ladder} labels $\psi$. The order in which they are given correspond to their respective dependence: $\opp$-labelings are functions of labels $\varphi$ and $\psi$-labelings are functions of both labels $\varphi$ and $\opp$.
The definition of $\opp$-labelings on general median graphs is very close to the notion of opposite previously defined for simplex graphs in Section~\ref{sec:simplex}.

\subsubsection{Definition of the POF-parallelism} \label{subsubsec:parallel}

We present a notion related to orthogonality which will be used in the remainder.

\begin{definition}[$L$-parallelism]
Given a POF $L$, we say that a POF $L^+$ is $L$-\emph{parallel} if, for any $E_j \in L^+$, $L \cup \set{E_j}$ is not a POF.
\label{def:pof_parallel}
\end{definition}

Presented differently, a $L$-parallel POF is such that any of its $\Theta$-classes is parallel to at least one $\Theta$-class of $L$.
When $L^+$ is a $L$-parallel POF, we have $L \cap L^+ = \emptyset$, otherwise $L \cup \set{E_j} = L$ for some $E_j \in L^+$. 
We extend Definition~\ref{def:pof_parallel} by adding the notion of adjacency. 

\begin{definition}
Given a pair $L,L^+$ of POFs, we say that $L^+$ is $L$-\textit{parallel-adjacent} if: 
\begin{itemize}
\item $L^+$ is $L$-parallel,
\item there exists a vertex $u$ such that $L$ is ingoing into $u$ and $L^+$ is outgoing from $u$.
\end{itemize}
\label{def:pof_parallel_adj}
\end{definition}

Given a POF $L$, we denote by $\paradj(L)$ the set of POFs $L^+$ which are $L$-\textit{parallel-adjacent}. As an example, in Figure~\ref{fig:compute_labels}, $\set{E_4,E_5} \in \paradj(\set{E_2,E_3})$ but $\set{E_1} \notin \paradj(\set{E_2,E_3})$.
Similarly, for a pair signature/anti-basis $(L,u)$, defining an hypercube:
\begin{definition}
We say POF $L^+$ is $(L,u)$-parallel if $L^+$ is both $L$-parallel and outgoing from $u$.
\label{def:hyp_parallel_adj} 
\end{definition}
We denote by $\paradj(L,u)$ the set of $(L,u)$-parallel POFs. Observe now that, in this hypercube viewpoint, the adjacency between the two POFs $L$ and $L^+$ is fixed by the existence of the anti-basis $u$ of $(L,u)$.

Our future algorithms will be strongly based on the enumeration of sets $\paradj(L)$, for all POFs $L$, but also of sets $\paradj(L,u)$ for all pairs signature/anti-basis. As a first version, our algorithm uses only the enumeration of all sets $\paradj(L,u)$, executed in a naive way. Indeed, we know from Lemma~\ref{le:pof_hypercube} that the number of $\Theta$-classes ingoing into $u$ is $d$ and that they are all pairwise orthogonal. Furthermore, all pairs $(u,L^+)$ such that $L^+$ is outgoing from $u$ can be enumerated in $\tilde{O}(2^dn)$. All in all, one can enumerate all sets $\paradj(L,u)$ in time $\tilde{O}(2^{2d}n)$.

\begin{theorem}
There is an algorithm, running in time $\tilde{O}(2^{2d}n)$, which returns all sets $\paradj(L,u)$.
\label{th:naive_paradj}
\end{theorem}
\begin{proof}
Enumerate all hypercubes of $G$ in $\tilde{O}(2^dn)$ (Lemma~\ref{le:enum_hypercubes}). We obtain all possible pairs $(u,L^+)$, where POF $L^+$ is outgoing from $u$. Then, among the $d$ $\Theta$-classes ingoing into $u$, try all collections of sets (at most $2^d$). These collections are POFs, since the $\Theta$-classes ingoing into $u$ are pairwise orthogonal (Lemma~\ref{le:pof_hypercube}). For each of these POFs $L$, verify whether $L^+$ is $L$-parallel. As $\card{L},\card{L^+} \le d$, this is done in logarithmic time $O(d^2)$. The total running time is thus $\tilde{O}(2^{2d}n)$.
\end{proof}

In Section~\ref{sec:discussion}, we propose a better algorithm to enumerate these ``parallel-adjacency'' lists. 
For the moment, in Section~\ref{sec:subquadratic}, we will use the naive version proposed in Theorem~\ref{th:naive_paradj}.

\subsubsection{Ladder labels} \label{subsubsec:ladder}

Some preliminary work has to be done before giving the definition of labels $\varphi$. We introduce the notion of \textit{ladder set}. It is defined only for pairs of vertices $u,v$ satisfying the condition $u \in I(v_0,v)$. In this situation, the edges of shortest $(u,v)$-paths are all oriented towards $v$ with the $v_0$-orientation.

\begin{definition}[Ladder set $L_{u,v}$]
Let $u,v \in V$ such that $u \in I(v_0,v)$. The ladder set $L_{u,v}$ is the subset of $\sigma_{u,v}$ which contains the $\Theta$-classes admitting an edge adjacent to $u$.
\label{def:ladder}
\end{definition}

Figure~\ref{fig:compute_labels} shows a small median graph with four vertices $v_0,u,v,x$ such that $u \in I(v_0,v)$ and $u \in I(v_0,x)$. It gives the composition of ladder sets $L_{u,v}$ and $L_{u,x}$.

\begin{figure}[h]
\centering
\scalebox{0.8}{\begin{tikzpicture}


\node[draw, circle, minimum height=0.2cm, minimum width=0.2cm, fill=black] (P01) at (-1,1) {};

\node[draw, circle, minimum height=0.2cm, minimum width=0.2cm, fill=black] (P11) at (1,1) {};
\node[draw, circle, minimum height=0.2cm, minimum width=0.2cm, fill=black] (P12) at (1,2.5) {};
\node[draw, circle, minimum height=0.2cm, minimum width=0.2cm, fill=black] (P13) at (3,1) {};
\node[draw, circle, minimum height=0.2cm, minimum width=0.2cm, fill=black] (P14) at (3,2.5) {};

\node[draw, circle, minimum height=0.2cm, minimum width=0.2cm, fill=black] (P21) at (2.0,1.4) {};
\node[draw, circle, minimum height=0.2cm, minimum width=0.2cm, fill=black] (P22) at (2.0,2.9) {};
\node[draw, circle, minimum height=0.2cm, minimum width=0.2cm, fill=black] (P23) at (4.0,1.4) {};
\node[draw, circle, minimum height=0.2cm, minimum width=0.2cm, fill=black] (P24) at (4.0,2.9) {};

\node[draw, circle, minimum height=0.2cm, minimum width=0.2cm, fill=black] (P31) at (3.0,4) {};
\node[draw, circle, minimum height=0.2cm, minimum width=0.2cm, fill=black] (P32) at (5.0,4) {};

\node[draw, circle, minimum height=0.2cm, minimum width=0.2cm, fill=black] (P41) at (5.0,1) {};
\node[draw, circle, minimum height=0.2cm, minimum width=0.2cm, fill=black] (P42) at (5.0,2.5) {};



\draw[line width = 1.4pt] (P01) -- (P11);

\draw[line width = 1.4pt, color = green] (P11) -- (P12);
\draw[line width = 1.4pt, color = blue] (P11) -- (P13);
\draw[line width = 1.4pt, color = red] (P11) -- (P21);
\draw[line width = 1.4pt, color = blue] (P12) -- (P14);
\draw[line width = 1.4pt, color = red] (P12) -- (P22);
\draw[line width = 1.4pt, color = red]  (P13) -- (P23);
\draw[line width = 1.4pt, color = green] (P13) -- (P14);
\draw[line width = 1.4pt, color = red]  (P14) -- (P24);
\draw[line width = 1.4pt, color = green] (P21) -- (P22);
\draw[line width = 1.4pt, color = blue] (P21) -- (P23);
\draw[line width = 1.4pt, color = green] (P23) -- (P24);
\draw[line width = 1.4pt, color = blue] (P22) -- (P24);

\draw[line width = 1.4pt, color = purple] (P14) -- (P31);
\draw[line width = 1.4pt, color = purple] (P42) -- (P32);
\draw[line width = 1.4pt, color = orange] (P13) -- (P41);
\draw[line width = 1.4pt, color = orange] (P14) -- (P42);
\draw[line width = 1.4pt, color = orange] (P31) -- (P32);
\draw[line width = 1.4pt, color = green] (P41) -- (P42);



\node[scale=1.2, color = red] at (1.4,3.0) {$E_1$};
\node[scale=1.2, color = blue] at (2.0,0.5) {$E_2$};
\node[scale=1.2, color = green] at (0.5,1.75) {$E_3$};
\node[scale=1.2, color = orange] at (4.0,0.5) {$E_4$};
\node[scale=1.2, color = purple] at (5.5,3.25) {$E_5$};

\node[scale = 1.2] at (-1.4,0.7) {$v_0$};
\node[scale = 1.2] at (0.6,0.7) {$u$};
\node[scale = 1.2] at (5.3,4.3) {$v$};
\node[scale = 1.2] at (4.3,3.2) {$x$};

\end{tikzpicture}}
\caption{Examples of ladder sets: $L_{u,v} = \set{E_2,E_3}$, $L_{u,x} = \set{E_1,E_2,E_3}$.}
\label{fig:compute_labels}
\end{figure}
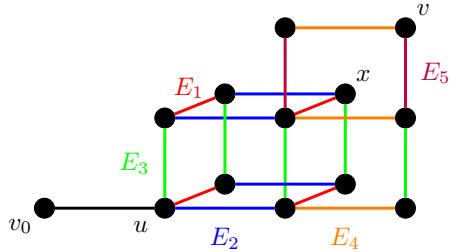

A key characterization on ladder sets states that their $\Theta$-classes are pairwise orthogonal. In brief, every set $L_{u,v}$ is a POF. Let us remind that the adjacency of all $\Theta$-classes of a POF $L$ with the same vertex $u$ implies the existence of a (unique) hypercube not only signed with this POF $L$ but also containing $u$ (Lemma~\ref{le:pof_adjacent}). If additionnally POF $L$ is \textit{outgoing from} $u$ - said differently, the edges adjacent to $u$ belonging to a $\Theta$-class of $L$ leave $u$ -, then $u$ is the basis of the hypercube.  As the $\Theta$-classes of $L_{u,v}$ are adjacent to $u$ by definition, there is a natural bijection between (i) hypercubes (ii) pairs made up of a vertex $u$ and a POF $L$ outgoing from $u$ and (iii) pairs vertex-ladder set $(u,L_{u,\cdot})$.

\begin{lemma}[\cite{BeHa21}]
Every ladder set $L_{u,v}$ is a POF. For any ordering $\tau$ of the $\Theta$-classes in $L_{u,v}$, there is a shortest $(u,v)$-path such that, for any $1\le i \le \card{L_{u,v}}$, the $\ith{i}$ first edge of the path belongs to the $\ith{i}$ $\Theta$-class of $L_{u,v}$ in ordering $\tau$.
\label{le:ladder_POF}
\end{lemma}

The necessary background to introduce labels $\varphi$ is now known.

\begin{definition}[Labels $\varphi$~\cite{BeHa21}]
Given a vertex $u$ and a POF $L$ outgoing from $u$, let $\varphi(u,L)$ be the maximum distance $d(u,v)$ such that $u \in I(v_0,v)$ and $L_{u,v} = L$.
\label{def:varphi}
\end{definition}

Intuitively, integer $\varphi(u,L)$ provides us with the maximum length of a shortest path starting from $u$ into ``direction'' $L$. Observe that the total size of labels $\varphi$ on a median graph $G$ does not exceed $O(2^dn)$, according to Lemma~\ref{le:number_hypercubes}. 

A combinatorial algorithm running in $\tilde{O}(2^{2d}n)$ which computes all labels $\varphi(u,L)$ was identified in~\cite{BeHa21}: we provide an overview of it. There is a crucial relationship between a label $\varphi(u,L)$ and the labels of (i) the anti-basis $u^+$ of the hypercube with basis $u$ and signature $L$ and (ii) the $L$-parallel POFs outgoing from $u^+$.

\begin{lemma}[Inductive formula for labels $\varphi$~\cite{BeHa21}]
Let $u \in V$, $L$ be a POF outgoing from $u$ and $Q$ be the hypercube with basis $u$ and signature $L$. We denote by $u^+$ the opposite vertex of $u$ in $Q$: $u$ is the basis of $Q$ and $u^+$ its anti-basis. A vertex $v\neq u^+$ verifies $u \in I(v_0,v)$ and $L_{u,v} = L$ if and only if (i) $u^+ \in I(v_0,v)$ and (ii) ladder set $L_{u^+,v}$ is $(L,u^+)$-parallel.
\label{le:pushing_phi}
\end{lemma}

A consequence of the previous lemma is that we can distinguish two cases for the computation of $\varphi(u,L)$. In the first case, $\varphi(u,L) = \card{L}$: it occurs when the farthest-to-$u$ vertex with ladder set $L$ is $u^+$ (base case). Indeed, $u^+$ is a candidate as $\sigma_{u,u^+} = L$: shortest $(u,u^+)$-paths pass through hypercube $Q$. This situation happens when either no $\Theta$-class is outgoing from $u^+$ or when all $\Theta$-classes outgoing from $u^+$ are orthogonal to $L$. In the second case, there are vertices farther to $u$ than $u^+$ with ladder set $L$. As announced in Lemma~\ref{le:pushing_phi}, $\varphi(u,L)$ is a function of labels $\varphi(u^+,\cdot)$.

\begin{equation}
\begin{split}
    \varphi(u,L) & = \max\limits_{\substack{L^+ ~\mbox{\scriptsize{POF outgoing from}}~ u^+ \\ \forall E_j \in L^+,~ L \cup \set{E_j} ~\mbox{\scriptsize{not POF}}}} (\card{L} + \varphi(u^+,L^+)).\\
    & = \max\limits_{\substack{(L,u^+)\mbox{\scriptsize{-parallel}} \\ \mbox{\scriptsize{POFs}}~ L^+}} (\card{L} + \varphi(u^+,L^+)).
\end{split}
\label{eq:induction_phi}
\end{equation}

If there exists such a POF $L^+$, then the label $\varphi(u,L)$ is given by Equation~\eqref{eq:induction_phi}. Otherwise, it is given by the first case. Observe that the POFs $L^+$ that are taken into account in Equation~\eqref{eq:induction_phi} are exactly the POFs such that $L^+$ is $(L,u^+)$-parallel. Briefly, the algorithm consists in listing all triplets $(L,u^+,L^+)$ such that $L^+$ is $(L,u^+)$-parallel. (Definition~\ref{def:hyp_parallel_adj}). Vertex $u$ is deduced, as it is the basis of the hypercube with signature $L$ and anti-basis $u^+$. We update $\varphi(u,L)$ if $\card{L} + \varphi(u^+,L^+)$ is greater than the current value. The total number of triplets $(L,u^+,L^+)$ is upper-bounded by $2^{2d}n$ and can be enumerated in $\tilde{O}(2^{2d}n)$ (Theorem~\ref{th:naive_paradj}). For this reason, the computation of $\varphi$-labelings takes $\tilde{O}(2^{2d}n)$.

\begin{theorem}[Computation of labels $\varphi$~\cite{BeHa21}]
There is a combinatorial algorithm which determines all labels $\varphi(u,L)$ in $\tilde{O}(2^{2d}n)$. It also stores, for each pair $(u,L)$, a vertex $v$ satisfying $L_{u,v} = L$ and $d(u,v) = \varphi(u,L)$, denoted by $\mu(u,L)$.
\label{th:compute_phi}
\end{theorem}

\subsubsection{Opposite labels}

The second type of labels needed to compute all eccentricities of a median graph $G$ are opposite labels. Their definition is very close to the function $\opp$ defined in Section~\ref{sec:simplex} for simplex graphs. Given a vertex $u$ and a POF $L$ outgoing from $u$, let $\opp_u(L)$ denote a POF with maximum label $\varphi$ which is disjoint from $L$. As for $\varphi$, the total size of $\opp$-labelings is $O(2^dn)$.

\begin{definition}[Labels $\opp$~\cite{BeHa21}]
Let $u \in V$ and $L$ be a POF outgoing from $u$. Let $\opp_u(L)$ be one of the POF $L'$ outgoing from $u$, disjoint from $L$, which maximizes value $\varphi(u,L')$.
\end{definition}

On simplex graphs, the opposite function provides in fact the $\opp$-labelings of vertex $v_0$: $\opp(X) = \opp_{v_0}(X)$. As all vertices belong to hypercubes with basis $v_0$, the ladder set $L_{v_0,v}$ for any vertex $v\in V$ is exactly the set $\mathcal{E}^-(v)$ of $\Theta$-classes incoming into $v$. So, value $\varphi(v_0,X)$ is the distance $d(v_0,v)$ between $v_0$ and the only vertex $v$ with ladder set $L_{v_0,v} = X$.

On general median graphs, the opposite label $\opp_u(L)$ allows us to obtain the maximum distance $d(s,t)$ such that $u = m(s,t,v_0)$ and the ladder set $L_{u,s}$ is $L$.

\begin{lemma}[Relationship between medians and disjoint outgoing POFs~\cite{BeHa21}]
Let $L,L'$ be two POFs outgoing from a vertex $u$. Let $s$ (resp. $t$) be a vertex such that $u \in I(v_0,s)$ (resp. $u \in I(v_0,t)$) and $L_{u,s} = L$ (resp. $L_{u,t} = L'$). Then, $u \in I(s,t)$ if and only if $L \cap L' = \emptyset$. Therefore, given a single vertex $s$ such that $u \in I(v_0,s)$ and $L_{u,s} = L$, the maximum distance $d(s,v)$ we can have with median $m(s,v,v_0) = u$ is exactly $d(u,s)+\varphi(u,\opp_u(L))$.
\label{le:property_opp}
\end{lemma}

Going further, given a vertex $u \in V$, the maximum distance $d(s,t)$ such that $u = m(s,t,v_0)$ is the maximum value $\varphi(u,L) + \varphi(u,\opp_u(L))$, where $L$ is POF outgoing from $u$.

An algorithm was initially proposed to compute all labels $\opp_u(L)$ consisting in a brute force bounded tree search~\cite{BeHa21}. Its execution time was $\tilde{O}(2^{O(d\log d)}n)$, leading to the global same asymptotic running time (Lemma~\ref{le:slightly_super_exp}) for finding all eccentricities.

Fortunately, the quasilinear time algorithm solving WOPP (Theorem~\ref{th:solving_wopp}, Section~\ref{subsec:partitioning}) offers us the opportunity to decrease the exponential term to a simple exponential function $2^d$. For any $u \in V$, let $G_u =  G\left[V_u\right]$ be the \textit{star} graph of $u$, using a definition from~\cite{ChLaRa19}. Its vertex set $V_u$ is made up of the vertices belonging to a hypercube with basis $u$ in $G$. Graph $G_u$ is the induced subgraph of $G$ on vertex set $V_u$ (see Figure~\ref{fig:compute_opposites} for an example). Chepoi {\em et al.}~\cite{ChLaRa19} showed that graph $G_u$ is a gated/convex subgraph of $G$. This notion of star graph is essential for the proof of the following key theorem.

\begin{theorem}[Computation of labels \opp]
There is a combinatorial algorithm which determines all labels $\opp_u(L)$ in $\tilde{O}(2^dn)$. 
\label{th:compute_opp}
\end{theorem}
\begin{proof}
Let $u \in V$: we denote by $N_u$ the number of hypercubes of $G$ with basis $u$. Convex subgraphs of median graphs are also median by considering the original definition of median graphs (Definition~\ref{def:median}). Consequently, star graph $G_u$ is median and all its maximal hypercubes contain a common vertex $u$. From Theorem~\ref{th:simplex}, $G_u$ is a simplex graph.

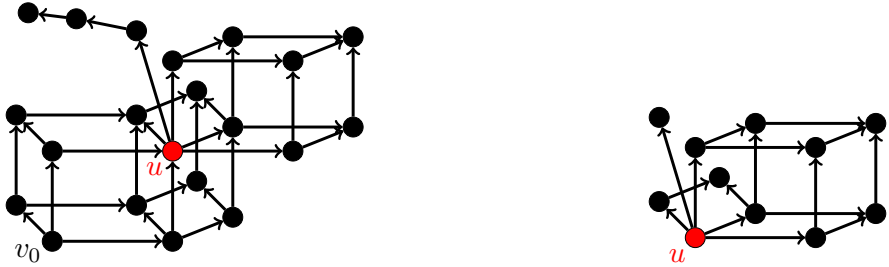
\begin{figure}[h]
\centering
\begin{subfigure}[b]{0.54\columnwidth}
\centering
\scalebox{0.8}{\begin{tikzpicture}


\node[draw, circle, minimum height=0.2cm, minimum width=0.2cm, fill=red] (P11) at (4,5) {};
\node[draw, circle, minimum height=0.2cm, minimum width=0.2cm, fill=black] (P12) at (4,6.5) {};

\node[draw, circle, minimum height=0.2cm, minimum width=0.2cm, fill=black] (P21) at (6,5) {};
\node[draw, circle, minimum height=0.2cm, minimum width=0.2cm, fill=black] (P22) at (6,6.5) {};

\node[draw, circle, minimum height=0.2cm, minimum width=0.2cm, fill=black] (P31) at (5.0,5.4) {};
\node[draw, circle, minimum height=0.2cm, minimum width=0.2cm, fill=black] (P32) at (5.0,6.9) {};
\node[draw, circle, minimum height=0.2cm, minimum width=0.2cm, fill=black] (P33) at (7.0,5.4) {};
\node[draw, circle, minimum height=0.2cm, minimum width=0.2cm, fill=black] (P34) at (7.0,6.9) {};

\node[draw, circle, minimum height=0.2cm, minimum width=0.2cm, fill=black] (P41) at (3.4,5.6) {};
\node[draw, circle, minimum height=0.2cm, minimum width=0.2cm, fill=black] (P42) at (4.4,6.0) {};

\node[draw, circle, minimum height=0.2cm, minimum width=0.2cm, fill=black] (P51) at (2.0,5.0) {};
\node[draw, circle, minimum height=0.2cm, minimum width=0.2cm, fill=black] (P52) at (1.4,5.6) {};

\node[draw, circle, minimum height=0.2cm, minimum width=0.2cm, fill=black] (P61) at (3.4,7.0) {};
\node[draw, circle, minimum height=0.2cm, minimum width=0.2cm, fill=black] (P62) at (2.4,7.2) {};
\node[draw, circle, minimum height=0.2cm, minimum width=0.2cm, fill=black] (P63) at (1.6,7.3) {};

\node[draw, circle, minimum height=0.2cm, minimum width=0.2cm, fill=black] (P11b) at (4,3.5) {};
\node[draw, circle, minimum height=0.2cm, minimum width=0.2cm, fill=black] (P51b) at (2,3.5) {};
\node[draw, circle, minimum height=0.2cm, minimum width=0.2cm, fill=black] (P41b) at (3.4,4.1) {};
\node[draw, circle, minimum height=0.2cm, minimum width=0.2cm, fill=black] (P52b) at (1.4,4.1) {};

\node[draw, circle, minimum height=0.2cm, minimum width=0.2cm, fill=black] (P31b) at (5.0,3.9) {};
\node[draw, circle, minimum height=0.2cm, minimum width=0.2cm, fill=black] (P42b) at (4.4,4.5) {};



\draw[->,line width = 1.4pt] (P11) -- (P12);

\draw[->,line width = 1.4pt] (P11) -- (P21);
\draw[->,line width = 1.4pt] (P12) -- (P22);
\draw[->,line width = 1.4pt] (P21) -- (P22);

\draw[->,line width = 1.4pt] (P11) -- (P31);
\draw[->,line width = 1.4pt] (P12) -- (P32);
\draw[->,line width = 1.4pt] (P21) -- (P33);
\draw[->,line width = 1.4pt] (P22) -- (P34);
\draw[->,line width = 1.4pt] (P31) -- (P32);
\draw[->,line width = 1.4pt] (P31) -- (P33);
\draw[->,line width = 1.4pt] (P32) -- (P34);
\draw[->,line width = 1.4pt] (P33) -- (P34);

\draw[->,line width = 1.4pt] (P11) -- (P41);
\draw[->,line width = 1.4pt] (P31) -- (P42);
\draw[->,line width = 1.4pt] (P41) -- (P42);

\draw[<-,line width = 1.4pt] (P11) -- (P51);
\draw[<-,line width = 1.4pt] (P41) -- (P52);
\draw[->,line width = 1.4pt] (P51) -- (P52);

\draw[->,line width = 1.4pt] (P11) -- (P61);
\draw[->,line width = 1.4pt] (P61) -- (P62);
\draw[->,line width = 1.4pt] (P62) -- (P63);

\draw[<-,line width = 1.4pt] (P11) -- (P11b);
\draw[<-,line width = 1.4pt] (P41) -- (P41b);
\draw[<-,line width = 1.4pt] (P51) -- (P51b);
\draw[<-,line width = 1.4pt] (P52) -- (P52b);
\draw[<-,line width = 1.4pt] (P31) -- (P31b);
\draw[<-,line width = 1.4pt] (P42) -- (P42b);

\draw[<-,line width = 1.4pt] (P11b) -- (P51b);
\draw[->,line width = 1.4pt] (P11b) -- (P41b);
\draw[<-,line width = 1.4pt] (P41b) -- (P52b);
\draw[->,line width = 1.4pt] (P51b) -- (P52b);

\draw[->,line width = 1.4pt] (P11b) -- (P31b);
\draw[->,line width = 1.4pt] (P41b) -- (P42b);
\draw[->,line width = 1.4pt] (P31b) -- (P42b);

\node[color = black, scale = 1.4] at (1.6,3.3) {$v_0$};
\node[color = red, scale = 1.4] at (3.7,4.7) {$u$};

\end{tikzpicture}}
\caption{A $v_0$-oriented median graph $G$ and a vertex $u \in V$}
\label{subfig:compute_opposites_1}
\end{subfigure}
\begin{subfigure}[b]{0.44\columnwidth}
\centering
\scalebox{0.8}{\begin{tikzpicture}


\node[draw, circle, minimum height=0.2cm, minimum width=0.2cm, fill=red] (P11) at (4,5) {};
\node[draw, circle, minimum height=0.2cm, minimum width=0.2cm, fill=black] (P12) at (4,6.5) {};

\node[draw, circle, minimum height=0.2cm, minimum width=0.2cm, fill=black] (P21) at (6,5) {};
\node[draw, circle, minimum height=0.2cm, minimum width=0.2cm, fill=black] (P22) at (6,6.5) {};

\node[draw, circle, minimum height=0.2cm, minimum width=0.2cm, fill=black] (P31) at (5.0,5.4) {};
\node[draw, circle, minimum height=0.2cm, minimum width=0.2cm, fill=black] (P32) at (5.0,6.9) {};
\node[draw, circle, minimum height=0.2cm, minimum width=0.2cm, fill=black] (P33) at (7.0,5.4) {};
\node[draw, circle, minimum height=0.2cm, minimum width=0.2cm, fill=black] (P34) at (7.0,6.9) {};

\node[draw, circle, minimum height=0.2cm, minimum width=0.2cm, fill=black] (P41) at (3.4,5.6) {};
\node[draw, circle, minimum height=0.2cm, minimum width=0.2cm, fill=black] (P42) at (4.4,6.0) {};

\node[draw, circle, minimum height=0.2cm, minimum width=0.2cm, fill=black] (P61) at (3.4,7.0) {};


\draw[->,line width = 1.4pt] (P11) -- (P12);

\draw[->,line width = 1.4pt] (P11) -- (P21);
\draw[->,line width = 1.4pt] (P12) -- (P22);
\draw[->,line width = 1.4pt] (P21) -- (P22);

\draw[->,line width = 1.4pt] (P11) -- (P31);
\draw[->,line width = 1.4pt] (P12) -- (P32);
\draw[->,line width = 1.4pt] (P21) -- (P33);
\draw[->,line width = 1.4pt] (P22) -- (P34);
\draw[->,line width = 1.4pt] (P31) -- (P32);
\draw[->,line width = 1.4pt] (P31) -- (P33);
\draw[->,line width = 1.4pt] (P32) -- (P34);
\draw[->,line width = 1.4pt] (P33) -- (P34);

\draw[->,line width = 1.4pt] (P11) -- (P41);
\draw[->,line width = 1.4pt] (P31) -- (P42);
\draw[->,line width = 1.4pt] (P41) -- (P42);

\draw[->,line width = 1.4pt] (P11) -- (P61);

\node[color = red, scale = 1.4] at (3.7,4.7) {$u$};

\end{tikzpicture}}
\caption{Star graph $G_u$}
\label{subfig:compute_opposites_2}
\end{subfigure}

\caption{Example of star graph $G_u$}
\label{fig:compute_opposites}
\end{figure}

Any pair $(u,L)$ of $G$, where $L$ is a POF outgoing from $u$ in $G$, can be associated to a unique hypercube with signature $L$ and basis $u$. Thus, there is a natural bijection between (i) the POFs of $G_u$ (ii) the vertices of $G_u$ and (iii) the POFs $L$ of $G$ outgoing from $u$. Hence, $\card{V_u} = N_u$.

We associate with any POF $L$ of $G_u$ the weight $\omega_u(L) = \varphi(u,L)$. We solve WOPP on graph $G_u$ with weight function $\omega_u$, using the algorithm evoked in Theorem~\ref{th:solving_wopp}. The opposite computed with that configuration correspond exactly to the labels $\opp_u(L)$: a POF $L'$ disjoint from $L$ and maximizing $\varphi(u,L')$ among all POFs outgoing from $u$. The running time of the algorithm is $O((d^3+\log \card{V_u})\card{V_u}) = O((d^3+\log n)N_u)$. Doing it for every vertex $u$ of $G$, we obtain all opposite labels of $G$ in $O(d^3+\log n)2^dn)$ as $\sum_{u \in V} N_u = 2^dn$ (Lemma~\ref{le:number_hypercubes}).
\end{proof}

\subsubsection{Anti-ladder labels} \label{subsubsec:anti_ladder}

We terminate with anti-ladder labels $\psi$ which play the converse role of ladder labels $\varphi$. While $\varphi(u,L)$ is defined for POFs $L$ outgoing from $u$, labels $\psi(u,R)$ are defined for POFs $R$ incoming into $u$, {\em i.e.} every $\Theta$-class of the POF $R$ has an edge entering $u$. As any such pair $(u,R)$ can be associated with a hypercube of anti-basis $u$ and signature $R$ (Lemma~\ref{le:pof_hypercube}), the total size of $\psi$-labelings is at most $O(2^dn)$ too.

The notion of \textit{milestone} intervenes in the definition of labels $\psi$. We consider two vertices $u,v$ such that $u \in I(v_0,v)$. Milestones are defined recursively.

\begin{definition}[Milestones $\Pi(u,v)$]
Let $L_{u,v}$ be the ladder set of $u,v$ and $u^+$ be the anti-basis of the hypercube with basis $u$ and signature $L_{u,v}$.
If $u^+ = v$, then pair $u,v$ admits two milestones: $\Pi(u,v) = \set{u,v}$. Otherwise, the set $\Pi(u,v)$ is the union of $\Pi(u^+,v)$ with vertex $u$: $\Pi(u,v) = \set{u} \cup \Pi(u^+,v)$.
\label{def:milestones}
\end{definition}

The milestones are the successive anti-bases of the hypercubes formed by the vertices and ladder sets traversed from $u$ to $v$. Both vertices $u$ and $v$ are contained in $\Pi(u,v)$. The first milestone is $u$, the second is the anti-basis $u^+$ of the hypercube with basis $u$ and signature $L_{u,v}$. The third one is the anti-basis $u^{++}$ of the hypercube with basis $u^+$ and signature $L_{u^+,v}$, etc. All milestones are metrically between $u$ and $v$: $\Pi(u,v) \subseteq I(u,v)$. 

\begin{definition}[Penultimate milestone $\overline{\pi}(u,v)$]
We say that the milestone in $\Pi(u,v)$ different from $v$ but the closest to it is called the \textit{penultimate milestone}. We denote it by $\overline{\pi}(u,v)$. Furthermore, we denote by $\overline{L}_{u,v}$ the \emph{anti-ladder set} of $u,v$, {\em i.e.} the $\Theta$-classes of the hypercube with basis $\overline{\pi}(u,v)$ and anti-basis $v$. 
\end{definition}

\begin{figure}[h]
\centering
\scalebox{0.9}{\begin{tikzpicture}


\node[draw, circle, minimum height=0.2cm, minimum width=0.2cm, fill=black] (P11) at (1,1) {};
\node[draw, circle, minimum height=0.2cm, minimum width=0.2cm, fill=black] (P12) at (1,2.5) {};
\node[draw, circle, minimum height=0.2cm, minimum width=0.2cm, fill=black] (P13) at (0.2,4.7) {};

\node[draw, circle, minimum height=0.2cm, minimum width=0.2cm, fill=red] (P21) at (3,1) {};
\node[draw, circle, minimum height=0.2cm, minimum width=0.2cm, fill=black] (P22) at (3,2.5) {};
\node[draw, circle, minimum height=0.2cm, minimum width=0.2cm, fill=black] (P23) at (3,4) {};
\node[draw, circle, minimum height=0.2cm, minimum width=0.2cm, fill=black] (P24) at (1.8,3.2) {};
\node[draw, circle, minimum height=0.2cm, minimum width=0.2cm, fill=black] (P25) at (1.8,4.7) {};

\node[draw, circle, minimum height=0.2cm, minimum width=0.2cm, fill=black] (P31) at (5,1) {};
\node[draw, circle, minimum height=0.2cm, minimum width=0.2cm, fill=red] (P32) at (5,2.5) {};
\node[draw, circle, minimum height=0.2cm, minimum width=0.2cm, fill=black] (P33) at (5,4) {};

\node[draw, circle, minimum height=0.2cm, minimum width=0.2cm, fill=black] (P41) at (7,1) {};
\node[draw, circle, minimum height=0.2cm, minimum width=0.2cm, fill=red] (P42) at (7,2.5) {};

\node[draw, circle, minimum height=0.2cm, minimum width=0.2cm, fill=black] (P51) at (9,1) {};
\node[draw, circle, minimum height=0.2cm, minimum width=0.2cm, fill=black] (P52) at (9,2.5) {};

\node[draw, circle, minimum height=0.2cm, minimum width=0.2cm, fill=black] (P61) at (8.0,1.4) {};
\node[draw, circle, minimum height=0.2cm, minimum width=0.2cm, fill=black] (P62) at (8.0,2.9) {};
\node[draw, circle, minimum height=0.2cm, minimum width=0.2cm, fill=black] (P63) at (10.0,1.4) {};
\node[draw, circle, minimum height=0.2cm, minimum width=0.2cm, fill=red] (P64) at (10.0,2.9) {};


\draw[line width = 1.4pt] (P11) -- (P12);
\draw[line width = 1.4pt] (P11) -- (P21);
\draw[line width = 1.4pt] (P12) -- (P22);
\draw[line width = 1.4pt,dashed,color = cyan] (P21) -- (P22);

\draw[line width = 1.4pt,dashed,color = cyan] (P21) -- (P31);
\draw[line width = 1.4pt,dashed,color = cyan] (P22) -- (P32);
\draw[line width = 1.4pt,dashed,color = cyan] (P31) -- (P32);

\draw[line width = 1.4pt] (P22) -- (P23);
\draw[line width = 1.4pt] (P23) -- (P33);
\draw[line width = 1.4pt] (P32) -- (P33);

\draw[line width = 1.4pt] (P22) -- (P24);
\draw[line width = 1.4pt] (P23) -- (P25);
\draw[line width = 1.4pt] (P24) -- (P25);
\draw[line width = 1.4pt] (P13) -- (P25);

\draw[line width = 1.4pt] (P31) -- (P41);
\draw[line width = 1.4pt, dashed, color = blue] (P32) -- (P42);
\draw[line width = 1.4pt] (P41) -- (P42);

\draw (P41) -- (P51);
\draw[line width = 1.4pt, dashed, color = purple] (P42) -- (P52);
\draw[line width = 1.4pt] (P51) -- (P52);

\draw (P41) -- (P61);
\draw[line width = 1.4pt, dashed, color = purple] (P42) -- (P62);
\draw (P51) -- (P63);
\draw[line width = 1.4pt, dashed, color = purple] (P52) -- (P64);
\draw[line width = 1.4pt] (P61) -- (P62);
\draw (P61) -- (P63);
\draw[line width = 1.4pt, dashed, color = purple] (P62) -- (P64);
\draw[line width = 1.4pt] (P63) -- (P64);


\node[scale=1.2, color = cyan] at (4.0,0.5) {$L_{u,v}$};
\node[scale=1.2, color = blue] at (6.0,2.9) {$L_{u^+,v}$};
\node[scale=1.2, color = purple] at (9.0,3.3) {$L_{u^{++},v}$};

\node[scale = 1.2] at (0.6,0.7) {$v_0$};
\node[scale = 1.2] at (2.6,0.7) {$u$};
\node[scale = 1.2] at (10.4,3.2) {$v$};

\node[scale = 1.2] at (4.6,2.2) {$u^+$};
\node[scale = 1.2] at (6.5,2.2) {$u^{++}$};

\end{tikzpicture}}
\caption{A pair $u,v$ with $u \in I(v_0,v)$ and its milestones $\Pi(u,v)$ in red.}
\label{fig:milestones}
\end{figure}
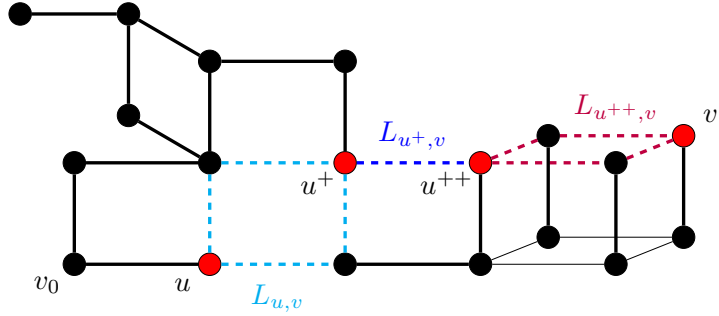

Figure~\ref{fig:milestones} shows the milestones $\Pi(u,v) = \set{u,u^+,u^{++},v}$. The hypercubes with the following pair basis-signature are highlighted with dashed edges: $(u,L_{u,v})$, $(u^+,L_{u^+,v})$, and $(u^{++},L_{u^{++},v})$. We have $\overline{\pi}(u,v) = u^{++}$ and $\overline{L}_{u,v} = L_{u^{++},v}$ is drawn in purple.

Let $R$ be a POF incoming to some vertex $u$ and $u^-$ be the basis of the hypercube with anti-basis $u$ and signature $R$.  Label $\psi(u,R)$ intuitively represents the maximum distance of a shortest path arriving to vertex $u$ from ``direction'' $R$. 

\begin{definition}[Labels $\psi$~\cite{BeHa21}]
The label $\psi(u,R)$ is the maximum distance $d(u,v)$ we can obtain with a vertex $v$ satisfying the following properties:
\begin{itemize}
\item $m = m(u,v,v_0) \neq u$,
\item the anti-ladder set of $m,u$ is $R$: $\overline{L}_{m,u} = R$.
\end{itemize}
Equivalently, vertex $u^-$ is the penultimate milestone of pair $m,u$: $u^- = \overline{\pi}(m,u)$.
\label{def:psi}
\end{definition}

As for the computation of labels $\varphi$, there is an induction process to determine all $\psi(u,R)$. As the base case, suppose that $u^- = v_0$. The largest distance $d(u,v)$ we can obtain with a vertex $v$ such that $v_0 \in I(u,v)$ consists in considering the opposite $\opp_{v_0}(R)$ of $R$ which is outgoing from $v_0$. Hence, $\psi(u,R) = \card{R} + \varphi(v_0,\opp_{v_0}(R))$.

For the induction step, we distinguish two cases. In the first one, assume that $m(u,v,v_0) = u^-$ - equivalently, $\Pi(m,u) = \Pi(u^-,u) = \set{u^-,u}$. A shortest $(u,v)$-path is the concatenation of the shortest $(u,u^-)$-path of length $\card{R}$ with a shortest $(u^-,v)$-path, and $u^- \in I(v_0,v)$. The largest distance $d(u,v)$ we can have, as for the base case, is $\psi(u,R) = \card{R} + \varphi(u^-,\opp_{u^-}(R))$.

In the second case, $m \neq u^-$, an inductive formula allows us to obtain $\psi(u,R)$. A consequence of Lemma~\ref{le:pushing_phi} is that, for two consecutive milestones in $\Pi(u,v)$, say $u$ and $u^+$ w.l.o.g, then $L_{u^+,v}$ is $L_{u,v}$-parallel. This observation, applied to the penultimate milestone, provides us with the following theorem.

\begin{lemma}[Inductive formula for labels $\psi$~\cite{BeHa21}]
Let $u,v \in V$ and $u \in I(v_0,v)$. Let $L$ be a POF outgoing from $v$ and $w$ the anti-basis of hypercube $(v,L)$. The following propositions are equivalent:
\begin{enumerate}
\item[(i)] vertex $v$ is the penultimate milestone of $(u,w)$: $\overline{\pi}(u,w) = v$,
\item[(ii)] the milestones of $(u,w)$ are the milestones of $(u,v)$ with $w$: $\Pi(u,w) = \Pi(u,v) \cup \set{w}$,
\item[(iii)] the POF $L$ is $\overline{L}_{u,v}$-parallel.
\end{enumerate}
\label{le:penultimate}
\end{lemma}

Set $\Pi(m,u)$ admits at least three milestones: $m$, $u^-$, and $u$. Let $R^-$ be the POF incoming to $u^-$ which is the ladder set (but also the signature) of (i) the milestone just before $u^-$ and (ii) $u^-$. According to Lemma~\ref{le:penultimate}, vertex $u^-$ is the penultimate milestone of $(m,u)$ if and only if $R^- \cup \set{E_j}$ is not a POF, for each $E_j \in R$. For this reason, value $\psi(u,R)$ can be expressed as:

\begin{equation}
\begin{split}
\psi(u,R)& = \max\limits_{\substack{R^- ~\mbox{\scriptsize{POF incoming to}}~ u^- \\ \forall E_j \in R, R^- \cup \set{E_j} ~\mbox{\scriptsize{not POF}}}} (\card{R} + \psi(u^-,R^-))\\
& = \max\limits_{\substack{(R^-,u^-)~\mbox{\scriptsize{s.t.}}\\ R ~\mbox{\scriptsize{is}}~ (R^-,u^-)\mbox{\scriptsize{-parallel}}}} (\card{R} + \psi(u^-,R^-)).
\end{split}
\label{eq:induction_psi}
\end{equation}

Our algorithm consists in taking the maximum value between the two cases. The number of triplets $(R^-,u,R)$ which satisfy the condition described in Equation~\eqref{eq:induction_psi} is at most $2^{2d}n$: it is identical to the one presented for $\varphi$-labelings (Theorem~\ref{th:naive_paradj}).

\begin{theorem}[Computation of labels $\psi$~\cite{BeHa21}]
There is a combinatorial algorithm which determines all labels $\psi(u,R)$ in $\tilde{O}(2^{2d}n)$. 
\label{th:compute_psi}
\end{theorem}

\subsubsection{Better time complexity for all eccentricities} \label{subsubsec:ecc}

 The computation of all labels $\varphi(u,L)$, $\opp_u(L)$ and $\psi(u,R)$ gives an algorithm which determines all eccentricities. Indeed, each eccentricity $\ecc(u)$ is a function of certain labels $\varphi$ and $\psi$. Let $v$ be a vertex in $G$ such that $\ecc(u) = d(u,v)$. If $m = m(u,v,v_0) = u$, then $u \in I(v_0,v)$ and value $d(u,v)$ is given by a label $\varphi(u,L)$. Otherwise, if $m \neq u$, let $u^-$ be the penultimate milestone in $\Pi(m,u)$ and $R$ be the classes of the hypercube with basis $u^-$ and anti-basis $u$. The eccentricity of $u$ is given by a label $\psi(u,R)$. Conversely, each $\varphi(u,L)$ and $\psi(u,R)$ is the distance between $u$ and another vertex by definition. Therefore, we have:

\begin{equation}
\ecc(u) = \mbox{max}\set{\max\limits_{\substack{L ~\mbox{\scriptsize{POF}} \\ \mbox{\scriptsize{outgoing from}}~ u}} \varphi(u,L), \max\limits_{\substack{R ~\mbox{\scriptsize{POF}} \\ \mbox{\scriptsize{incoming to}}~ u}} \psi(u,R)}
\label{eq:ecc_labels}
\end{equation}

In other words, the eccentricity of $u$ is the maximum label $\varphi$ or $\psi$ centered at $u$. We can conclude with the main result of this subsection: the eccentricities of any median graph can be determined in linear time multiplied by a simple exponential function $2^{O(d)}$ of the dimension $d$.

\begin{theorem}[All eccentricities in $\tilde{O}(2^{2d}n)$-time for median graphs]
There is a combinatorial algorithm computing the list of all eccentricities of a median graph $G$ in time $\tilde{O}(2^{2d}n)$.
\label{th:simple_ecc}
\end{theorem}
\begin{proof}
We simply determine the labels $\varphi$, $\opp$ and $\psi$ with the algorithms mentioned in Theorem~\ref{th:compute_phi},~\ref{th:compute_opp}, and~\ref{th:compute_psi}. Thanks to the time improvement obtained for the computation of opposite labels, the overall running time to compute the labels is only $\tilde{O}(2^{2d}n)$. Then, for each vertex $u$, Equation~\eqref{eq:ecc_labels} guides us to obtain its eccentricity. We take the maximum over all labels $\varphi(u,L)$ - they are at most $N_u$ - and all labels $\psi(u,R)$ - they are at most $2^d$. As $\sum_u N_u \le 2^dn$, the execution time of this operation on all vertices is $\tilde{O}(2^dn)$. Therefore, it does not overpass the time complexity needed to determine the labels.
\end{proof}

\subsection{Tackling the general case} \label{subsec:reduction}

Our new FPT algorithm for computing the list of eccentricities in a median graph has a runtime in $2^{{\cal O}(d)}n$, with $d$ being the dimension (Theorem~\ref{th:simple_ecc}). 
This runtime stays subquadratic in $n$ as long as $d < \alpha \cdot \log{n}$, for some constant $\alpha < 1$.
In what follows, we present a simple partitioning scheme for median graphs into convex subgraphs of dimension at most $\alpha \cdot \log{n}$, for an arbitrary value of $\alpha \leq 1$.
By doing so, we obtain (in combination with Theorem~\ref{th:simple_ecc}) the first known subquadratic-time algorithm for computing all eccentricities in a median graph.

We start with a simple relation between the eccentricity function of a median graph and the respective eccentricity functions of any two complementary halfspaces.

\begin{lemma}\label{lem:guigui-1}
Let $G$ be a median graph.
For every $1 \leq i \leq q$, let $v \in V(H_i')$ be arbitrary, and let $v^*$ be its gate in $\partial H_i''$.
Then, $\ecc(v) = \max\{\ecc_{H_i'}(v), d(v,v^*) + \ecc_{H_i''}(v^*)\}$.
\end{lemma}
\begin{proof}
We have $\ecc(v) = \ecc_G(v) = \max\{d(u,v) \mid u \in V(H_i')\} \cup \{d(w,v) \mid w \in V(H_i'')\}$.
Since $H_i'$ is convex, we have $\max\{d(u,v) \mid u \in V(H_i')\} = \ecc_{H_i'}(v)$.
In the same way, since $H_i''$ is gated (and so, convex), we have $\max\{d(w,v) \mid w \in V(H_i'')\} = d(v,v^*) + \max\{d(v^*,w) \mid w \in V(H_i'')\} = d(v,v^*) + \ecc_{H_i''}(v^*)$.
\end{proof}
We will use this above Lemma~\ref{lem:guigui-1} later in our proof in order to compute in linear time the list of eccentricities in a median graph being given the lists of eccentricities in any two complementary halfspaces.

Next, we give simple properties of $\Theta$-classes, to be used in the analysis of our main algorithm in this section (see Lemma~\ref{lem:guigui-4}).

\begin{lemma}\label{lem:guigui-2}
Let $H$ and $G$ be median graphs.
If $H$ is an induced subgraph of $G$ then, every $\Theta$-class of $H$ is contained in a $\Theta$-class of $G$.
\end{lemma}
\begin{proof}
Every square of $H$ is also a square of $G$.
In particular, two edges of $H$ are in relation $\Theta_0$ if and only if, as edges of $G$, they are also in relation $\Theta_0$.
Since the $\Theta$-classes of $H$ (resp., of $G$) are the transitive closure of its relation $\Theta_0$, it follows that every $\Theta$-class of $H$ is contained in a $\Theta$-class of $G$.
\end{proof}

This above Lemma~\ref{lem:guigui-2} can be strenghtened in the special case of {\em isometric} subgraphs, namely:

\begin{lemma}\label{lem:guigui-2bis}
Let $H$ and $G$ be median graphs, and let $E_1,E_2,\ldots,E_q$ denote the $\Theta$-classes of $G$.
If $H$ is an isometric subgraph of $G$ then, the $\Theta$-classes of $H$ are exactly the nonempty subsets among $E_i \cap E(H)$, for $1 \leq i \leq q$.
\end{lemma}
\begin{proof}
It is known~\cite{winkler1984isometric} that two edges $uv,xy$ of $G$ are in the same $\Theta$-class if and only if $d_G(u,x) + d_G(v,y) \neq d_G(u,y) + d_G(v,x)$.
In particular, since $H$ is isometric in $G$, two edges of $H$ are in the same $\Theta$-class of $H$ if and only if they are in the same $\Theta$-class of $G$.
\end{proof}

An important consequence of Lemma~\ref{lem:guigui-2} is the following relation between the dimension $d$ of a median graph and the cardinality of its $\Theta$-classes.

\begin{lemma}\label{lem:guigui-3}
Let $G$ be a median graph, and let $D := \max\{ |E_i| \mid 1 \leq i \leq q\}$ be the maximum cardinality of a $\Theta$-class of $G$. Then, $d = \dim(G) \leq \lfloor\log{D}\rfloor + 1$.
\end{lemma}
\begin{proof}
Any induced $d$-dimensional hypercube of $G$ contains exactly $2^{d-1}$ edges of its $\Theta$-classes.
\end{proof}

We are now ready to present our main technical contribution in this section.

\begin{lemma}\label{lem:guigui-4}
If there is an algorithm for computing all eccentricities in an $n$-vertex median graph of dimension at most $d$ in $\tilde{O}(c^d \cdot n)$ time, then in $\tilde{O}(n^{2 - \frac 1 {1+\log{c}}})$ time we can compute all eccentricities in {\em any} $n$-vertex median graph.
\end{lemma}
\begin{proof}
Let $G$ be an $n$-vertex median graph.
We compute its $\Theta$-classes $E_1,E_2,\ldots,E_q$, that takes linear time (Lemma~\ref{le:linear_classes}).
For some parameter $D$ (to be fixed later in the proof), let us assume without loss of generality $E_1,E_2,\ldots,E_p$ to be the subset of all $\Theta$-classes of cardinality $\geq D$, for some $p \leq q$. Note that we have $p \leq m/D = \tilde{O}(n/D)$, where $m$ is the number of edges in $G$.

We reduce the problem of computing all eccentricities in $G$ to the same problem on every connected component of $G \setminus (E_1 \cup E_2 \cup \ldots \cup E_p)$. 
More formally, we construct a rooted binary tree\footnote{This tree $T$ is independent to the one built in Section~\ref{subsec:partitioning}} $T$, whose leaves are labelled with convex subgraphs of $G$.
Initially, $T$ is reduced to a single node with label equal to $G$.
Then, for $i = 1 \ldots p$, we further refine this tree so that, at the end of any step $i$, its leaves are labelled with the connected components of $G \setminus \left(E_1 \cup E_2 \cup \ldots \cup E_i\right)$. An example of $T$ is shown in Figure~\ref{fig:reduction} with $D=3$ and two $\Theta$-classes reaching this cardinality bound.

For that, we proceed as follows.
We consider all leaves of $T$ whose label $H$ satisfies $E(H) \cap E_i \neq \emptyset$.
By Lemma~\ref{lem:guigui-2bis}, $E(H) \cap E_i$ is a $\Theta$-class of $H$.
Both halfspaces of $E_i$ become the left and right children of $H$ in $T$.
Recall that the leaves of $T$ at this step $i$ are the connected components of $G \setminus (E_1 \cup E_2 \cup \ldots \cup E_{i-1})$, and in particular that they form a partition of $V(G)$.
Therefore, each step takes linear time by reduction to computing the connected components in vertex-disjoint subgraphs of $G$.
Overall, the total time for constructing the tree $T$ is in ${O}(pm) = \tilde{O}(n^2/D)$.

\begin{figure}[h]
\centering
\scalebox{0.7}{\begin{tikzpicture}


\node[draw, circle, minimum height=0.2cm, minimum width=0.2cm, fill=black] (P11) at (1,1) {};
\node[draw, circle, minimum height=0.2cm, minimum width=0.2cm, fill=black] (P12) at (1,2) {};

\node[draw, circle, minimum height=0.2cm, minimum width=0.2cm, fill=black] (P21) at (2.5,1) {};
\node[draw, circle, minimum height=0.2cm, minimum width=0.2cm, fill=black] (P22) at (2.5,2) {};
\node[draw, circle, minimum height=0.2cm, minimum width=0.2cm, fill=black] (P23) at (2.5,3) {};

\node[draw, circle, minimum height=0.2cm, minimum width=0.2cm, fill=black] (P31) at (4,1) {};
\node[draw, circle, minimum height=0.2cm, minimum width=0.2cm, fill=black] (P32) at (4,2) {};
\node[draw, circle, minimum height=0.2cm, minimum width=0.2cm, fill=black] (P33) at (4,3) {};

\node[draw, circle, minimum height=0.2cm, minimum width=0.2cm, fill=black] (P4) at (5.5,3) {};

\draw[line width = 1.4pt, color = green] (P11) -- (P12);
\draw[line width = 1.4pt] (P11) -- (P21);
\draw[line width = 1.4pt] (P12) -- (P22);
\draw[line width = 1.4pt, color = green] (P21) -- (P22);

\draw[line width = 1.4pt, color = blue] (P21) -- (P31);
\draw[line width = 1.4pt, color = blue] (P22) -- (P32);
\draw[line width = 1.4pt, color = green] (P31) -- (P32);

\draw[line width = 1.4pt] (P22) -- (P23);
\draw[line width = 1.4pt, color = blue] (P23) -- (P33);
\draw[line width = 1.4pt] (P32) -- (P33);
\draw[line width = 1.4pt] (P33) -- (P4);


\node[draw, circle, minimum height=0.2cm, minimum width=0.2cm, fill=black] (P71) at (9,1) {};
\node[draw, circle, minimum height=0.2cm, minimum width=0.2cm, fill=black] (P72) at (9,2) {};

\node[draw, circle, minimum height=0.2cm, minimum width=0.2cm, fill=black] (P81) at (10.5,1) {};
\node[draw, circle, minimum height=0.2cm, minimum width=0.2cm, fill=black] (P82) at (10.5,2) {};
\node[draw, circle, minimum height=0.2cm, minimum width=0.2cm, fill=black] (P83) at (10.5,3) {};

\draw[line width = 1.4pt, color = green] (P71) -- (P72);
\draw[line width = 1.4pt] (P71) -- (P81);
\draw[line width = 1.4pt, color = green] (P81) -- (P82);
\draw[line width = 1.4pt] (P72) -- (P82);

\draw[line width = 1.4pt] (P82) -- (P83);


\node[draw, circle, minimum height=0.2cm, minimum width=0.2cm, fill=black] (P51) at (9,4) {};
\node[draw, circle, minimum height=0.2cm, minimum width=0.2cm, fill=black] (P52) at (9,5) {};
\node[draw, circle, minimum height=0.2cm, minimum width=0.2cm, fill=black] (P53) at (9,6) {};

\node[draw, circle, minimum height=0.2cm, minimum width=0.2cm, fill=black] (P61) at (10.5,6) {};

\draw[line width = 1.4pt, color = green] (P51) -- (P52);
\draw[line width = 1.4pt] (P52) -- (P53);
\draw[line width = 1.4pt] (P53) -- (P61);


\node[draw, circle, minimum height=0.2cm, minimum width=0.2cm, fill=black] (P91) at (14,1) {};
\node[draw, circle, minimum height=0.2cm, minimum width=0.2cm, fill=black] (P92) at (15.5,1) {};

\draw[line width = 1.4pt] (P91) -- (P92);


\node[draw, circle, minimum height=0.2cm, minimum width=0.2cm, fill=black] (Pa1) at (14,2) {};
\node[draw, circle, minimum height=0.2cm, minimum width=0.2cm, fill=black] (Pa2) at (15.5,2) {};
\node[draw, circle, minimum height=0.2cm, minimum width=0.2cm, fill=black] (Pa3) at (15.5,3) {};

\draw[line width = 1.4pt] (Pa1) -- (Pa2);
\draw[line width = 1.4pt] (Pa2) -- (Pa3);


\node[draw, circle, minimum height=0.2cm, minimum width=0.2cm, fill=black] (Pb1) at (14,4) {};


\node[draw, circle, minimum height=0.2cm, minimum width=0.2cm, fill=black] (Pc1) at (14,5) {};
\node[draw, circle, minimum height=0.2cm, minimum width=0.2cm, fill=black] (Pc2) at (14,6) {};
\node[draw, circle, minimum height=0.2cm, minimum width=0.2cm, fill=black] (Pc3) at (15.5,6) {};

\draw[line width = 1.4pt] (Pc1) -- (Pc2);
\draw[line width = 1.4pt] (Pc2) -- (Pc3);


\draw[line width = 1.4pt] (-0.2,0.6) -- (-0.2,4.2) -- (6.0,4.2) -- (6.0,0.6) -- (-0.2,0.6);
\draw[line width = 1.4pt] (8.5,3.7) -- (8.5,6.4) -- (11.0,6.4) -- (11.0,3.7) -- (8.5,3.7);
\draw[line width = 1.4pt] (8.5,0.6) -- (8.5,3.3) -- (11.0,3.3) -- (11.0,0.6) -- (8.5,0.6);
\draw[line width = 2.5pt, color = red] (6.0,3) -- (8.5,5.0);
\draw[line width = 2.5pt, color = red] (6.0,2.5) -- (8.5,2.0);

\draw[line width = 1.4pt] (13.5,3.7) -- (13.5,4.3) -- (14.5,4.3) -- (14.5,3.7) -- (13.5,3.7);
\draw[line width = 1.4pt] (13.5,4.7) -- (13.5,6.4) -- (16.0,6.4) -- (16.0,4.7) -- (13.5,4.7);
\draw[line width = 1.4pt] (13.5,0.6) -- (13.5,1.3) -- (16.0,1.3) -- (16.0,0.6) -- (13.5,0.6);
\draw[line width = 1.4pt] (13.5,1.7) -- (13.5,3.3) -- (16.0,3.3) -- (16.0,1.7) -- (13.5,1.7);
\draw[line width = 2.5pt, color = red] (11.0,2.2) -- (13.5,2.5);
\draw[line width = 2.5pt, color = red] (11.0,1.5) -- (13.5,1.0);
\draw[line width = 2.5pt, color = red] (11.0,4.5) -- (13.5,4.0);
\draw[line width = 2.5pt, color = red] (11.0,5.2) -- (13.5,5.5);


\node[color = blue, scale = 1.5] at (3.25,3.5) {$E_1$};
\node[color = green, scale = 1.5] at (0.4,1.5) {$E_2$};

\node[color = red, scale = 1.5] at (7.2,4.6) {$H_1''$};
\node[color = red, scale = 1.5] at (7.2,1.8) {$H_1'$};

\end{tikzpicture}}
\caption{An example of tree $T$ associated with a graph $G$ for $D=3$: here, $p=2$.}
\label{fig:reduction}
\end{figure}
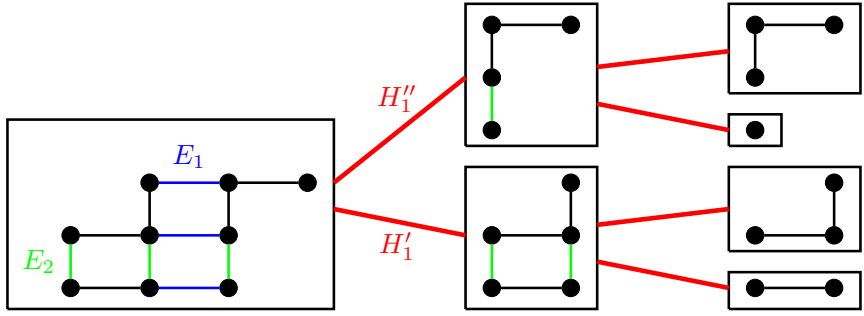

Then, we compute the list of eccentricities for all the subgraphs labelling a node, by dynamic programming on $T$. 
In particular, doing so we compute the list of eccentricities for $G$ because it is the label of the root.
There are two cases:
\begin{itemize}
    \item If $H$ labels a leaf (base case) then, we claim that we have $\mbox{dim}(H) \leq \lfloor \log{D} \rfloor + 1$.
Indeed, by Lemma~\ref{lem:guigui-2}, every $\Theta$-class of $H$ is contained in a $\Theta$-class of $G$.
Since we removed all $\Theta$-classes of $G$ with at least $D$ edges, the claim now follows from Lemma~\ref{lem:guigui-3}.
In particular, we can compute the list of all eccentricities for $H$ in $\tilde{O}(c^{\lfloor \log{D} \rfloor + 1}|V(H)|) = \tilde{O}(D^{\log{c}}|V(H)|)$ time.
Recall that the leaves of $T$ partition $V(G)$, and therefore, the total runtime for computing the list of eccentricities for the leaves is in $\tilde{O}(D^{\log{c}}n)$.
    \item From now on, let us assume $H$ labels an internal node of $T$ (inductive case).
Let $H_i',H_i''$ be its children nodes, obtained from the removal of $E(H) \cap E_i$ for some $1 \leq i \leq p$.
-- For convenience, we will say later in the proof that $H$ is an $i$-node. --
Recall that $E(H) \cap E_i$ is a $\Theta$-class of $H$.
In particular, $H_i',H_i''$ are gated subgraphs.
By Lemma~\ref{lem:guigui-1}, we can compute in ${O}(|V(H_i')|)$ time the eccentricities in $H$ of all vertices in $H_i'$ if we are given as input: the list of eccentricities in $H_i'$, the list of eccentricities in $H_i''$, and for every $v \in V(H_i')$ its gate $v^* \in \partial H_i''$ and the distance $d(v,v^*)$. 
The respective lists of eccentricities for $H_i'$ and $H_i''$ were pre-computed by dynamic programming on $T$.
Furthermore, we can compute the gate $v^*$ and $d(v,v^*)$ for every vertex $v \in V(H_i')$, in total $\tilde{O}(|V(H)|)$ time, by using a modified BFS rooted at $H_i''$ (we refer to~\cite[Lemma 17]{ChLaRa19} for a detailed description of this standard procedure).
Overall (by proceeding the same way for $H_i''$ as for $H_i'$) we can compute the list of eccentricities for $H$ in $\tilde{O}(|V(H)|)$ time.
This is in total $\tilde{O}(n)$ time for the $i$-nodes ({\it i.e.}, because they were leaves of $T$ at step $i$, and therefore, they are vertex-disjoint), and so, in total $\tilde{O}(pn) = \tilde{ O}(n^2/D)$ time for all the internal nodes.
\end{itemize}
The total runtime for our algorithm is in $\tilde{O}(n^2/D + D^{\log{c}}n)$, that is optimized for $D = n^{\frac 1 {\log{c}+1}}$.
\end{proof}

\begin{theorem}\label{thm:guigui-5}
There is an $\tilde{O}(n^{5/3})$-time algorithm for computing all eccentricities in any $n$-vertex median graph.
\end{theorem}
\begin{proof}
This result directly follows from the combination of Theorem~\ref{th:simple_ecc} with Lemma~\ref{lem:guigui-4} (applied for $c = 4$).
\end{proof}

Observe that the design of a linear FPT algorithm for eccentricities in $\tilde{O}(c^dn)$ with $c < 4$ would imply a lower subquadratic constant for this problem. Such an algorithm is proposed in Section~\ref{subsec:faster_enum}.

\section{Generalization and improvements} \label{sec:discussion}

In this section, we discuss some consequences and possible improvements of the algorithms established in Section~\ref{sec:subquadratic}. 

First, we focus on another metric parameter called \textit{reach centrality}. We prove that it can be computed if the labelings $\varphi$, $\opp$, and $\psi$ are already known. A consequence is the existence of an exact algorithm for the reach centrality in $\tilde{O}(2^{3d}n)$ on median graphs. 

Second, we propose a discrete structure strongly related to both POFs and hypercubes but slightly different to them: \textit{Maximal Outgoing POFs}, also called MOPs. We introduce a different way to compute labelings $\varphi$, $\opp$ and $\psi$ based on this structure. This yields a second subquadratic-time algorithm which determines all eccentricities with a better running time.

Finally, we introduce a new relationship between POFs: the \textit{minimal parallelism}. It allows us to design a linear FPT algorithm for eccentricities on median graphs in time $\tilde{O}(3.5394^dn)$, which outperforms the running time obtained in Section~\ref{subsec:constant_dim}. Furthemore, we deduce from it a third subquadratic-time algorithm which provides us with the best running time so far: $\tilde{O}(n^{1.6408})$.

\subsection{Reach centrality} \label{subsec:reach_centrality}

In this subsection, we propose a linear FPT algorithm, parameterized by $d$, dedicated to the computation of all reach centralities of a median graph $G$. The \textit{reach centrality} $\rc(u)$ of a vertex $u$ is a parameter related to the length of shortest paths on which vertex $u$ lies. The farther a vertex $u$ is from the two extremities of a shortest path traversing it, the larger the reach centrality of $u$ is. This notion originally inspired some efficient routing strategies on road networks~\cite{Gu04}. The relationship between reach centrality and the well-known metric parameters has been studied: Abboud {\em et al.}~\cite{AbGrWi15} proved that determining the diameter and the reach centrality are equivalent under subcubic reductions. The formal definition of $\rc(u)$ follows.

\begin{equation}
\rc(u) = \max_{u \in I(s,t)} \min \set{d(s,u),d(u,t)}
\label{eq:reach_centrality}
\end{equation}

In Theorem~\ref{th:simple_ecc}, we showed that all eccentricities of a median graph are functions of the labelings $\varphi$, $\opp$, and $\psi$. Here, a similar result is established for the reach centralities.

We begin with a first observation which will be useful to state the dependence of $\rc$ on the labels already computed. Under a certain orthogonality condition (the $R$-parallelism of $L$), for any pair of POFs $R$ and $L$ respectively incoming in and outgoing from $u \in V$, there are two vertices $s,t$ such that $d(u,s) = \varphi(u,L)$, $d(u,t) = \psi(u,R)$ and $u \in I(s,t)$.

\begin{theorem}
Let $u \in V$, $R$ be a POF incoming into $u$ and $L$ a POF outgoing from $u$ such that $L$ is $R$-parallel. There exists a pair $(s,t)$ of vertices satisfying the following properties:
\begin{itemize}
    \item Vertex $u$ belongs to interval $I(v_0,s)$ and $L_{u,s} = L$,
    \item The median $m = m(s,t,v_0)$ is different from $u$ and $\overline{L}_{m,u} = R$.
    \item The distance $d(u,s)$ and $d(u,t)$ are given by the labels: $d(u,s) = \varphi(u,L)$ and $d(u,t) = \psi(u,R)$.
    \item Vertex $u$ belongs to the interval $I(s,t)$.
\end{itemize}
\label{th:psi_plus_phi}
\end{theorem}
\begin{proof}
Let $s$ be a vertex such that $u \in I(v_0,s)$, $L_{u,s} = L$ and $d(u,s) = \varphi(u,L)$. Let $t$ be a vertex such that $\overline{L}_{m,u} = R$ for $m = m(u,t,v_0)$, and $d(u,t) = \psi(u,R)$.
By definition of labels $\varphi$ and $\psi$, such vertices exist.
At this moment, the three first bullets are verified.
To show the fourth one, $u \in I(s,t)$, we prove that the signatures $\sigma_{u,s}$ and $\sigma_{u,t}$ are disjoint. Assume there is a $\Theta$-class $E_i \in \sigma_{u,s} \cap \sigma_{u,t}$.

\textit{Claim 1:} $u \in H_i'$.
As $E_i \in \sigma_{u,s}$, any shortest $(u,s)$-path contains an edge of $E_i$. Let $(u',v')$ be one of these edges of $E_i$ which is as close as possible from $u$. Vertex $u'$ - the endpoint of this edge closer to $u$ - belongs to $\partial H_i'$. Moreover, there is a shortest $(u,s)$-path $P_{u,s}$ passing through $(u',v')$ because $v'$ is the gate of $u$ in $H_i''$. As $u \in I(v_0,s)$ and $u' \in I(u,s)$, then $u \in I(v_0,u')$. We know that both $v_0$ and $u'$ belong to $H_i'$: by convexity of halfspaces, $u \in H_i'$. 

\textit{Claim 2:} $u \in H_i''$. We have $E_i \in \sigma_{u,t}$: we prove that $E_i$ is necessarily in $\sigma_{m,u}$ and not in $\sigma_{m,t}$. The class $E_i$ cannot form a POF if we add it to $R$. We already know it if $E_i \in L$. We prove that: if it was the case for a class $E_i \in \sigma_{u,s} \backslash L$, then it would imply the orthogonality of $R$ and all $\Theta$-classes of $L$, a contradiction. Indeed, let $z$ be the vertex such that $\mathcal{E}^-(z) = R \cup \set{E_i}$. We have $v',z \in \partial H_i''$. As $v'$ is the gate of $u$ in $H_i''$, there is a shortest $(u,z)$-path passing through $(u',v')$. We denote by $\partial H_R''$ the intersection of all $\partial H_j''$ for all $E_j \in R$. As $u,z \in \partial H_R''$, which is convex, all vertices metrically between $u$ and $z$ belong to $\partial H_R''$, in particular $v'$. The ladder set of $u,v'$ is the same as $u,s$ because $v'$ does not belong to the hypercube $Q_{u,L}$ of basis $u$ and signature $L$: $L_{u,v'} = L$. In brief, $Q_{u,L} \subsetneq I(u,v')$. So, all $\Theta$-classes of $R$ are adjacent to the vertices of $Q_{u,L}$, a contradiction as $R \cup L$ is not a POF.
 
We know now that there is a class $E_j \in R$ such that $E_i$ and $E_j$ are parallel. Thus, $H_i'' \subsetneq H_j''$. Suppose, by way of contradiction, that $E_i \in \sigma_{m,t}$. Vertex $t$ is in $H_i''$, so it is also in $H_j''$. As $E_j \in R \subseteq \sigma_{m,u}$, then $m \in H_j'$: $m$ and $t$ are not in the same halfspace of $E_j$. In other words, $E_j \in \sigma_{m,t}$. This is a contradiction because $E_j$ is both in $\sigma_{m,u}$ and $\sigma_{m,t}$ while $m$ is metrically between $u$ and $t$. Any shortest $(u,t)$-path should pass through two edges of $E_j$, which is not possible from Lemma~\ref{le:signature}. Finally, $E_i \in \sigma_{m,u}$ and $u \in H_i''$.

Both Claims 1 and 2 yield a contradiction: the signature sets $\sigma_{u,s}$ and $\sigma_{u,t}$ are disjoint. Therefore, $u \in I(s,t)$.
\end{proof}

We present now an algorithm which determines, for any vertex $u \in V$, a label $\chi(u)$. The vertices can be considered in any arbitrary order. The objective is to obtain, at the end of the execution, $\chi(u) = \rc(u)$, for any vertex $u$. To start, we fix all $\chi(u)$ equal to 0. Let $\xleftarrow{\max}$ be the operator which modify the left-hand side variable with the maximum between itself and the right-hand side one. Formally, for $a \in \mathbb{N}$, $\chi(u) \xleftarrow{\max} a$ is equivalent to $\chi(u) \leftarrow \max \set{\chi(u),a}$. Given a vertex $u \in V$, we proceed in three steps.

\textbf{Step 1: Reach when $u$ is the median of $s,t,v_0$}. This step amounts to determining the reach centrality of $u$ if we restrict ourselves to pairs $s,t$ such that $u = m(s,t,v_0)$. Given a vertex $s$ such that $u \in I(v_0,s)$, the extremity $t$ maximizing $d(u,t)$ such that $u = m(s,t,v_0)$ is at distance $\varphi(u,\opp(L_{u,s}))$, according to Lemma~\ref{le:property_opp}. If $d(u,s) \le \varphi(u,L_{u,s}) \le \varphi(u,\opp_u(L_{u,s}))$, then distance $d(u,s)$ has no influence on $\rc(u)$ as another candidate - any vertex at distance $\varphi(u,L_{u,s})$ from $u$ - overpasses it. Moreover, if $d(u,s) \le \varphi(u,\opp_u(L_{u,s})) \le \varphi(u,L_{u,s})$, then $d(u,s)$ also cannot be equal to $\rc(u)$ because $\varphi(u,\opp_u(L_{u,s}))$ overpasses it and will count, according to Lemma~\ref{le:property_opp} and the fact that $\varphi(u,L_{u,s})$ is greater than it. Eventually, if $d(u,s) > \varphi(u,\opp_u(L_{u,s}))$, it does not count as we cannot form a pair $(s,t)$ such that $u\in I(s,t)$ and $d(s,u) \le d(u,t)$. In summary, the only values that have to be taken into account for the reach centrality when $u$ is a median are the $\varphi$-labelings $\varphi(u,L)$. We modify label $\chi(u)$ according to these observations.

For any POF $L$ outgoing from $u$, if $\varphi(u,L) < \varphi(u,\opp_u(L))$, then modify the label $\chi(u) \xleftarrow{\max} \varphi(u,L)$, otherwise do nothing.

\textbf{Step 2: Reach when $u\neq m=m(s,t,v_0)$ but is a milestone of $m,s$}. Let $L$ be the ladder set of $u,s$; $L = L_{u,s}$ and $R$ the anti-ladder set of $m,u$; $R = \overline{L}_{m,u}$. According to Lemma~\ref{le:penultimate}, $L$ is $R$-parallel. Theorem~\ref{th:psi_plus_phi} intervenes: there is a pair of vertices $s^*,t^*$ such that $u \in I(s,t)$, $L_{u,s^*} = L$, $\overline{L}_{m^*,u} = R$, $d(u,s^*) = \varphi(u,L)$ and $d(u,t^*) = \psi(u,R)$, where $m^* = m(s^*,t^*,v_0)$. As $\varphi(u,L) \ge d(u,s)$ and $\psi(u,R) \ge d(u,t)$, the  reach centrality of $u$ in this step can be written as a function of only  $\varphi,\psi$-labelings. For example, if a POF $L$ admits an anti-ladder set $R$ such that $\varphi(u,L) < \psi(u,R)$, then it has to be taken into account for the computation of $\chi(u)$.

For any pair $L,R$ of POFs respectively outgoing from and incoming to $u$ such that $L$ is $R$-parallel: if $\varphi(u,L) < \psi(u,R)$, we modify the label $\chi(u) \xleftarrow{\max} \varphi(u,L)$. Otherwise, we set $\chi(u) \xleftarrow{\max} \psi(u,R)$.

\textbf{Step 3: Reach when $u\neq m=m(s,t,v_0)$ and is not a milestone of $m,s$}. Let $u'$ be the milestone of $\Pi(m,s)$ and its $u''$ its successor in $\Pi(m,s)$ such that $u$ belongs to the hypercube of basis $u'$, anti-basis $u''$ and, hence, signature $L = L_{u',u''} = \sigma_{u',u''}$. We distinguish two cases.

- \textit{Case 1: $u'$ is the median of $s,t,v_0$}. In this case, the distance $d(u,s)$ is less than $\varphi(u',L)-d(u',u)$. Moreover, the distance $d(u,t)$ is less than $d(u',u)+\varphi(u',\opp_{u'}(L))$. So, the  reach centrality can be expressed only as a function of labels $\varphi$, $\psi$, and the distance $d(u',u)$.

For any POF $L$ outgoing from some $u' \in V$ such that $u$ belongs to the hypercube of basis $u'$ and signature $L$, if $\varphi(u',L)-d(u',u) < d(u',u)+\varphi(u',\opp_{u'}(L))$, then we modify the label $\chi(u) \xleftarrow{\max} \varphi(u',L)-d(u',u)$. Otherwise, we set $\chi(u) \xleftarrow{\max} d(u',u)+\varphi(u',\opp_{u'}(L))$.

- \textit{Case 2: $u'$ is not the median of $s,t,v_0$}. Let $R$ be the anti-ladder set of $m,u'$. Theorem~\ref{th:psi_plus_phi} implies the existence of a pair of vertices $s^*,t^*$ with the same (anti-)ladder sets $L$ and $R$ than $u'$ with $s,t$ and such that $d(u',s^*) = \varphi(u',L)$ and $d(u',t^*) = \psi(u',R)$. Furthermore, Lemma~\ref{le:ladder_POF} ensures us that a shortest path between $u'$ and $s^*$ can be prefixed with the $\Theta$-classes of $L$ in any ordering. As a consequence, there is a shortest $(u',s^*)$-path containing $u$. As $\varphi(u',L)-d(u',u) \ge d(u,s)$ and $\psi(u',R) + d(u',u) \ge d(u,t)$, the  reach centrality of $u$ in this step can be written only as a function of  $\varphi,\psi$-labelings and distance $d(u,u')$.

For any pair $L,R$ of POFs respectively outgoing from and incoming to some $u' \in V$ such that $L$ is $R$-parallel: enumerate all vertices $u$ belonging to the hypercube of basis $u'$ and signature $L$. For each of them, if $\varphi(u',L)-d(u',u) < \psi(u,R)+d(u',u)$, we modify the label $\chi(u) \xleftarrow{\max} \varphi(u',L)-d(u',u)$. Otherwise, we set $\chi(u) \xleftarrow{\max} \psi(u',R)+d(u',u)$.

\textbf{Pseudocode}. Algorithm~\ref{algo:chi} provides us with the pseudocode of this procedure. The steps corresponding to the updates of $\chi(u)$ are mentioned as comments, surrounded by symbol \#.

\begin{algorithm}[h]
\SetKwFor{For}{for}{do}{\nl endfor}
\SetKwFor{Forall}{for all}{do}{\nl endfor}
\SetKwIF{If}{ElseIf}{Else}{if}{then}{else if}{else}{}
\DontPrintSemicolon
\SetNlSty{}{}{:}
\SetAlgoNlRelativeSize{0}
\SetNlSkip{1em}
\nl\KwIn{Median graph $G$, weight function $C : V \rightarrow \mathbb{N}$, labels $\varphi, \opp, \psi$, $\varphi, \opp, \psi$.}
\nl\KwOut{Labels $\chi(u)$ for any vertex $u \in V$.}
\nl Initialize $\chi(u) \leftarrow 0$ for any vertex $u$;\;
\nl \For{every pair $(u,L)$ where $L$ is a POF outgoing from $u$\label{line:for_triplets}}{
    \nl \If{$\varphi(u,L) < \varphi(u,\opp_{u}(L))$}{
        \nl $\chi(u) \xleftarrow{\max} \varphi(u,L)$; \# Step 1 \# \label{line:step1}\;
    }
	\nl \ifend\;
	\nl \For{every vertex $u^*$ belonging to the hypercube of basis $u$ and signature $L$}{
	    \nl $\chi(u^*) \xleftarrow{\max} \min \set{\varphi(u,L) - d(u,u^*), \varphi(u,\opp_{u}(L)) + d(u,u^*)}$; \# Step 3-1 \# \label{line:step31}\; 
    }

\nl \For{every POF $R$ incoming into $u$ such that $L$ is $R$-parallel \label{line:for_anti_ladder}}{
	\nl $\chi(u) \xleftarrow{\max} \min \set{\varphi(u,L),\psi(u,R)} $ \# Step 2 \# \label{line:step2};\;
	\nl \For{every vertex $u^*$ belonging to the hypercube of basis $u$ and signature $L$}{
	    \nl $\chi(u^*) \xleftarrow{\max} \min \set{\varphi(u,L) - d(u,u^*), \psi(u,R) + d(u,u^*)}$; \# Step 3-2 \# \label{line:step32}\;
	}
}
}
\caption{Computation of labels $\chi$}
\label{algo:chi}
\end{algorithm}

The computation of all labels $\varphi$, $\opp$, and $\psi$  is a necessary preprocessing of this algorithm. We remind the reader that they can be obtained in $\tilde{O}(2^{2d}n)$. Steps 1, 2 and 3 cover all possible configurations of triplet $u,s,t$ such that $\rc(u) = \min \set{d(s,u),d(s,t)}$. Indeed, either $u$ is the median of $s,t,v_0$ (Step 1) or not. If not, it is either a milestone of at least one pair among $(m(s,t,v_0),s)$ and $(m(s,t,v_0),t)$ (Step 2), or not (Step 3). In each situation, both distances $d(s,u)$ and $d(u,t)$ are upper-bounded in function of some label values. Conversely, these upper bounds correspond to the distance between $u$ and certain vertices $s^*,t^*$, such that $u \in I(s^*,t^*)$. Hence, $\rc(u)$ can be expressed as a function of labelings $\varphi$, $\opp$, and $\psi$, as described in Algorithm~\ref{algo:chi}.

\begin{theorem}
There is a combinatorial algorithm computing all  reach centralities $\rc(u)$ of a median graph in $\tilde{O}(2^{3d}n)$.
\label{th:simple_rc}
\end{theorem}
\begin{proof}
The correctness of Algorithm~\ref{algo:chi} is now clear. We focus on its runtime. The most expensive part corresponds to Step 3, Case 2 (line~\ref{line:step32}). Indeed, we enumerate all triplets $(u,L,R)$: we know they are at most $2^{2d}n$. For each of them, we list all vertices lying on the hypercube of basis $u$ and signature $L$, which contains potentially $2^d$ vertices. The total number of 4-uplets $(u,L,R,u^*)$ considered in line~\ref{line:step32} of Algorithm~\ref{algo:chi} is thus at most $2^{3d}n$.
\end{proof}

\subsection{MOP structure} \label{subsec:mop}

In this subsection, we introduce a new discrete structure for median graphs, called Maximal Outgoing POFs (MOPs). Each MOP refers to a unique hypercube but the reverse is false. We present a less trivial way to compute the labels $\varphi$, $\opp$ and $\psi$ based on the enumeration of MOPs. Thanks to this result, we obtain an improvement of Theorem~\ref{thm:guigui-5} via a win-win approach. When $d \le a^*\log n$ (value $a^*<1$ will be determined in the proof), we can apply Theorem~\ref{th:simple_ecc}. Otherwise, when $d > a^* \log n$, we show that $G$ admits a subquadratic number of MOPs and the labels can be computed more efficiently than in Section~\ref{subsec:constant_dim}.

\subsubsection{Definition and relationship with labelings} \label{subsubsec:def_mop}

Recall that a pair made up of a vertex $u$ and a POF $L$ outgoing from this vertex can be seen as a hypercube (of basis $u$ and signature $L$, which is unique). The MOPs are defined to highlight certain hypercubes which satisfy a maximality property.

\begin{definition}[Maximal Outgoing POFs]
Pair $(u,L)$ is a MOP if $L$ is outgoing from $u$ and there is no other $L' \supsetneq L$ outgoing from $u$.
\label{def:mop}
\end{definition}

On one hand, we can associate with each MOP $(u,L)$ the unique hypercube with basis $u$ and signature $L$. However, there are some hypercubes such that their pair basis-signature is not a MOP. As a trivial example, consider the square $C_4$ with $\Theta$-classes $E_1,E_2$. The two edges which are incident to $v_0$ are hypercubes of dimension 1, $(v_0,\set{E_1})$ and $(v_0,\set{E_2})$, but are not MOPs since the POF $\set{E_1,E_2}$ is outgoing from $v_0$ and maximal.

On the other hand, there is an interesting relationship between MOPs and maximal POFs. We remind the reader that maximal POFs are in bijection with maximal induced hypercubes (Theorem~\ref{th:maximal_pofs}). Thus, a maximal POF is a MOP if we consider the pair basis-signature of the maximal hypercube representing it. Conversely, MOPs are not necessarily signed with maximal POFs. Let us consider the same trivial example $C_4$: the two edges which are not incident to $v_0$ are MOPs but do not form a maximal hypercube. In brief, MOPs represent some intermediary discrete structure between hypercubes and maximal hypercubes (or maximal POFs).

The execution time of the algorithms (Theorems~\ref{th:simple_ecc} and~\ref{thm:guigui-5}) we designed to determine the eccentricities of median graphs depend on our methods to compute all labels $\varphi(u,L)$ and $\psi(u,R)$, which are both in $\tilde{O}(2^{2d}n)$, as stated in Theorems~\ref{th:compute_phi} and~\ref{th:compute_psi}. Both of them consist in listing all pairs $(L,R)$ of POFs such that $L$ (resp. $R$) is outgoing from (resp. incoming into) vertex $u$ and $L$ is $R$-parallel. We show how the MOPs offer an alternative to this ``brute force'' enumeration. In fact, we can determine all labels by listing only these pairs $(L,R)$ for which $(u,L)$ is a MOP (instead of being only a hypercube).

\begin{theorem}
Assume graph $G$ has at most $\tilde{O}(f(d,n)n)$ MOPs, $f(d,n) = o(2^d)$. There is a combinatorial algorithm computing all labels $\varphi(u,L)$, $\opp_u(L)$, and $\psi(u,R)$ in $\tilde{O}(2^df(d,n)n)$.
\label{th:labels_mops}
\end{theorem}
\begin{proof}
Let $u \in V$, $N_u$ be the number of hypercubes with basis $u$. 
Recall that $\sum_{u \in V} N_u \le 2^dn$. We begin with the definition of a DAG $\hul$ for hypercubes. It is called the \textit{ladder Hasse diagram} of $u$. Its vertex set is made up of all pairs $(u,L)$ of hypercubes with basis $u$, in other words, $L$ is outgoing from $u$. There is an arc $(u,L')\rightarrow (u,L)$ if $L' \subsetneq L$ and $\card{L'} = \card{L} - 1$. All diagrams $\hul$, $u \in V$, can be constructed in time $\tilde{O}(2^dn)$. Indeed, all hypercubes can be enumerated in $\tilde{O}(2^dn)$ with a BFS (Lemma~\ref{le:enum_hypercubes}). Then, it suffices, for each $(u,L)$, to consider the at most $d$ subsets $L' \subsetneq L$ differing from one element from $L$ and connect $(u,L')$ to $(u,L)$. In this way, we obtain a directed graph where its connected components are the diagrams $\hul$. An example of DAG $\hul$ follows.

\begin{figure}[h]
\centering
\begin{subfigure}[b]{0.49\columnwidth}
\centering
\scalebox{0.9}{\begin{tikzpicture}


\node[draw, circle, minimum height=0.2cm, minimum width=0.2cm, fill=red] (P11) at (4,5) {};
\node[draw, circle, minimum height=0.2cm, minimum width=0.2cm, fill=black] (P12) at (4,6.5) {};

\node[draw, circle, minimum height=0.2cm, minimum width=0.2cm, fill=black] (P21) at (6,5) {};
\node[draw, circle, minimum height=0.2cm, minimum width=0.2cm, fill=black] (P22) at (6,6.5) {};

\node[draw, circle, minimum height=0.2cm, minimum width=0.2cm, fill=black] (P31) at (5.0,5.4) {};
\node[draw, circle, minimum height=0.2cm, minimum width=0.2cm, fill=black] (P32) at (5.0,6.9) {};
\node[draw, circle, minimum height=0.2cm, minimum width=0.2cm, fill=black] (P33) at (7.0,5.4) {};
\node[draw, circle, minimum height=0.2cm, minimum width=0.2cm, fill=black] (P34) at (7.0,6.9) {};

\node[draw, circle, minimum height=0.2cm, minimum width=0.2cm, fill=black] (P41) at (2,5) {};
\node[draw, circle, minimum height=0.2cm, minimum width=0.2cm, fill=black] (P42) at (2,6.5) {};

\node[draw, circle, minimum height=0.2cm, minimum width=0.2cm, fill=black] (P51) at (3.0,5.4) {};
\node[draw, circle, minimum height=0.2cm, minimum width=0.2cm, fill=black] (P52) at (3.0,6.9) {};


\draw[->,line width = 1.4pt] (P11) -- (P12);

\draw[->,line width = 1.4pt] (P11) -- (P21);
\draw[->,line width = 1.4pt] (P12) -- (P22);
\draw[->,line width = 1.4pt] (P21) -- (P22);

\draw[->,line width = 1.4pt] (P11) -- (P31);
\draw[->,line width = 1.4pt] (P12) -- (P32);
\draw[->,line width = 1.4pt] (P21) -- (P33);
\draw[->,line width = 1.4pt] (P22) -- (P34);
\draw[->,line width = 1.4pt] (P31) -- (P32);
\draw[->,line width = 1.4pt] (P31) -- (P33);
\draw[->,line width = 1.4pt] (P32) -- (P34);
\draw[->,line width = 1.4pt] (P33) -- (P34);

\draw[->,line width = 1.4pt] (P11) -- (P41);
\draw[->,line width = 1.4pt] (P31) -- (P51);
\draw[->,line width = 1.4pt] (P41) -- (P42);
\draw[->,line width = 1.4pt] (P51) -- (P52);

\draw[->,line width = 1.4pt] (P41) -- (P51);
\draw[->,line width = 1.4pt] (P42) -- (P52);
\draw[->,line width = 1.4pt] (P12) -- (P42);
\draw[->,line width = 1.4pt] (P32) -- (P52);

\node[color = red, scale = 1.4] at (4,4.6) {$u$};

\node[scale = 1.2] at (3.0,4.6) {$E_1$};
\node[scale = 1.2] at (5.0,4.6) {$E_2$};
\node[scale = 1.2] at (1.6,5.7) {$E_3$};
\node[scale = 1.2] at (6.6,4.9) {$E_4$};

\end{tikzpicture}}
\caption{Example of star graph $G_u$}
\label{subfig:two_cubes}
\end{subfigure}
\begin{subfigure}[b]{0.49\columnwidth}
\centering
\scalebox{0.9}{\begin{tikzpicture}

\node (P1) at (1.0,3.2) {$E_1$};
\node (r1) at (1.3,3.2) {};

\node (P2) at (1.0,0.8) {$E_2$};
\node (r2) at (1.3,0.8) {};

\node (P3) at (1.0,2.4) {$E_3$};
\node (r3) at (1.3,2.4) {};

\node (P4) at (1.0,1.6) {$E_4$};
\node (r4) at (1.3,1.6) {};

\node (P13) at (3.5,3.5) {$E_1E_3$};
\node (l13) at (3.1,3.5) {};
\node (r13) at (3.9,3.5) {};

\node (P14) at (3.5,2.75) {$E_1E_4$};
\node (l14) at (3.1,2.75) {};
\node (r14) at (3.9,2.75) {};

\node (P34) at (3.5,2.0) {$E_3E_4$};
\node (l34) at (3.1,2.0) {};
\node (r34) at (3.9,2.0) {};

\node (P23) at (3.5,1.25) {$E_2E_3$};
\node (l23) at (3.1,1.25) {};
\node (r23) at (3.9,1.25) {};

\node (P24) at (3.5,0.5) {$E_2E_4$};
\node (l24) at (3.1,0.5) {};
\node (r24) at (3.9,0.5) {};

\node (P134) at (6.0,2.7) {$E_1E_3E_4$};
\node (l134) at (5.5,2.7) {};

\node (P234) at (6.0,1.3) {$E_2E_3E_4$};
\node (l234) at (5.5,1.3) {};

\draw[->,line width = 1.2pt] (r1) -- (l13);
\draw[->,line width = 1.2pt] (r1) -- (l14);

\draw[->,line width = 1.2pt] (r2) -- (l23);
\draw[->,line width = 1.2pt] (r2) -- (l24);

\draw[->,line width = 1.2pt] (r3) -- (l13);
\draw[->,line width = 1.2pt] (r3) -- (l34);
\draw[->,line width = 1.2pt] (r3) -- (l23);

\draw[->,line width = 1.2pt] (r4) -- (l14);
\draw[->,line width = 1.2pt] (r4) -- (l34);
\draw[->,line width = 1.2pt] (r4) -- (l24);

\draw[->,line width = 1.2pt] (r13) -- (l134);
\draw[->,line width = 1.2pt] (r14) -- (l134);
\draw[->,line width = 1.2pt] (r34) -- (l134);
\draw[->,line width = 1.2pt] (r34) -- (l234);
\draw[->,line width = 1.2pt] (r23) -- (l234);
\draw[->,line width = 1.2pt] (r24) -- (l234);

\end{tikzpicture}}
\caption{DAG $\hul$ representing its hypercubes}
\label{subfig:hul}
\end{subfigure}

\caption{Ladder Hasse diagram $\hul$}
\label{fig:hul}
\end{figure}
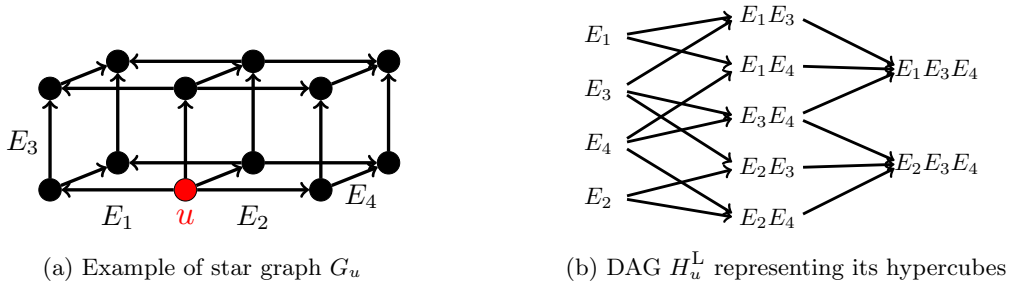

Thanks to the ladder Hasse diagram, computing a list of the MOPs in $\tilde{O}(2^dn)$ is straightforward as they are exactly the leaves of the DAGs $\hul$.

From now on, we divide the proof into three steps, each one correspond to the computation of a certain type of labels. We begin with the labels $\varphi(u,L)$. 

\textit{Computation of $\varphi$-labelings}. 
We present a new procedure to compute the labels $\varphi(u,L)$. A key distinction with the former method should be mentioned: instead of determining labels $\varphi(u,L)$ only, we also compute another type of labels, denoted by $\varphis(u,L)$. These \textit{maximal ladder} labels contain in fact the maximum over all labels $\varphi(u,L')$ such that $L' \subseteq L$. In brief, $\varphis(u,L) = \max\limits_{L' \subseteq L} \varphi(u,L')$.

We compute these two label functions recursively. First, we propose a recursive procedure to obtain $\varphi(u,L)$ based on the enumeration of MOPs. We list all triplets $(L,u^+,L^+)$ such that $(u^+,L^+)$ is a MOP and $L$ is incoming into $u^+$. We denote by $u$ the basis of the hypercube with anti-basis $u^+$ and signature $L$. 
The base case is the same as in Section~\ref{subsubsec:ladder}. If $(u,L)$ is such that either no edge is leaving the anti-basis $u^+$ or all $\Theta$-classes outgoing from $u^+$ are orthogonal to $L$, then fix $\varphi(u,L) = \card{L}$. 
We pursue with the inductive step. We denote by $L_{\perp}^+$ the maximal $L$-parallel subset of $L^+$, {\em i.e.} the set containing exactly the $\Theta$-classes $E_i$ of $L^+$ for which $L \cup \set{E_i}$ is not a POF. As $L_{\perp}^+ \subseteq L^+$, set $L_{\perp}^+$ is a POF. The computation of $L_{\perp}^+$ implies a logarithmic extra cost of $O(d)$ for each MOP $(u^+,L^+)$. We write the MOP-equivalent formula of Equation~\eqref{eq:induction_phi}.

\begin{equation}
     \varphi(u,L) = \max\limits_{\substack{(u^+,L^+) ~\mbox{\scriptsize{MOP,}}~ L ~\mbox{\scriptsize{incoming into}}~ u  \\ \forall E_j \in L_{\perp}^+,~ L \cup \set{E_j} ~\mbox{\scriptsize{not POF}}}} (\card{L} + \varphis(u^+,L_{\perp}^+)).
    \label{eq:induction_phi_mop}
\end{equation}

We explain why Equation~\eqref{eq:induction_phi_mop} is correct. According to Equation~\eqref{eq:induction_phi}, $\varphi(u,L)$ is $\card{L}$ plus the maximum over all $\varphi(u^+,L_*^+)$ - not necessarily MOPs - such that $L_*^+$ is $L$-parallel. Assume $(u^+,L_*^+)$ is not a MOP: there is a MOP $(u^+,L^+)$, where $L_*^+ \subsetneq L^+$. As $L_{\perp}^+$ is the maximal $L$-parallel subset of $L^+$, we have $L_*^+ \subseteq L_{\perp}^+$. So, $\varphi(u^+,L_*^+) \le \varphis(u^+,L_{\perp}^+)$ by definition. Conversely, value $\varphis(u^+,L_{\perp}^+)$ counts in the computation of $\varphi(u,L)$ because all subsets of $L_{\perp}^+$ are $L$-parallel. Therefore, it suffices to consider the MOPs $(u^+,L^+)$ with their maximal subset $L_{\perp}^+$ instead of all hypercubes $(u^+,L_*^+)$.

Second, we explain how the labels $\varphis$ are deduced from the values $\varphi(u,L)$. Assume that for a given vertex $u$, all $\varphi(u,L)$, $L$ outgoing from $u$, have been determined recursively, thanks to the base case or Equation~\eqref{eq:induction_phi_mop}. We deduce all $\varphis(u,L)$ with the ladder Hasse diagram structure $\hul$. We proceed inductively. The base case concerns singleton POFs: $\varphi(u,\set{E_i}) = \varphis(u,\set{E_i})$. Then, we describe the induction step. We modify the value $\varphis(u,L)$ by comparing $\varphi(u,L)$ with all $\varphis(u,L')$, where $L' \subsetneq L$, $\card{L'} = \card{L} -1$. In other words, we initialize $\varphis(u,L)$ as $\varphi(u,L)$ and, for each arc $(u,L')\rightarrow (u,L)$ of $\hul$, we execute $\varphis(u,L) \xleftarrow{\max} \varphis(u,L')$. Concretely, we transfer the $\varphis$-labelings from the roots to the leaves of the diagram $\hul$.

To compute labels $\varphi(u,L)$, the enumeration of MOPs is needed to apply Equation~\eqref{eq:induction_phi_mop}. For each MOP $(u^+,L^+)$, we have to consider all POFs $L$ incoming into $u$, which gives a total of at most $\tilde{O}(2^df(d,n)n)$ triplets $(L,u^+,L^+)$, as the number of MOPs is $\tilde{O}(f(d,n)n)$. The logarithmic extra costs do not increase this runtime. To compute labels $\varphis(u,L)$, our induction is based on the structure of all diagrams $\hul$. The total size (number of arcs) of the DAG $\hul$ is $\tilde{O}(N_u)$ as each element has at most $d$ parents. Therefore, the execution time needed to determine all $\varphis(u,L)$ does not exceed $\tilde{O}(2^dn)$. In summary, the entire procedure to compute $\varphi,\varphis$-labelings is in $\tilde{O}(2^df(d,n)n)$.

\textit{Computation of $\opp$-labelings}. We already know that all labels $\opp_u(L)$ can be determined in time $\tilde{O}(2^dn)$, according to Theorem~\ref{th:compute_opp}.

\textit{Computation of $\psi$-labelings}. Our inductive procedure to determine all labels $\psi(u,R)$ based on the MOP structure is in fact very close to the one produced for $\varphi$-labelings.

We define the \textit{anti-ladder Hasse diagram} $\hual$ of $u$. Its vertex set is made up of all pairs $(u,R)$ of hypercubes with anti-basis $u$ and a POF $R$ incoming into $u$. There is an arc $(u',R')\rightarrow (u,R)$ if both hypercubes (defined by their anti-basis and signature) have the same basis, $R' \supsetneq R$ and $\card{R'} = \card{R} + 1$. As for diagrams $\hul$, all DAGs $\hual$ can be constructed in time $\tilde{O}(2^dn)$ with a standard BFS.

As for $\varphi$-labelings, we define two label functions which will be computed jointly. However, the description is a bit trickier. We compute not only labels $\psi(u,R)$ but also the new ones $\psis(u,R)$. Contrary to $\varphis$-labelings, value $\psis(u,R)$ is not so easy to define. To understand, we remind the inductive process to compute labels $\psi(u,R)$ in Section~\ref{subsubsec:anti_ladder}.

Remember that there are two cases. Let $u^-$ be the basis of $(u,R)$. First, value $\psi(u,R)$ can be given by a distance $d(u,v)$ such that $m(u,v,v_0) = u^-$. In this case, $\psi(u,R) = \card{R} + \varphi(u^-,\opp_{u^-}(R))$. In the new procedure, we initialize all $\psi(u,R)$ with this value. Second, we may have $m(u,v,v_0) \neq u^-$. In this case, value $\psi(u,R)$ is given by the recursive formula in Equation~\eqref{eq:induction_psi}. We can now define labels $\psis$. Let $R^-$ be a POF incoming into $u^-$ and assume that $(u^-,R)$ is a MOP. Let $R_{\perp}$ be the set containing exactly the $\Theta$-classes $E_i$ of $R$ such that $R$ is $R^-$-parallel. We denote by $u_{\perp}$ the anti-basis of the hypercube with basis $u^-$ and signature $R_{\perp}$. Then, value $\card{R} + \psi(u^-,R^-)$, which counts originally in the computation of $\psi(u,R)$ (Equation~\eqref{eq:induction_psi}), will count only for the computation of $\psis(u_{\perp},R_{\perp})$. More formally,

\begin{equation}
    \psis(u_{\perp},R_{\perp}) = \max\limits_{\substack{R^- ~\mbox{\scriptsize{POF incoming to}}~ u^-, \\ \exists (u^-,R) ~\mbox{\scriptsize{MOP with}}~ R_{\perp} ~\mbox{\scriptsize{max subset of}}~ R\\ ~\mbox{\scriptsize{such that}}~ \forall E_j \in R_{\perp}, R^- \cup \set{E_j} ~\mbox{\scriptsize{not POF}}}} (\card{R} + \psi(u^-,R^-))
    \label{eq:induction_psi_mop}
\end{equation}
Observe that certain pairs $(u,R)$ may not admit a value $\psis(u,R)$ according to this definition. In this case, we simply fix $\psis(u,R) = 0$.

Now, we show that value $\psi(u,R)$, in the case $m\neq u^-$, is exactly the maximum over all $\psis(u',R')$ such that both hypercubes (defined by anti-basis and signature) have the same basis and $R \subseteq R'$. On one hand, if $(u^-,R)$ is not a MOP, there is a MOP $(u^-,R^*)$, where $R \subsetneq R^*$. Let $R^-$ be a POF incoming into $u^-$ such that $R$ is $R^-$-parallel. Let $R_{\perp}^*$ be the maximal $R^-$-parallel subset of $R^*$. Set $R_{\perp}^*$ is a POF and $R \subseteq R_{\perp}^*$ by definition. Assume $\psi(u,R) = \card{R} + \psi(u^-,R^-)$: this value is counted in $\psis(u_{\perp}^*,R_{\perp}^*)$ but not in $\psis(u,R)$ if $R\neq R_{\perp}^*$. On the other hand, if some $R' \supseteq R$ is $R^-$-parallel, then $R$ also does. So, we have:

\begin{equation}
    \psi(u,R) = \max \set{\card{R} + \varphi(u^-,\opp_{u^-}(R)),~ \max\limits_{\substack{(u',R') ~\mbox{\scriptsize{same basis as}}~ (u,R) \\ R \subseteq R'}} \psis(u',R')}.
    \label{eq:psi_max}
\end{equation}

Given a basis $u^-$, we compute all nonnegative values $\psis(u,R)$ such that $u^-$ is the basis of the hypercube with anti-basis $u$ and signature $R$. To do so, we enumerate all triplets $(R^-,u^-,R)$ such that $R^-$ is incoming into $u^-$ and $(u^-,R)$ is a MOP. We compute the maximal $R^-$-parallel subset $R_{\perp}$ of $R$ and apply Equation~\eqref{eq:induction_psi_mop}. As for $\varphi$-labelings, it consist in an enumeration scheme in $\tilde{O}(2^df(d,n)n)$.

Then, we deduce labels $\psi(u,R)$ of the hypercubes with basis $u^-$. We use the anti-ladder Hasse diagram $\hual$. If $(u^-,R)$ is a MOP, then $(u,R)$ is a root of $\hual$, and we fix $\psi(u,R)$ as the maximum between its initial value (case $m(u,v,v_0) = u^-$) and $\psis(u,R)$: it corresponds to Equation~\eqref{eq:psi_max} when $R$ is maximal. Otherwise, $\psi(u,R)$ can be computed from Equation~\eqref{eq:psi_max} in function of $\psis(u',R')$ where $(u',R')$ is a parent of $(u,R)$ in $\hual$. The time cost of this step is at most $\tilde{O}(2^dn)$, due to the size of DAG $\hual$. In summary, we obtain the same global running time than for labels $\varphi(u,L)$, which is $\tilde{O}(2^df(d,n)n)$.
\end{proof}

\subsubsection{Cardinality of MOPs} \label{subsubsec:bound_mop}

Our objective is now to express the cardinality of MOPs in function of $n$ and $d$ in order to apply Theorem~\ref{th:labels_mops} and improve the subquadratic execution time established in Theorem~\ref{thm:guigui-5}. To do so, we establish a relationship between MOPs and the subsets of maximal POFs.

\begin{definition}
Let $p$ be the application which, given a vertex $u$ and a POF $L$ outgoing from $u$, returns a pair $(L,L^*)$, where $L^*$ is the POF of $\Theta$-classes incoming into the anti-basis of $(u,L)$.
\end{definition}

If we restrict application $p$ to MOPs, it returns a pair made up of a maximal POF and one of its subsets.

\begin{lemma}[MOPs as subsets of maximal POFs]
Let $(u,L)$ be a MOP and $p(u,L) = (L,L^*)$. Then, $L^*$ is a maximal POF.
\end{lemma}
\begin{proof}
Let $u^+$ be the anti-basis of $(u,L)$. 
Assume, by way of contradiction, that $L^*$ is not maximal. There is a $\Theta$-class $E_h$ such that $L^{**} = L^* \cup \set{E_h}$ is a POF. Let $u'$ be the vertex such that $\mathcal{E}^-(u') = L^{**}$.

For every $E_i \in L^*$, vertices $u^+$ and $u'$ belongs to $\partial H_i''$. The intersection of these boundaries, that we denote by $\partial H^*$, is convex. Therefore, all vertices metrically between $u^+$ and $u'$ belong to $\partial H^*$. For this reason, there is necessarily an edge $e$ incident to $u^+$ whose second endpoint is also in $\partial H^*$. Let $e=(u^+,w)$ be this edge and we denote by $E_j$ its $\Theta$-class. If $e$ is incoming into $u^+$ with the $v_0$-orientation, then we have a contradiction, since $p(u,L) = (L,L^*) \neq (L,L^{**})$. We know that $e$ is outgoing from $u^+$. Recall that $E_j$ is orthogonal to any $\Theta$-class of $L^*$. By successive applications of Lemma~\ref{le:squares}, we show that for any vertex of the hypercube of anti-basis $u^+$ and signature $L^*$, among them $u$, there is an edge of $E_j$ outgoing from it. Here comes the contradiction: $L$ is not a maximal POF outgoing from $u$ as $L \cup \set{E_j}$ is a POF ($L \subseteq L^*$).
\end{proof}

Maximal POFs can be interpreted in the crossing graph $G^{\#}$ (Definition~\ref{def:crossing}). As it describes the orthogonality of $\Theta$-classes, a POF of $G$ is exactly a clique of $G^{\#}$. Naturally, a maximal POF corresponds to a maximal clique of $G^{\#}$. We provide an upper bound of the number of MOPs of $G$ which depends on the maximal cliques of its crossing graph.

\begin{corollary}[MOPs as subsets of maximal cliques in $G^{\#}$]
Let $G^{\#}$ be the crossing graph of $G$ and $\mcalcm^{\#}$ be the set of maximal cliques of $G^{\#}$. The number of MOPs in $G$ is at most $\sum\limits_{C \in \mcalcm^{\#}} 2^{\card{C}}$. 
\label{co:mop_crossing}
\end{corollary}
\begin{proof}
We begin with the proof that application $p$ is injective. Let us consider a pair $(L,L^*) = p(u,L)$. There is a unique vertex $u^+$ such that $\mathcal{E}^- = L^*$. Necessarily, if $(u,L) = p^{-1}(L,L^*)$, then vertex $u$ is the basis of the hypercube with signature $L$ and anti-basis $u^+$, which is unique.

Now, application $p$ is restricted to MOPs only. We know it is injective and that, for any MOP $(u,L)$, set $L^*$ of $p(u,L) = (L,L^*)$ is a maximal POF. Concretely, the number of MOPs is at most the following value: for each maximal POF, add the number of its nonempty subsets (which is $2$ power its cardinality minus 1).

We remind the reader that a POF of $G$ corresponds to an induced clique of its crossing graph $G^{\#}$. Naturally, a maximal clique of $G^{\#}$ represents a maximal POF. It follows that the number of MOPs of $G$ is at most $\sum_{C \in \mcalcm^{\#}} (2^{\card{C}} - 1) \le \sum_{C \in \mcalcm^{\#}} 2^{\card{C}}$.
\end{proof}

We define a parameter, the \textit{maximal clique ratio}. Its role is to provide an upper bound of the ratio between the number of MOPs of a median graph and the number of vertices $n$.

\begin{definition}[Maximal clique ratio]
The maximal clique ratio $r(H)$ of a graph $H$ is the quotient between the sum of the number of subsets of each maximal clique of $H$ by the number of cliques of $H$. Formally,
\[
r(H) = \frac{R\left[ H\right]}{N\left[ H\right]} = \frac{\sum\limits_{C \in \mcalcm(H)} 2^{C}}{\card{\mcalc(H)}}
\]
\end{definition}

Observe that, for $H = G^{\#}$, $N\left[H\right] = \card{V(G)} = n$. The \textit{clique number} of a graph is the size of its maximum clique. The dimension $d$ of $G$ is also the clique number of $G^{\#}$. Complete multipartite graphs (with clique number $d$) are the graphs whose vertex set can be partitioned into independent sets $A_i$, $1\le i\le d$, and any pair of vertices belonging to a different set form an edge.

We begin with the proof that the complete multipartite graphs maximize the ratio $r(H)$. Next, we show that the more the complete multipartite graph is balanced, the largest $r(H)$ is.

The complete multipartite graphs are exactly the graphs fulfilling the following property: for any non-adjacent vertices $u,v$, $N(u)=N(v)$. Let $\trp(H)$ be the number of triplets $(u,v,w)$ of vertices of $H$ such that $uv\notin E$, $uw \notin E$, but $vw \in E$. An equivalent way to characterize complete multipartite graphs is $\trp(H) = 0$. In other words, any graph which is not complete multipartite verify $\trp(H) > 0$. 

\begin{theorem}
Let $H$ be a graph with $\card{V(H)} = q$, clique number at most $d$, which maximizes $r(H)$. If $H$ is not complete multipartite, there is another graph $H'$ with $\card{V(H')} = q$, clique number at most $d$, such that $r(H') = r(H)$ and $\emph{Trp}(H') < \emph{Trp}(H)$. 
\label{th:complete_multi}
\end{theorem}
\begin{proof}
If $H$ is not complete multipartite, there is a pair $u,v$ of vertices such that $(u,v) \notin E$ and $N(u) \neq N(v)$. Let $N_u$ (resp. $N_v$) be the number of cliques of $H$ containing $u$ (resp. $v$). For $x \in \set{u,v}$, we denote by $R_x$ the following value: $R_x = \sum\limits_{\substack{C \in \mcalcm(H) \\ x \in C}} 2^{\card{C}}$. For the remainder, we fix $\Delta R = R_v - R_u$ and $\Delta N = N_v - N_u$.
We define two graphs $\huv$ and $\hvu$, and compare their maximal clique ratio with the initial graph $H$.

Graph $\huv$ is obtained from $H$ by removing $u$ and adding a copy $v'$ of $v$ such that $N(v') = N(v)$. We have $r(\huv) = \frac{R\left[ H\right] + \Delta R}{N\left[ H\right] + \Delta N}$, as $u$ and $v$ cannot belong to the same clique. Conversely, graph $\hvu$ is obtained from $H$ by removing $v$ and adding a copy $u'$ of $u$ such that $N(u') = N(u)$. We have $r(\hvu) = \frac{R\left[ H\right] - \Delta R}{N\left[ H\right] - \Delta N}$. Both $\huv$ and $\hvu$ do not increase the clique number of $H$. One can check that if $\frac{R\left[ H\right] + \Delta R}{N\left[ H\right] + \Delta N} \le \frac{R\left[ H\right]}{N\left[ H\right]}$, then $\frac{R\left[ H\right] - \Delta R}{N\left[ H\right] - \Delta N} \ge \frac{R\left[ H\right]}{N\left[ H\right]}$ and vice-versa. If the inequality is strict, then we have a contradiction since $H$ is supposed to maximize $r(H)$ for graphs with $q$ vertices and clique number at most $d$. The only possibility we have is $r(H) = r(\huv) = r(\hvu)$.

We prove that either $\huv$ or $\hvu$ has less triplets with only two adjacent vertices than $H$. Let $\trp_u$ (resp. $\trp_v$) be the number of triplets of $H$ containing $u$ and not $v$ (resp. $v$ and not $u$). Let $\trp_{u/v}$ (resp. $\trp_{v/u}$) the number of triplets of $H$ containing both $u$ and $v$, where $u$ is the isolated vertex (resp. $v$ is the isolated vertex). We have
\begin{description}
    \item $\trp(\huv) = \trp(H) + \trp_v - \trp_u - \trp_{u/v} - \trp_{v/u}$,
    \item $\trp(\hvu) = \trp(H) + \trp_u - \trp_v - \trp_{u/v} - \trp_{v/u}$.
\end{description}
At least one of this values is smaller than $\trp(H)$, otherwise $\trp_{u/v} = \trp_{v/u} = 0$, which is equivalent to saying $N(u) = N(v)$, a contradiction.
\end{proof}

By successive applications of this result, for any graph $H$ of clique number at most $d$ maximizing $r(H)$, there exists a complete multipartite graph $H'$ with the same ratio and clique number at most $d$. 
Tur\'an graphs $T(q,d)$ are the most balanced complete multipartite graphs with $q$ vertices and clique number $d$. The size of two of its independent sets differ of at most one. Now, the objective is to prove that, among complete multipartite graphs, Tur\'an graphs maximize the maximal clique ratio.

\begin{theorem}
Tur\'an graphs $T(q,d)$ maximize the maximal clique ratio for graphs with $q$ vertices and clique number $d$.
\label{th:turan}
\end{theorem}
\begin{proof}
We consider a complete multipartite graph $H$ with $q$ vertices and clique number $d$. The vertex set of $H$ can be partitioned into $d$ independent sets $A_i$: $V(H) = \bigcup_1^d A_i$. Let $\alpha_i = \card{A_i}$. Assume, w.l.o.g, that the sizes of sets $A_i$ are increasing, {\em i.e.} $\alpha_1 \le \alpha_2 \le \cdots \le \alpha_d$. Recall that $\sum_1^d \alpha_i = q$. Tur\'an graphs are the graphs such that $\alpha_d - \alpha_1 \le 1$.

Assume that, on graph $H$, $\alpha_d - \alpha_1 \ge 2$. We prove adding a vertex to $A_1$ and removing another one from $A_d$ increases the maximal clique ratio, without changing $q$ or $d$. Let $H'$ denote the graph after this transformation. All maximal cliques in $H$ have size $d$, so $R\left[ H\right] = 2^d \prod_1^d \alpha_i$. The number $N\left[ H\right]$ of (not necessarily maximal) cliques of $H$ is expressed as: $N\left[ H\right] = \sum\limits_{\mathcal{J} \subseteq \set{1,\ldots,d}} \prod_{j \in \mathcal{J}} \alpha_j$. Values $R\left[ H'\right]$ and $N\left[ H'\right]$ can be deduced by replacing respectively $\alpha_1$ and $\alpha_d$ by $\alpha_1 + 1$ and $\alpha_d-1$. We assess the quotient between $r(H')$ and $r(H)$. Certain details of our calculations are omitted to keep the paper readable, we restrict ourselves to the main steps of the reasoning. 
\[
\frac{r(H')}{r(H)} = \frac{\frac{R\left[ H'\right]}{R\left[ H\right]}}{\frac{N\left[ H'\right]}{N\left[ H\right]}} = \frac{\frac{\alpha_1+1}{\alpha_1}\frac{\alpha_d-1}{\alpha_d}}{1+\frac{\alpha_d-\alpha_1-1}{(\alpha_1+1)(\alpha_d+1)}} = \left(1+\frac{1}{\alpha_1(\alpha_1+2)}\right)\left(1-\frac{1}{\alpha_d^2}\right) >  1.
\]
Indeed, as $\alpha_d - \alpha_1 \ge 2$, one can check that $\left(1+\frac{1}{\alpha_1(\alpha_1+2)}\right)\left(1-\frac{1}{(\alpha_1+2)^2}\right)$ is greater than $1$ for any value of $\alpha_1$.
\end{proof}

Naturally, we use the maximal clique ratio of Tur\'an graphs to deduce an upper bound for the number of MOPs of any median graph $G$.

\begin{corollary}
The number of MOPs in a median graph is $O(f(d,n)n)$, where $f(d,n) = \left(2.\frac{2^{\frac{\log n}{d}}-1}{2^{\frac{\log n}{d}}}\right)^d$.
\label{co:number_mops}
\end{corollary}
\begin{proof}
According to Corollary~\ref{co:mop_crossing}, the number of MOPs of a median graph $G$ is upper-bounded by $r(G^{\#})n$. Furthermore, Theorem~\ref{th:turan} states that the balanced complete multipartite graphs maximize $r(H)$. Hence,
\begin{equation}
    r(G^{\#})\le \frac{2^d(\frac{q}{d})^d}{\sum_{j=1}^d \binom{d}{j}(\frac{q}{d})^j} \le \frac{2^d(\frac{q}{d})^d}{(1+\frac{q}{d})^d}
    \label{eq:ratio_a}
\end{equation}

The number of cliques of $G^{\#}$ is exactly the number of POFs of $G$. Therefore, $n = \card{V(G)} = (1+\frac{q}{d})^d$. We fix $a = \frac{d}{\log n}$, we have $0 < a \le 1$. We obtain that $2^{\frac{d}{a}} = (1+\frac{q}{d})^d$ and, thus, $\frac{q}{d} = 2^{\frac{1}{a}} - 1$. Finally, we inject this equality into Equation~\eqref{eq:ratio_a}.
\end{proof}

By observing the expression of function $f(d,n)$, we see that the larger the dimension $d$, the smaller the number of MOPs. The following result comes from the win-win approach we announced earlier. When the dimension $d$ is lower than a certain threshold (below $\frac{1}{2}$, the details are in the proof), we can apply the linear FPT algorithm of Theorem~\ref{th:simple_ecc} which gives a subquadratic running time. Otherwise, when $d$ is larger than this threshold, the number of MOPs admits an upper bound less than $2^{2d}n$ and can be enumerated to obtain all eccentricities via the labelings (Theorem~\ref{th:labels_mops}).

\begin{theorem}
There is a combinatorial algorithm determining all eccentricities in $\tilde{O}(n^{1.6456})$.
\label{th:subquadramop}
\end{theorem}
\begin{proof}
Let $a = \frac{d}{\log n}$ and we define two functions: $f(x) = 2.\frac{2^{\frac{1}{x}}-1}{2^{\frac{1}{x}}}$ and $g(x) = 2-\frac{1}{1+\log (2f(x))}$.

One one hand, according to Theorem~\ref{th:simple_ecc}, there is a combinatorial algorithm determining all eccentricities in $\tilde{O}(2^{2d}n) = \tilde{O}(n^{1+2a})$. 

On the other hand, according to Theorem~\ref{th:labels_mops} and Corollary~\ref{co:number_mops}, we can compute all labels $\varphi$, $\opp$ and $\psi$ in time $\tilde{O}(2^d(f(a))^dn)$. Using the reduction scheme of Theorem~\ref{lem:guigui-4}, we can compute all eccentricities in time $\tilde{O}(n^{g(a)})$.

To obtain the best runtime possible, we have to minimize the subquadratic constant $h(a) =\max \set{1+2a,g(a)}$. Function $h$ admits a unique minimum for $0 < a \le 1$, which is reached for a certain $a^*$ we can approximate by $0.3327 \le a^* \le 0.3328$. This gives $h(a^*) \simeq 1.6456$.

We describe the combinatorial algorithm computing all eccentricities in $\tilde{O}(n^{h(a^*)})$. List all hypercubes as stated in Lemma~\ref{le:enum_hypercubes}. If $\frac{d}{\log n} \le a^*$, then apply the linear FPT algorithm evoked in Theorem~\ref{th:simple_ecc}. Otherwise, if $\frac{d}{\log n} > a^*$, then compute the labelings by enumerating the MOPs of $G$ (Theorem~\ref{th:labels_mops}). Deduce from it all eccentricities. Eventually, apply the reduction scheme proposed in Lemma~\ref{lem:guigui-4}.
\end{proof}

\subsection{Faster enumeration of $(L,u)$-parallel POFs} \label{subsec:faster_enum}

We propose in this subsection an alternative for Theorem~\ref{th:naive_paradj}. We remind that the enumeration of $(L,u)$-parallel POFs, for all $(L,u)$ is the key to obtain all labels (see Equation~\eqref{eq:induction_phi}). We design a procedure to decrease the running time needed for this enumeration: using more involved tools and observations, and based on a result of Fomin {\em et al.}~\cite{FoGrPySt05}, we execute this task in $\tilde{O}(3.5394^dn)$, instead of $\tilde{O}(4^dn)$. 

\subsubsection{Preliminaries} \label{subsubsec:prelim}

We begin with some preliminary definitions. Let $E_i$ be some $\Theta$-class.

\begin{definition}
We say a POF $L^+$ is $E_i$-adjacent if there is a vertex $u$ such that an edge of $E_i$ is ingoing into $u$ and $L^+$ is outgoing from $u$.
\label{def:class_adj}
\end{definition}

\begin{definition}
We say a POF $L^+$ is $E_i$-orthogonal if $L^+ \cup \set{E_i}$ is a POF.\footnote{Observe that if $L^+$ is $E_i$-orthogonal, then it is necessarily $E_i$-adjacent. Indeed, as $L^+ \cup \set{E_i}$ is a POF, there is a hypercube with this signature and one of its vertex satisfies Definition~\ref{def:class_adj}.}
\label{def:class_ortho}
\end{definition}

\begin{definition}
We say a POF $L^+$ is $E_i$-aligned if it is $E_i$-adjacent and not $E_i$-orthogonal.
\label{def:class_aligned}
\end{definition}

\begin{figure}[h]
\centering
\begin{subfigure}[b]{0.40\columnwidth}
\centering
\scalebox{0.8}{\begin{tikzpicture}


\node[draw, circle, minimum height=0.2cm, minimum width=0.2cm, fill=black] (P11) at (1,1) {};
\node[draw, circle, minimum height=0.2cm, minimum width=0.2cm, fill=black] (P12) at (1,2.5) {};
\node[draw, circle, minimum height=0.2cm, minimum width=0.2cm, fill=black] (P13) at (3,1) {};
\node[draw, circle, minimum height=0.2cm, minimum width=0.2cm, fill=black] (P14) at (3,2.5) {};

\node[draw, circle, minimum height=0.2cm, minimum width=0.2cm, fill=black] (P21) at (2.0,1.4) {};
\node[draw, circle, minimum height=0.2cm, minimum width=0.2cm, fill=black] (P22) at (2.0,2.9) {};
\node[draw, circle, minimum height=0.2cm, minimum width=0.2cm, fill=black] (P23) at (4.0,1.4) {};
\node[draw, circle, minimum height=0.2cm, minimum width=0.2cm, fill=black] (P24) at (4.0,2.9) {};



\draw[line width = 1.4pt, color = green,->] (P11) -- (P12);
\draw[line width = 1.4pt, color = blue,->] (P11) -- (P13);
\draw[line width = 1.4pt, color = red,->] (P11) -- (P21);
\draw[line width = 1.4pt, color = blue,->] (P12) -- (P14);
\draw[line width = 1.4pt, color = red,->] (P12) -- (P22);
\draw[line width = 1.4pt, color = red,->]  (P13) -- (P23);
\draw[line width = 1.4pt, color = green,->] (P13) -- (P14);
\draw[line width = 1.4pt, color = red,->]  (P14) -- (P24);
\draw[line width = 1.4pt, color = green,->] (P21) -- (P22);
\draw[line width = 1.4pt, color = blue,->] (P21) -- (P23);
\draw[line width = 1.4pt, color = green,->] (P23) -- (P24);
\draw[line width = 1.4pt, color = blue,->] (P22) -- (P24);



\node[scale=1.2, color = red] at (1.4,3.0) {$E_1$};
\node[scale=1.2, color = blue] at (2.0,0.5) {$E_2$};
\node[scale=1.2, color = green] at (0.5,1.75) {$E_3$};

\node[scale = 1.2] at (0.6,2.7) {$u$};

\end{tikzpicture}}
\caption{$L^+ = \set{E_1,E_2}$ is $E_3$-orthogonal as $\set{E_1,E_2,E_3}$ is a POF.}
\label{subfig:ortho}
\end{subfigure}
~
\begin{subfigure}[b]{0.40\columnwidth}
\centering
\scalebox{0.8}{\begin{tikzpicture}


\node[draw, circle, minimum height=0.2cm, minimum width=0.2cm, fill=black] (P13) at (3,1) {};
\node[draw, circle, minimum height=0.2cm, minimum width=0.2cm, fill=black] (P14) at (3,2.5) {};

\node[draw, circle, minimum height=0.2cm, minimum width=0.2cm, fill=black] (P23) at (4.0,1.4) {};
\node[draw, circle, minimum height=0.2cm, minimum width=0.2cm, fill=black] (P24) at (4.0,2.9) {};

\node[draw, circle, minimum height=0.2cm, minimum width=0.2cm, fill=black] (P31) at (3.0,4) {};
\node[draw, circle, minimum height=0.2cm, minimum width=0.2cm, fill=black] (P32) at (5.0,4) {};

\node[draw, circle, minimum height=0.2cm, minimum width=0.2cm, fill=black] (P41) at (5.0,1) {};
\node[draw, circle, minimum height=0.2cm, minimum width=0.2cm, fill=black] (P42) at (5.0,2.5) {};



\draw[line width = 1.4pt, color = green,->]  (P13) -- (P14);
\draw[line width = 1.4pt, color = blue,->]  (P13) -- (P23);
\draw[line width = 1.4pt, color = blue,->]  (P14) -- (P24);
\draw[line width = 1.4pt, color = green,->] (P23) -- (P24);

\draw[line width = 1.4pt, color = red,->] (P14) -- (P31);
\draw[line width = 1.4pt, color = red,->] (P42) -- (P32);
\draw[line width = 1.4pt, color = orange,->] (P13) -- (P41);
\draw[line width = 1.4pt, color = orange,->] (P14) -- (P42);
\draw[line width = 1.4pt, color = orange,->] (P31) -- (P32);
\draw[line width = 1.4pt, color = green,->] (P41) -- (P42);



\node[scale=1.2, color = blue] at (3.4,3.1) {$E_2$};
\node[scale=1.2, color = green] at (2.5,1.75) {$E_3$};
\node[scale=1.2, color = orange] at (4.0,0.5) {$E_4$};
\node[scale=1.2, color = red] at (5.5,3.25) {$E_1$};

\node[scale = 1.2] at (2.6,2.7) {$u$};

\end{tikzpicture}}
\caption{$L^+ = \set{E_1,E_4}$ is $E_3$-aligned while $\set{E_2}$ is $E_3$-orthogonal.}
\label{subfig:aligned}
\end{subfigure}
\caption{Illustration of Definitions~\ref{def:class_adj},~\ref{def:class_ortho} and~\ref{def:class_aligned}}
\label{fig:aligned}
\end{figure}
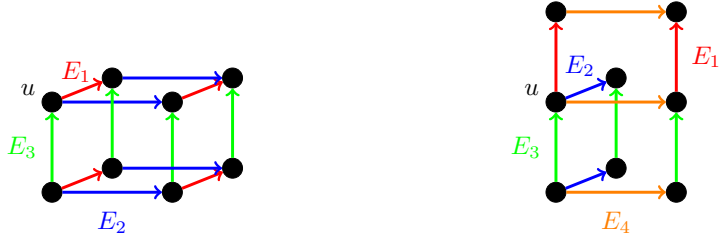

In other words, a POF $L^+$ is $E_i$-\textit{aligned} if it is $E_i$-adjacent and there exists at least one $\Theta$-class $E_j \in L^+$ such that $E_i$ and $E_j$ are parallel, {\em i.e.} not orthogonal. For any POF $L^+$, we split the $\Theta$-classes $E_i$ for which it is adjacent into two categories: the ones that are orthogonal and the ones that are not. Let us begin with the latter. We define $\paradji(L^+)$ as the set of $\Theta$-classes $E_i$ such that $L^+$ is $E_i$-aligned. Its cardinality is at most the dimension $d$.

\begin{lemma}
For any POF $L^+$, the $\Theta$-classes $E_i$ such that $L^+$ is $E_i$-adjacent are pairwise orthogonal. Thus, $\card{\paradji(L^+)}\le d$. Moreover, all sets $\paradji(L^+)$, for all POFs $L^+$, can be enumerated in quasilinear time $\tilde{O}(n)$.
\label{le:paradji}
\end{lemma}
\begin{proof}
We consider the $v_0$-orientation. There is a unique vertex $u^+$ such that the $\Theta$-classes of the edges ingoing into $u^+$ correspond exactly to POF $L^+$: $\mathcal{E}^-(u^+) = L^+$ (Lemma~\ref{le:pof_hypercube}). Moreover, according to the same lemma, this is the closest-to-$v_0$ vertex $v$ for which $L^+ \subseteq \mathcal{E}^-(v)$. Let $u^*$ be the basis of the hypercube defined by signature $L^+$ and anti-basis $u^+$. POF $L^+$ is thus outgoing from $u^*$. We also denote by $\partial H^*$ the set of vertices such that $L^+$ is outgoing from them. We have: (i) $u^* \in \partial H^*$, (ii) $\partial H^*$ is the intersection of all $\partial H_i'$ for each $E_i \in L^+$, so it is convex/gated (the intersection of convex sets is convex). Similarly, we can define $\partial H^+$ as the set of vertices such that $L^+$ is ingoing into them. We have $u^+ \in \partial H^+$ and $\partial H^+$ is convex/gated. As a consequence of Lemma~\ref{le:boundaries}, there is a natural bijection between $\partial H^*$ and $\partial H^+$. Indeed, a vertex $v^*$ of $\partial H^*$ is the basis of a hypercube of signature $L^+$, and this hypercube admits an anti-basis $v^+$ belonging to $\partial H^+$. In brief, the pairs basis/anti-basis of the hypercubes of signature $L^+$ define this bijection. Furthermore, $d(v_0,v^+) = d(v_0,v^*) + \card{L^+}$ because of the edges of the hypercube are oriented towards $v^*$. In particular, $d(v_0,u^+) = d(v_0,u^*) + \card{L^+}$, therefore $u^*$ is the closest-to-$v_0$ vertex of $\partial H^*$. As $\partial H^*$ is gated, $u^*$ is the gate of $v_0$ for this set.

Let $E_i$ a $\Theta$-class such that $L^+$ is $E_i$-adjacent. There is some vertex $u \in \partial H^*$ for which $E_i$ is ingoing into $u$. As $u^*$ is the gate of $v_0$ for $\partial H^*$, there is a shortest $(v_0,u)$-path $P$ passing through $u^*$. The section $P_{(u^*,u)}$ of $P$, from $u^*$ to $u$, contains $\Theta$-classes in $G[\partial H^*]$, which are orthogonal to $L^+$, so $P_{(u^*,u)}$ does not contain an edge of $E_i$. However, an edge of $E_i$ is ingoing into $u$, hence $E_i \in \sigma_{v_0,u}$. Consequently, $P$ contains necessarily an edge of $E_i$, which appears between $v_0$ and $u^*$. As $\partial H_i''$ is convex, $u^* \in \partial H_i''$ and there is an edge of $E_i$ ingoing into $u^*$.

We can now state the following: for each $\Theta$-class $E_i$ for which $L^+$ is $E_i$-aligned, then $E_i \in \mathcal{E}^-(u^*)$. All $\Theta$-classes of $\paradji(L^+)$ are thus pairwise orthogonal: it is a POF and its cardinality is at most $d$.

Conversely, any $\Theta$-class of $\mathcal{E}^-(u^*)$ is such that $L^+$ is $E_i$-aligned. Indeed, assume that $E_j \in \mathcal{E}^-(u^*)$ and $L^+ \cup \set{E_j}$ is a POF. We denote by $(w^*,u^*)$ the edge of $E_j$ incident to $u^*$. According to Lemma~\ref{le:pof_adjacent} applied to POF $L^+ \cup \set{E_j}$ and vertex $u^*$, POF $L^*$ is also outgoing from $w^*$. As $(w^*,u^*)$ is oriented towards $u^*$, then $d(v_0,w^*) < d(v_0,u^*)$. This is a contradiction as $u^*$ is supposed to be the closest-to-$v_0$ vertex of $\partial H^*$.

The enumeration of all sets $\paradji$ can be achieved in the following way. Create a list which, for each vertex $u^+$ of $G$, stores $L^+ = \mathcal{E}^-(u^+)$, the basis $u^*$ of the pair signature/anti-basis $(L^+,u^+)$ and, finally $\mathcal{E}^-(u^*)$. From the observations above, we know that $\paradji(L^+) = \mathcal{E}^-(u^*)$. Moreover, the vertices of $G$ are in bijection with its POFs (Lemma~\ref{le:pof_hypercube}), so no POF will be omitted. We thus obtain all sets $\paradji(L^+)$, for all POFs $L^+$, in quasilinear time $\tilde{O}(n)$.
\end{proof}

Similarly, we denote by $\perpadj(L^+)$ the set of $\Theta$-classes $E_i$ such that $L^+$ is both $E_i$-adjacent and $E_i$-orthogonal. In summary, the set of $\Theta$-classes $E_i$ such that $L^+$ is $E_i$-adjacent can be partitioned into two sets: one which contains the $E_i$-orthogonal classes, $\perpadj(L^+)$, and its complementary $\paradji(L^+)$. We focus on the enumeration of sets $\perpadj(L^+)$ for all POFs $L^+$, and more particularly on the total size of these sets. 
\begin{lemma}
All sets $\perpadj(L^+)$, for all POFs $L^+$, can be enumerated in quasilinear time $\tilde{O}(n)$.
\label{le:enum_perp}
\end{lemma}
\begin{proof}
For each $E_i \in \perpadj(L^+)$, $L^+ \cup \set{E_i}$ is a POF, hence every pair $(E_i,L^+)$ with $E_i \in \perpadj(L^+)$ is in bijection with a pair $(L^+,L)$ with $L,L^+$ POFs, $L^+ \subsetneq L$, $\card{L^+} = \card{L} - 1$. As the cardinality of each POF is at most $d$ and that there are exactly $n$ POFs, we conclude that the total size of sets $\perpadj(L^+)$ is at most $dn$. Consequently, they can be enumerated in quasilinear time $\tilde{O}(n)$, by  simply listing all POFs $L$, listing all their $\Theta$-classes $E_i$ and finally adding $E_i$ into the set $\perpadj(L^+)$ for $L^+ = L\backslash \set{E_i}$.
\end{proof}

\subsubsection{Counting minimal $(L,u)$-parallel POFs}

We pursue with the presentation of a result on the enumeration of \textit{minimal set covers} which produces a key observation on the cardinality of sets $\paradj(L,u)$, for all pairs signature/anti-basis $(L,u)$. We begin with the definition of \textit{minimal} $L$-parallel and $(L,u)$-parallel POFs.

\begin{definition}
We say POF $L^+$ is:
\begin{itemize}
\item minimal $L$-parallel if $L^+$ is $L$-parallel but not $L'$-parallel for every subset $L' \subsetneq L$,
\item minimal $(L,u)$-parallel if $L^+$ is $(L,u)$-parallel but not $(L',u)$-parallel for every subset $L' \subsetneq L$.
\end{itemize} 
\end{definition}

We denote by $\paradjm(L)$ the set of minimal $L$-parallel POFs and $\paradjm(L,u)$ the set of minimal $(L,u)$-parallel POFs: $\paradjm(L) \subseteq \paradj(L)$ and $\paradjm(L,u) \subseteq \paradj(L,u)$. If $L^+$ is $L$-parallel, then it is also $L^*$-parallel for any POF $L^* \supseteq L$ (see Definition~\ref{def:pof_parallel}). Similarly, if $L^+$ is $(L,u)$-parallel, then it is also $(L^*,u)$ for any superset $L^* \supseteq L$ which is ingoing into $u$. Consequently, our approach will consist, in the remainder, in enumerating only the minimal $(L,u)$-parallel POFs (instead of all $(L,u)$-parallel POFs as in Theorem~\ref{th:naive_paradj}). 

There is an interesting property dealing with minimal $L$-parallel POFs and a notion introduced in Section~\ref{subsubsec:prelim}.

\begin{lemma}
If $L^+$ is minimal $L$-parallel (or $(L,u)$-parallel), then $L \subseteq \paradji(L^+)$.
\label{le:aligned_minimal}
\end{lemma}
\begin{proof}
Suppose, by way of contradiction, that $L^+$ is minimal $L$-parallel and there is a $\Theta$-class $E_j \in L$ such that $L^+$ is $E_j$-orthogonal. For any $\Theta$-class $E_i \in L^+$, $L \cup \set{E_i}$ is not a POF, so there exists a $\Theta$-class of $L$ which is parallel to $E_i$. This cannot be $E_j$ which is orthogonal to all $\Theta$-classes of $L^+$. Hence, for any $\Theta$-class $E_i \in L^+$, there exists a $\Theta$-class of $L \backslash \set{E_j}$ parallel to it. As $L \backslash \set{E_j}$ is a POF which is also adjacent to $L^+$, POF $L$ is not minimal, a contradiction.

If $L^+$ is minimal $(L,u)$-parallel, then, for the same reason, $L$ cannot contain a $\Theta$-class $E_j$ orthogonal to $L^+$ otherwise it contradicts the minimality of $L$.
\end{proof}

We will compute the labelings only for the minimal $(L,u)$-parallel POFs and then retrieve all values using Hasse diagrams.
Given a collection of sets $\mathcal{S} = \set{S_1,S_2,S_3,\ldots}$ over a universe $U$, a \textit{minimal set cover} is a sub-collection of $\mathcal{S}$ which not only covers all elements of $U$ but also does not admit a sub-collection which covers $U$. Fomin {\em et al.}~\cite{FoGrPySt05} proved that, for any universe $U$ and collection $\mathcal{S}$, there are at most $1.1175^k$ minimal set covers, where $k$ is a parameter fulfilling $k \le \card{U} + 4.1401\card{\mathcal{S}}$. We show how the enumeration of sets $\paradjm(L,u)$ can be transposed to the set covers. The following lemma refines the trivial upper bound $4^dn$.

\begin{lemma}
Let $L^+$ be some POF.
The total cardinality of sets $\paradjm(L,u)$, for all pairs $(L,u)$, is upper-bounded by $3.5394^dn$. In particular,
\begin{itemize}
\item (A) for any POF $L^+$, there are at most $1.7697^d$ POFs $L$ such that $L^+$ is minimal $L$-parallel.
\item (B) for any pair basis/signature $(u,L^+)$, there are at most $1.7697^d$ POFs $L$ such that $L^+$ is minimal $(L,u)$-parallel.
\end{itemize}
\label{le:number_par_pofs}
\end{lemma}
\begin{proof}
We see the POF $L^+$ as fixed and we consider the following instance for the enumeration of minimal set covers. Let $U = L^+$, $\card{U} \le d$. Then, for each $\Theta$-class $E_j \in \paradji(L^+)$, we define $S_j$ as the set of $\Theta$-classes in $U$ which are parallel to $E_j$. We fix $\mathcal{S} = \set{S_j : E_j \in \paradji(L^+)}$, $\card{\mathcal{S}} \le d$. 

A minimal set cover of this instance thus represents a minimal set $L$ of $\Theta$-classes of $\paradji(L^+)$ such that $L^+$ is $L$-parallel. Indeed, from Lemma~\ref{le:aligned_minimal}, we know that $L$ is a subset of $\paradji(L^+)$. We have $k \le 5.1401d$. Therefore, the number of POFs $L$ such that $L^+$ is minimal $L$-parallel is at most $1.1175^{5.1401d} = 1.7697^d$, which proves (A).

Similarly, if we fix $(u,L^+)$, the POFs $L$ such that $L^+$ is $(L,u)$-parallel correspond to the minimal set covers for the following instance: $U = L^+$, $\mathcal{S} = \set{S_j : E_j \in \mathcal{E}^-(u)}$. Hence, there are at most $1.7697^d$ POFs $L$ such that $L^+$ is minimal $(L,u)$-parallel, {\em i.e.} (B).

We are now ready to prove the main statement of the lemma. The total cardinality of sets $\paradjm(L,u)$ is exactly the number of triplets $(L,u,L^+)$ such that $L^+$ is minimal $(L,u)$-parallel. Pair $(u,L^+)$ can be seen as a pair basis/signature as $L^+$ is outgoing from $u$: there are at most $2^dn$ such pairs (Lemma~\ref{le:number_hypercubes}). 
By considering all pairs $(u,L^+)$ and, for each of them, using (B), we obtain the upper bound $1.7697^d2^dn = 3.5394^dn$.
\end{proof}

\subsubsection{Enumerating minimal $(L,u)$-parallel POFs}
\label{subsubsec:enum_minimal}

Our objective is now to enumerate all triplets $(L,u,L^+)$ such that $L^+ \in \paradjm(L,u)$. We know from Lemma~\ref{le:number_par_pofs} that their cardinality is at most $3.5394^dn$. We show that they can be enumerated in time proportional to their cardinality, if we neglect poly-logarithmic factors, {\em i.e.} in $\tilde{O}(3.5394^dn)$. The algorithm uses the results obtained in the two subsubsections above.

\begin{theorem}
There is a combinatorial algorithm enumerating all sets $\paradjm(L,u)$, for all pairs $(L,u)$, in time $\tilde{O}(3.5394^dn)$.
\label{th:enum_paradjm}
\end{theorem}
\begin{proof}
We begin with the enumeration of sets $\paradjm(L^+)$. Let $L^+$ be a POF. From Lemma~\ref{le:aligned_minimal}, we know that if $L^+$ is minimal $L$-parallel for some $L$, then $L \subseteq \paradji(L^+)$. For any POF $L^+$, we begin with the enumeration of sets $\paradji(L^+)$, which takes quasilinear time (Lemma~\ref{le:paradji}). Next, we list all subsets $L$ of $\paradji(L^+)$ and keep the ones which satisfy that $L^+$ is minimal $L$-parallel. This consists, for each $L$, in looking at all subsets $L' \subsetneq L$, $\card{L'} = \card{L}-1$ (at most $d$) and verifying whether $L^+$ is $L$-parallel but not $L'$-parallel for each $L'$. Hence, the total running time for enumerating all  $\paradjm(L^+)$ is $\tilde{O}(2^dn)$. By Lemma~\ref{le:number_par_pofs}, each set $\paradjm(L^+)$ is of size at most $1.7697^d$.

For each hypercube (pair basis/signature) $(u,L^+)$, we enumerate all sets $L \in \paradjm(L^+)$. As the number of hypercubes is upper-bounded by $2^dn$, there are at most $3.5394^dn$ such triplets. For each triplet $(L,u,L^+)$, we verify whether $L$ enters in $u$ in time $O(d)$. We keep the triplets $(L,u,L^+)$ satisfying this condition and obtain all triplets $(L,u,L^+)$ such that $L^+ \in \paradjm(L,u)$.
\end{proof}

\subsubsection{Retrieving the $\varphi$-labelings}
\label{subsubsec:retrieve_phi}

Our objective is now to use the enumeration of minimal $(L,u)$-parallel POFs in order to compute the $\varphi$-labelings with a better running time. We modify the notation used in Section~\ref{subsubsec:enum_minimal} in order to be in accordance with Section~\ref{subsubsec:ladder} which deals with our first algorithm determining the $\varphi$-labelings.  We focus on 4-uplets $(u,L,u^+,L^+)$, where $u,u^+ \in V(G)$, $L^+$ is a POF outgoing from $u^+$, $L$ is a POF ingoing into $u^+$, and $u$ is the basis of the hypercube defined by signature $L$ and anti-basis $u^+$. We abuse notation and say that:
\begin{itemize}
\item $(u,L,u^+,L^+) \in \paradj$ if $L^+$ is $(L,u^+)$-parallel and $u$ is the basis of the pair $(L,u^+)$,
\item $(u,L,u^+,L^+) \in \paradjm$ if $L^+$ is minimal $(L,u^+)$-parallel and $u$ is the basis of the pair $(L,u^+)$.
\end{itemize}

Our idea consists in applying Equation~\eqref{eq:induction_phi} only on pairs $(u,L,u^+,L^+) \in \paradjm$ (instead of $\paradj$). Such an operation is executed in time $\tilde{O}(3.5394^dn)$, instead of the ``naive'' running time $\tilde{O}(4^dn)$ established in Theorem~\ref{th:compute_phi}.

However, $\varphi(u,L)$ may be equal to $\card{L} + \varphi(u^+,L^+)$ for some 4-uplet $(u,L,u^+,L^+)$ belonging to $\paradj$, but not to $\paradjm$. Assume it is the case. There is a subset $L'$ of $L$ such that $(u',L',u^+,L^+) \in \paradjm$ and $u'$ is the basis of pair $(L',u^+)$. Conversely, if $(u',L',u^+,L^+) \in \paradjm$, then for all supersets $L$ of $L'$, we have $(u,L,u^+,L^+) \in \paradj$. 
As a consequence, $\varphi(u',L') = \card{L'} + \varphi(u^+,L^+)$: indeed, if it was greater, it would mean that $\varphi(u,L) > \card{L} + \varphi(u^+,L^+)$, a contradiction. 

In summary, for any pair $(u,L)$, there is a subset $L' \subseteq L$ such that:
\begin{itemize}
\item  $(u',L',u^+,L^+) \in \paradjm$, where $u'$ is the basis of pair $(L',u^+)$,
\item $\varphi(u,L) - \card{L} = \varphi(u',L') - \card{L'}$.
\end{itemize}
The equality case $L' = L$ occurs when the 4-uplet $(u,L,u^+,L^+)$ belongs to $\paradjm$.

\begin{theorem}
There is a combinatorial algorithm which determines all labels $\varphi(u,L)$ in time $\tilde{O}(3.5394^dn)$.
\label{th:compute_phi_fast}
\end{theorem}
\begin{proof}
Due to the observations above, we can retrieve all $\varphi$-labelings by spreading the intermediary labels obtained only from minimal $(L,u^+)$-parallel POFs. First, we apply Equation~\eqref{eq:induction_phi} restricted to the minimal cases, {\em i.e.} 4-uplets $(u,L,u^+,L^+) \in \paradjm$. We obtain intermediary labelings denoted by $\varphim(u,L)$.

For each vertex $u^+$, we compute the Hasse diagram made up of all POFs ingoing into $u^+$ and the inclusion relationship fixes the arcs of the diagram\footnote{The definition of this Hasse diagram is different from the ladder and anti-ladder Hasse diagrams proposed in Theorem~\ref{th:labels_mops}. Nevertheless, its structure is very similar and it will look like the one proposed in Figure~\ref{subfig:hul}}. Concretely, there is an arc $L' \rightarrow L$ if both $L',L$ are ingoing into $u^+$, $L' \subsetneq L$, and $\card{L'} = \card{L} - 1$. We associate with each node $L$ of this diagram an initial weight which is the value $\varphim(u,L)$, where $u$ is the basis of pair $(L,u^+)$.

We execute a BFS in the diagram. At each node $L$ visited, we assume the labels $\varphi$ associated with its predecessors have been determined. We compare value $\varphim(u,L)-\card{L}$ with the ones of its predecessors, {\em i.e.} values $\varphi(u',L')-\card{L'}$ for each $L' \subsetneq L$, $\card{L'} = \card{L} - 1$. We simply pick up the maximum value and the corresponding POF, say $L_{\max}$. If $L_{\max} = L$, we fix $\varphim(u,L) = \varphi(u,L)$. Otherwise, we fix $\varphi(u,L) = \varphi(u_{\max},L_{\max}) + \card{L} - \card{L_{\max}}$, where $u_{\max}$ is the basis of $(L_{\max},u^+)$.

The size of the diagram (number of nodes/edges) for a given vertex $u^+$ is $\tilde{O}(2^d)$. As a consequence, the time taken to build all diagrams and execute a BFS on each of them is $\tilde{O}(2^dn)$. This is negligible compared to the time needed to enumerate all 4-uplets $(u,L,u^+,L^+) \in \paradjm$, which gives our overall running time.
\end{proof}

\subsubsection{Consequences for the computation of eccentricities}

The enumeration of all sets $\paradjm(L,u)$ (Theorem~\ref{th:enum_paradjm}) can also be used to improve the computation of $\psi$-labelings. Indeed, as $\varphi$, they admit an inductive formula (Equation~\eqref{eq:induction_psi}) based on the enumeration of all 4-uplets in $\paradj$. Therefore, one can also retrieve all $\psi$-labelings by using the methods introduced in Section~\ref{subsubsec:retrieve_phi}: first compute intermediary labels $\psi_m(u,R)$ by considering 4-uplets in $\paradjm$, second retrieve the $\psi$-labelings by spreading the values obtained over supersets.

\begin{theorem}
There is a combinatorial algorithm which determines all labels $\psi(u,R)$ in time $\tilde{O}(3.5394^dn)$.
\label{th:compute_psi_fast}
\end{theorem}

As a conclusion, Theorem~\ref{th:enum_paradjm} offers us the opportunity to determine all eccentricities with a running time $\tilde{O}(3.5394^dn)$.

\begin{corollary}
There is a combinatorial algorithm which determines all the eccentricities of a median graph in time $\tilde{O}(3.5394^dn)$.
\label{co:compute_ecc_fast}
\end{corollary}
\begin{proof}
As shown in the proof of Theorem~\ref{th:simple_ecc}, the eccentricities can be deduced directly from labels $\varphi$, $\opp$, and $\psi$. Theorems~\ref{th:compute_phi_fast},~\ref{th:compute_opp}, and~\ref{th:compute_psi_fast} terminate the proof.
\end{proof}

Remember that the framework of Section~\ref{subsec:reduction} allows us to transform a linear FPT algorithm into a subquadratic-time one. As a direct consequence of Lemma~\ref{lem:guigui-4} applied for $c=3.5394$, we deduce a better running time for the computation of eccentricities on median graphs: $\tilde{O}(n^{1.6458})$.
Observe that the time obtained is very close to the one given by Theorem~\ref{th:subquadramop}, but worse. 

Our intention is to combine the two improvements (MOP structures and minimal parallelism) to obtain the best running time possible.

\begin{theorem}
There is a combinatorial algorithm determining all eccentricities on median graphs in $\tilde{O}(n^{1.6408})$.
\label{th:subquadradj}
\end{theorem}
\begin{proof}
As in the proof of Theorem~\ref{th:subquadramop}, we provide a tradeoff between the algorithms obtained in Corollary~\ref{co:compute_ecc_fast} and Theorem~\ref{th:labels_mops}.
Let $a = \frac{d}{\log n}$ and we define two functions: $f(x) = 2.\frac{2^{\frac{1}{x}}-1}{2^{\frac{1}{x}}}$ and $g(x) = 2-\frac{1}{1+\log (2f(x))}$.

One one hand, according to Corollary~\ref{co:compute_ecc_fast}, there is a combinatorial algorithm determining all eccentricities in $\tilde{O}(n^{1+a\log(3.5394)})$. 
On the other hand, according to Theorem~\ref{th:labels_mops} and Corollary~\ref{co:number_mops}, we can compute all eccentricities in time $\tilde{O}(n^{g(a)})$. Depending on the values of $a$, one can execute the algorithm which offers the best computation time.

The worst case occurs when $a$ reaches some value $0.35140\le a^* \le 0.35141$. For $a=a^*$, the running time obtained is $n^{1+a^*\log(3.5394)} \le n^{1.6408}$.
\end{proof}

\section{Conclusion} \label{sec:conclusion}

As a natural extension of this work, the question of designing a linear-time or quasilinear-time algorithm to compute
the diameter and all eccentricities of median graphs is now open. With the recursive splitting procedure of Lemma~\ref{lem:guigui-4}, unfortunately, the best execution time we could obtain at best is $\tilde{O}(n^{\frac{3}{2}})$. Reaching this bound could represent a first reasonable objective: it would ``suffice'' to propose a FPT combinatorial algorithm which computes all labels in $\tilde{O}(2^dn)$ in order to obtain such time complexity. We see the MOP-approach as a gateway to identify such a procedure.

Another - certainly easier - objective after this work is to adapt the recursive splitting of Lemma~\ref{lem:guigui-4} for reach centralities. We tried to define a weighted version of the reach centralities problem in order to fit them to the halfspace separation, but this task seems to be not so easy. Our hope is to obtain a subquadratic-time algorithm computing reach centralities in median graphs. 

Eventually, we note two lines of research on which this paper could have some influence: (i) the study of efficient algorithms for the computation of other metric parameters on median graphs (perhaps, the \textit{betweeness centrality}~\cite{AbGrWi15}) and (ii) the design of subquadratic-time algorithms for the diameter and all eccentricities on larger families of graphs (\textit{almost-median} or \textit{semi-median} graphs~\cite{Br07,KlSh12} for example). Concerning the betweeness centrality, our intuition is that the labeling framework introduced~\cite{BeHa21} does not suffice to describe the number of $(u,v)$-paths passing through some vertex, which is exactly what betweeness centrality assesses.

\bibliographystyle{plain}
\bibliography{subquadratic_median}

\end{document}